\documentclass[10pt,journal,compsoc]{IEEEtran}

\usepackage{booktabs}
\usepackage{multirow}
\usepackage{xargs} %
\usepackage{mathtools} 
\usepackage{amsmath}
\usepackage{amsfonts}
\usepackage{amsthm}
\usepackage{amssymb}  
\usepackage{xurl}
\usepackage[hidelinks,breaklinks=true]{hyperref}
\usepackage[shortlabels]{enumitem} %

\usepackage{times}
\usepackage{xcolor}

\usepackage{subfigure}

\usepackage[linesnumbered, ruled]{algorithm2e}
\usepackage{algpseudocode}

\usepackage{relsize} 

\usepackage{mwe}
\usepackage{graphicx}

\usepackage{listings}%
\usepackage[switch]{lineno}

\newtheorem{assumption}{Assumption}
\newtheorem{definition}{Definition}
\newtheorem{theorem}{Theorem}
\newtheorem{corollary}{Corollary}

\makeatletter 
\newcommand{\@giventhatstar}[2]{\ensuremath{\left({#1}\;\middle|\;{#2}\right)}} \newcommand{\@giventhatnostar}[3][]{#1(#2\;#1|\;#3#1)} \newcommand{\giventhat}{\@ifstar\@giventhatstar\@giventhatnostar} 
\makeatother

\SetCommentSty{mycommfont}
\SetKwInput{KwInput}{Input}                %
\SetKwInput{KwOutput}{Output}              %

\newcommand{\etal}{\textit{et al.}}

\newcommand{\squishlist}{
	\begin{list}{$\bullet$}
		{
			\setlength{\itemsep}{0pt}
			\setlength{\parsep}{3pt}
			\setlength{\topsep}{3pt}
			\setlength{\partopsep}{0pt}
			\setlength{\leftmargin}{1.5em}
			\setlength{\labelwidth}{1em}
			\setlength{\labelsep}{0.5em} } }

\newcommand{\squishend}{
	\end{list}  }

\ifCLASSOPTIONcompsoc
  \usepackage[nocompress]{cite}
\else
  \usepackage{cite}
\fi

\ifCLASSINFOpdf
\else
\fi

\begin{document}

\title{ 
Preventing Inferences through Data Dependencies on Sensitive Data
}

\author{Primal Pappachan,
        \IEEEmembership{Member, IEEE,}
        Shufan Zhang,
        \IEEEmembership{Student Member, IEEE,}
        Xi He,
        \IEEEmembership{Member, IEEE,}
        and~Sharad Mehrotra,
        \IEEEmembership{Fellow, IEEE}%
\IEEEcompsocitemizethanks{
\IEEEcompsocthanksitem A preliminary version of this article was accepted and presented in VLDB 2022 \cite{Pappachan2022tattletale}.
\IEEEcompsocthanksitem P. Pappachan is with Portland State University, 1825 SW Broadway, Portland, OR 97201.
E-mail: primal@pdx.edu.
\IEEEcompsocthanksitem  S. Zhang and X. He are with the University of Waterloo, Waterloo, Ontario, Canada, N2L 3G1.
E-mail: \{shufan.zhang, xi.he\}@uwaterloo.ca. 
\IEEEcompsocthanksitem S. Mehrotra is with University of California,
Irvine, CA 92697 USA.
E-mail: sharad@ics.uci.edu.

}%
}

\newcommand{\vDatabase}{\ensuremath{\mathbb{D}}}
\newcommand{\vRelationSet}[1]{\ensuremath{\mathcal{R}}_{#1}}
\newcommand{\vRelation}[2]{\ensuremath{\mathbb{R}^{#1}_{#2}}}
\newcommand{\vTupleSet}[1]{\ensuremath{\mathcal{T}}_{#1}}
\newcommand{\vTuple}[2]{\textit{t}^{#1}_{#2}}
\newcommand{\vQuery}[2]{\textit{Q}^{#1}_{#2}}
\newcommand{\vAttribute}[2]{\textit{A}^{#1}_{#2}}
\newcommand{\vDomain}[1]{\textit{Dom}({#1})}
\newcommand{\vCount}[1]{|{#1}|}
\newcommand{\vUserSet}[1]{\ensuremath{\mathcal{U}}_{#1}}
\newcommand{\vUser}[2]{\textit{U}^{#1}_{#2}}

\newcommand{\vCellSet}[1]{\ensuremath{\mathbb{C}}_{#1}}
\newcommand{\vCell}[2]{\textit{c}^{#1}_{#2}}
\newcommand{\vValue}[2]{\textit{value}^{#1}_{#2}}

\newcommand{\view}{\mathbb{V}}
\newcommand{\inference}{\mathbb{I}}

\newcommand{\vDataDepSet}[1]{\ensuremath{S}_{\Delta{#1}}}
\newcommand{\vSchemaDepSet}[1]{\ensuremath{\Delta}_{#1}}
\newcommand{\vSchemaDep}[1]{\ensuremath{\delta}_{#1}}
\newcommand{\vDataDep}[2]{\ensuremath{\tilde{\delta}}^{#1}_{#2}}
\newcommand{\vPredicate}[2]{\textit{Pred}^{#1}_{#2}}
\newcommand{\vPredicates}[1]{\textit{Preds}({#1})}
\newcommand{\vOperator}{\ensuremath{\theta}}
\newcommand{\vConstant}{\textit{const}}
\newcommand{\vCells}[1]{\textit{Cells}({#1})}
\newcommand{\true}{\textit{True}\xspace}
\newcommand{\false}{\textit{False}\xspace}
\newcommand{\nullvalue}{\textit{NULL}\xspace}
\newcommand{\unknown}{\textit{Unknown}\xspace}

\newcommand{\vPolicy}[2]{\textit{P}^{#1}_{#2}}
\newcommand{\vSensValues}[2]{\ensuremath{\psi}^{#1}_{#2}}
\newcommand{\hidden}{\ensuremath{\textit{*}}}

\newcommand{\sz}[1]{{\color{cyan} (SZ: #1)}}
\newcommand{\pp}[1]{{\color{red} (Primal: #1)}}
\newcommand{\xh}[1]{{\color{pink} (XH: #1)}}
\newcommand{\todo}[1]{{\color{purple} (TODO: #1)}}
\newcommand{\revise}[1]{{#1}}
\newcommand{\extend}[1]{{#1}}
\newcommand{\commentrequired}[1]{{#1}}
\newcommand{\reviseA}[1]{{#1}}
\newcommand{\eat}[1]{}
\newcommand{\reviseTKDE}[1]{{#1}}

\newcounter{example}
\newenvironment{example}[1][]{\refstepcounter{example}\par\medskip   \noindent \textbf{Example~\theexample. #1} \rmfamily}{\qedsymbol  \medskip}

\newcommand{\ourapproach}{\textit{Our Approach}\xspace}
\newcommand{\baselineOne}{\textit{Random Hiding}\xspace}
\newcommand{\baselineTwo}{\textit{Oblivious Cueset}\xspace}

\newcommand{\stitle}[1]{\smallskip \noindent{\bf #1}}

\twocolumn

\IEEEtitleabstractindextext{%

\begin{abstract}
Simply
restricting the  computation to  non-sensitive part of the data may lead to inferences on sensitive data through data dependencies.
Prior work on preventing inference control through data dependencies detect and deny queries which may lead to leakage, or only protect against exact reconstruction of the sensitive data. These solutions result in poor utility, and poor security respectively.
In this paper, we present a novel security model called \emph{full deniability}. Under this stronger security model, any information inferred about sensitive data from non-sensitive data is considered as a leakage. 
We describe algorithms for efficiently implementing full deniability on a given database instance with a set of data dependencies and sensitive cells.
Using experiments on two different datasets, we demonstrate that our approach protects against realistic adversaries while hiding only minimal number of additional non-sensitive cells and scales well with database size and sensitive data.
 \end{abstract}

\begin{IEEEkeywords}
Inference Control, Data Dependencies, Inference Protection, Security \& Privacy, Access Control
\end{IEEEkeywords}}

\makeatletter
\long\def\@IEEEtitleabstractindextextbox#1{\parbox{0.922\textwidth}{#1}}
\makeatother

\maketitle

\IEEEdisplaynontitleabstractindextext

\IEEEpeerreviewmaketitle

\IEEEraisesectionheading{\section{Introduction}\label{sect:intro}}

\IEEEPARstart{O}{rganizations} today collect data about individuals that could be used to infer
  their habits, religious affiliations, and health status --- properties that we typically
  consider as sensitive.
New regulations, such as the European General Data Protection Regulation (GDPR)~\cite{gdpr}, the California Online Privacy Protection Act (CalOPPA)~\cite{caloppa}, and the Consumer Privacy Act (CCPA)~\cite{ccpa}, have made  it mandatory for organizations to provide appropriate mechanisms to enable users' control over their data, i.e., (how| why| for how long) their data is collected, stored,  shared, or analyzed. 
\emph{Fine Grained Access Control Policies (FGAC)} supported by databases is an 
integral technology component to implement such user control. 
FGAC policies enable data owners/administrators to specify which data (i.e., 
 tables, columns, rows, and cells ) can/cannot be accessed by  which querier (individuals posing queries on the database) and is, hence, sensitive \cite{ferrari2010accessBook} for that querier. 
 Traditionally, Database Management Systems (DBMS)  implement FGAC by filtering away data that is sensitive for a querier and computing the query on only 
 the non-sensitive part of the data. 
 Such a strategy is implemented using  either 
query rewriting~\cite{agrawal2002hippocratic, pappachan2020sieve} or view-based mechanisms~\cite{rizvi2004extending}. 
It is well recognized that restricting query computation to only non-sensitive data may not prevent the querier from inferring  
sensitive data  based on  semantics inherent in
the data \cite{farkas2002inference, chen2018disclose}.  For instance, the querier may exploit knowledge of data dependencies
 to infer values of sensitive data as illustrated
in the example below.

\begin{example} \label{eg:intro1}
 Consider an Employees table (Figure~\ref{fig:exampleTables}) and an FGAC policy by a user \textit{Bobby} to hide his salary per hour (\textit{SalPerHr}) from all the queries by  other users. 
If the semantics of the data dictates that any two employees who are  faculty should have the same \textit{SalPerHr}, then hiding \textit{SalPerHr} of \textit{Bobby} would not prevent
its leakage from a querier who has access to \textit{Carrie}'s \textit{SalPerHr}. 
\end{example}

In general, leakage may occur  from direct/indirect inferences due  to  different types of data dependencies,  such as conditional functional dependencies (CFD)~\cite{fan2008conditional},
denial constraints~\cite{xuchuDC}, aggregation constraints~\cite{ross1998foundations},  and/or 
\emph{function-based constraints} that exist when  dependent data values are derived/computed using other data values as shown below.

\begin{example} \label{eg:intro2}
Consider the Employee and Wage tables shown   in Table \ref{fig:exampleTables}.
Let \textit{Danny} specify  FGAC policies to hide his \textit{SalPerHour} in Employee Table and \textit{Salary} in Wage Table. 
Suppose there exists a constraint that employees with role  \textit{Staff} cannot have a higher salary per hour than a faculty in the state of California. Using \textit{Bobby}'s salary per hour that is leaked in Example~\ref{eg:intro1}, the new constraint about the staff salary, and the function-based constraint between  that \textit{Salary} and
the fields  function of \textit{WorkHrs} and \textit{SalPerHr}, information about  the salary and the salary per hour of 
\textit{Danny} will be leaked
even though they are sensitive.
\end{example}

\begin{figure*}[ht]
\centering
\includegraphics[width=\textwidth]{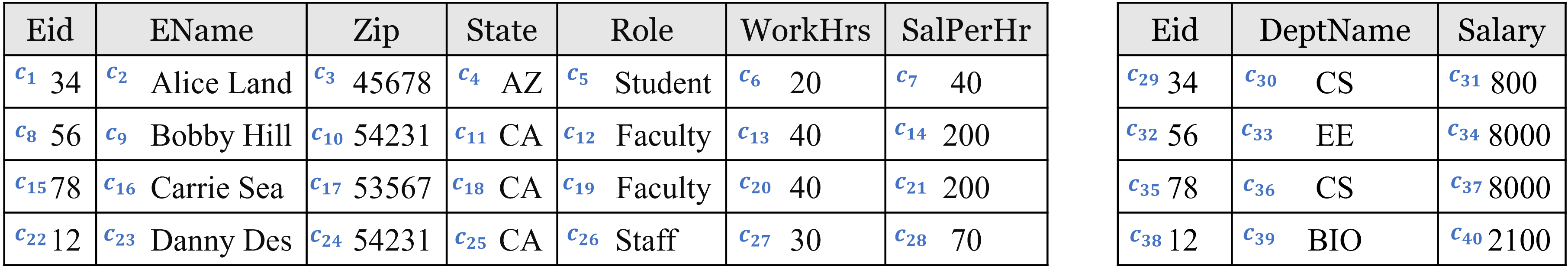}
\caption{Employee and Wages Table}
\label{fig:exampleTables}
\end{figure*}

To gain insight into the extent to which leakage could occur due to data semantics, 
we conducted a simple experiment on a synthetic dataset~\cite{xuchuDC,bohannon2007conditional} that contains the address and tax information of individuals.  The tax data set consists of 14 attributes
and  has associated with it 10 data dependencies,   an  example of which is a  denial constraint  ``if two persons live in the same state, the one earning a lower salary has a lower tax rate''.  
\revise{An adversary can use the above dependency to infer knowledge about the sensitive cells.
Suppose the \textit{salary} attribute of an individual is sensitive and therefore hidden.
If the disclosed data contains information about another individual who lives in the same state and has a lower tax rate, an adversary can infer the upper bound of this individual's salary using the dependency.}
\revise{To demonstrate this leakage,}
we considered an attribute with a large number of data dependencies defined on them (e.g., state) to be sensitive, and thus, replaced its values by \nullvalue.  
We then used state-of-the-art data-cleaning software, Holoclean~\cite{rekatsinas2017holoclean}, as a real-world attacker to reconstruct the \nullvalue values associated with the  sensitive cells.
 Holoclean was able to reconstruct the actual values of the state 100\% of the  time highlighting the importance of preventing leakage through data dependencies on access control protected data.

Prior literature has studied the  challenge of controlling inferences about sensitive data using data dependencies and  called it the ``inference control problem'' \cite{farkas2002inference}. Existing techniques used to protect against inferences can be categorized based on when the leakage prevention is applied \cite{brodsky2000secure}. 
In the first category, inference channels between sensitive and non-sensitive attributes are detected and controlled at the time of database design \cite{DelugachHinkeInf, garvey1993toward}. A database designer uses methods in this category to detect and prevent inferences by upgrading classification of inferred attributes.  However, they result in poor data availability if a significant number of attributes are marked as sensitive to prevent leakages. 
The second category of work includes detection and control at the time of query answering. Works such as \cite{THURAISINGHAM1987479, brodsky2000secure} determine if answers
to the query could result in inferences about sensitive data using data dependencies, and reject the query if such an inference is detected. 
Such query control approaches can lead to the rejection of many queries when there is a non-trivial number of sensitive cells and background knowledge. 
Another limitation of the 
prior work is the weak security model used in determining how to process queries. 
All prior work on inference control considers a query answer to leak sensitive data if the answer 
can be used to reconstruct 
the   exact value of a sensitive object. 
Leakages that
do not reveal the exact value but, perhaps, limit
the values a sensitive object may take are not considered as leakage. For instance, in  Example~\ref{eg:intro2} above, since the constraints do not reveal \textit{Danny}'s exact salary but only that it is below
\$200 per hour, prior works will not consider it to be a leakage even though the querier/adversary could eliminate a significant number of possible domain values based on the data constraints.
As we explain in detail in Section~\ref{sect:related}, the existing solutions to the inference control  cannot be easily   generalized to prevent such leakages.

In this paper, we study 
the problem of answering user queries 
under a new, much stronger
model of security --- viz., 
\emph{full deniability}. 
Under full deniability,  any new knowledge learned about the sensitive cell through data dependencies is considered as leakage. Thus, eliminating a domain value as a possible value an attribute / cell can take violates full-deniability.
One can, of course, 
naively, achieve 
full deniability 
by hiding  the entire database.  Instead,
our goal is to
identify the
minimal  additional non-sensitive cells that must
be hidden so as to achieve
full deniability. 
In addition, 
we require
the algorithm that identifies
data to hide in order to 
achieve full deniability to be efficient and scalable to both large data sets and to a large number of 
constraints.

We study our approach to ensuring full deniability during query processing under two classes of data dependencies\footnote{Other data dependencies such as Join dependencies (JD) and Multivalued dependencies are not common in a clean, normalized database and therefore not interesting to our problem setting.}:
\squishlist
    \item  \emph{Denial Constraints (DCs)}: that are general forms of data dependencies expressed using universally quantified first-order logic. They can express commonly used types of constraints such as functional dependencies (FD) and conditional functional dependencies (CFD) and are more expressive than both
\item  \emph{Function-based Constraints (FCs)}: that establish relationships between input data and
the output data it is derived from, using functions. Such constraints arise
naturally when databases store materialized aggregates or when  data sensor data,  collected
over time (e.g., from sensors),  is enriched (using %
appropriate machine learning tools) to higher level observations. 
\squishend

To achieve full deniability, we first develop a method for 
\emph{Inference Detection}, 
that  detects, for each sensitive cell, the non-sensitive cells
that could result in a violation of full deniability. 
The candidate cells identified by Inference Detection are passed to the second function, \emph{Inference Protection} that minimally selects the non-sensitive data to hide to prevent leakages. 
Our technique  is geared towards maximizing utility when preventing inferences for a large number of sensitive cells and their dependencies.
After hiding additional cells, Inference Detection is invoked repeatedly to detect any indirect leakages on the sensitive cells through the new set of hidden cells and their associated dependencies. These methods are invoked cyclically until no further leakages are detected either on the sensitive cells or any additional cells hidden by Inference Protection.
Using these two different methods, we are able to achieve the security, utility, and performance objectives of our solution.

The main contributions in our paper are:
\squishlist
    \item A security model, entitled \emph{full deniability} to protect against leakage of sensitive data due to data semantics
    in the form of Denial Constraints and Function-based Constraints.
    \item Identification of conditions under which full deniability can be achieved and efficient algorithms for inference detection and protection to achieve full deniability while only minimally hiding additional non-sensitive data.

    \item \reviseTKDE{A relaxed \emph{k-percentile deniability} model, relaxations of security assumptions, and algorithms to achieve these relaxations.}
    
    \item A prototype middleware ($\sim$10K LOC)
    that works alongside DBMS to ensure full deniability given
    a set of dependencies and policies.
    \item Experimental results on two different data sets show that our approach is efficient and only minimally hides non-sensitive
    cells while achieving full deniability. 
\squishend

\noindent
\textbf{Paper Organization.}
We introduce the notations used in the paper and describe access control policies and data dependencies in Section~\ref{sect:preliminaries}. In Section~\ref{sect:problem_definition}, we present the security model --- full deniability --- proposed in this work. In Section~\ref{sect:computation}, we describe how the leakage of sensitive data occurs through dependencies and introduce function-based constraints. We present in Section~\ref{sect:algorithm}, the algorithms for inference detection and protection along with optimizations to improve utility. In Section~\ref{sec:weaker_security_model}, we extend the full deniability model to $k$-percentile deniability and in Section \ref{sec:assumption_relaxation} we relax the security assumptions in our model. In Section~\ref{sect:experiments}, we present results from an end-to-end evaluation of our approach with two different data sets and different baselines. In Section~\ref{sect:related} we go over the related work and we conclude the work by summarizing our contributions, and possible future extensions in Section~\ref{sect:conclusions}.

\vspace{-1em}
\section{Preliminaries}
\label{sect:preliminaries}

Consider a database instance $\vDatabase$ consisting of a set of \textbf{relations} $\vRelationSet{}$. Each relation $\vRelation{}{} \in \vRelationSet{} = \{A_1, A_2, \ldots, A_n\}$ where $\vAttribute{}{j}$ is an attribute in the relation.
Given an attribute $\vAttribute{}{j}$ in a relation $\vRelation{}{}$ we use $\vDomain{\vAttribute{}{j}}$ to denote the domain of the attribute and $\rvert$$\vDomain{\vAttribute{}{j}}$$\rvert$ to denote the number of unique values in the domain (i.e. the domain size)\footnote{We say the domain size in the context of an attribute with discrete domain values and for continuous attributes we discretize their domain values into a number of non-overlapping bins.}. 
A relation contains a number of indexed \textbf{tuples},  $t_i$ represents the $i^{th}$ tuple in the relation $\vRelation{}{}$, and $\vTuple{}{i}[\vAttribute{}{j}]$ refers to the $j^{th}$ attribute of this tuple. 

We will use the \textbf{cell-based} representation of a relation to simplify  notation 
when discussing the fine-grained access control policies and data dependencies. Figure~\ref{fig:exampleTables} shows two tables, the \emph{Employee} table with cells 
$\vCell{}{1}$ to $\vCell{}{28}$ and the
\emph{Wages} table with cells 
$\vCell{}{29}$ to $\vCell{}{40}$. 
Note that in the cell-based notation each table, row, column corresponds to a set of cells. For instance, the second tuple/row of \emph{Wages}
table is the set of cells $\{ \vCell{}{32}, \vCell{}{33}, \vCell{}{34} \}   $ and
the column for attribute \emph{Zip} in the \emph{Employee} table is the set 
$\{ \vCell{}{3}, \vCell{}{10}, \vCell{}{17}, 
\vCell{}{24}\}$. Each cell has an associated value. For instance, the value of cell $\vCell{}{11}$ is ``CA".

\subsection{Access Control Policies}
\label{subsec:background_ac}

Data sharing is controlled using
access control policies, or simply policies.
We classify  users $\vUser{}{}$  
as data owners, who set the access control policies, and as queriers, who pose queries on the data.
Ownership of data is specified
at  tuple level and a data owner of a tuple may specify policies marking
one or more cells ($\vCell{}{i}$) in the tuple $\vTuple{}{}$ as sensitive against queries by other users. 
When another user queries the database, the returned data has to be policy compliant (i.e., policies relevant to the user are applied to the query results). 
We assume queries have associated metadata that contains information about the querier\footnote{In general, policies control
access to data based not just on the identity of the querier, but also on purpose \cite{byun2008purpose}. Thus, metadata associated with the query will also contain purpose in addition to the querier identity.}. 

\reviseTKDE{\stitle{Query model}. The SELECT-FROM-WHERE query posed by a user $\vUser{}{}$ is denoted by $\vQuery{}{}$. In our model, we consider that queries have associated metadata which consists of information about the querier and the context of the query. This way, we assume that for any given query $\vQuery{}{}$, it contains  the metadata such as the identity of the querier (i.e., $\vQuery{querier}{}$) as well as the purpose of the query (i.e., $\vQuery{purpose}{}$). For example, $\vQuery{querier}{}$=``John" and $\vQuery{purpose}{}$=``Analytics".}

\stitle{Policy model}. A policy $P$ is expressed as $<$\textit{$OC$, $SC$, $AC$}$>$, where 
$AC$ corresponds to the action, i.e.,  either deny or allow, $SC$ corresponds to the subject condition i.e, the user to whom the policy applies
(e.g., the identity of the querier, or the group for which the policy applies, in case queriers are organized into groups), and
\reviseTKDE{$OC$ corresponds to a set of object conditions that identifies the cells on which the policy is to be enforced.
Each object condition $OC_i$ is represented using the following 3-tuple: \{$\vRelation{}{}$, $\sigma$, $\Phi$\} where $\vRelation{}{}$ is the relation, $\sigma$ and $\Phi$ are the selection and projection conditions respectively that together select the cells that are sensitive.}
The application of a policy is done by a function over the database that returns \nullvalue for a cell if it is disallowed by the policy or the original cell value if it is allowed. This is modelled after FGAC policy models used in previous works~\cite{pappachan2020sieve, colombo2015efficient}.
We denote the set of cells identified by $OC_i$ as $\vCellSet{OC_i}$.

\begin{definition}[Sensitive Cell]
Given a policy $P=<$\textit{OC, SC, AC}$>$, we say that a cell $\vCell{}{}$ is sensitive to a user $\vUser{}{}$ if $\vCell{}{} \in \vCellSet{OC_i}$ where $OC_i \in OC$, $\vUser{}{} = SC.querier$, and $AC =$ \textit{deny}.
After applying $P$, $\vCell{}{}$ is replaced with \nullvalue.
The set of cells sensitive to the user $\vUser{}{}$ is denoted by $\vCellSet{U}^S$ or simply $\vCellSet{}^S$ when the context is clear.
\label{def:sensCellExt}
\end{definition}

\begin{example}\label{ex:policy}
    An example policy from scenario in Section~\ref{sect:intro} is $<$\textit{\{Employee, EName = ``Carrie Sea'', SalPerHr\}, \{``John Doe'', %
    \}, \{deny\}}$>$. 
  \reviseTKDE{The policy specifies that the salary information (\textit{SalPerHr}) of Employee Carrie (\textit{EName = ``Carrie Sea''}) in the \textit{Employee} table should be denied (i.e., it is sensitive)
  to 
  the \textit{Querier = ``John Doe'' }.} 
\end{example}

\vspace{-1em}

\subsection{Data Dependencies}
\label{subsec:background_dep}
The semantics of data is expressed in the form of \emph{data dependencies}, 
that restrict the set of possible values a cell can take based on the values of other cells in the database. 
Several types of data dependencies have been studied in the literature such as foreign keys, functional dependencies (FDs), and conditional functional dependencies (CFDs), etc. 
We consider two types of dependencies as follows:

\stitle{Denial Constraints (DC)} are first-order formulas of the form
$\forall \, \vTuple{}{i}, \vTuple{}{j}, \ldots \in \vDatabase, \vSchemaDep{}: \, \neg(\vPredicate{}{1} \land \vPredicate{}{2} \land \ldots \land \vPredicate{}{N})$ where $\vPredicate{}{i}$ is the $i$th predicate in the form of \revise{$\vTuple{}{x}[\vAttribute{}{j}] \vOperator \vTuple{}{y}[\vAttribute{}{k}]$ or $\vTuple{}{x}[\vAttribute{}{j}] \vOperator const$} with $x, y \in \{i, j, \ldots \}$, $\vAttribute{}{j}, \vAttribute{}{k} \in R$, $const$ is a constant, and $\vOperator \in \{=, >, <, \neq, \geq, \leq\}$.
DCs are quite general --- they can model dependencies such as  FDs \& CFDs and are flexible enough to model much more complex relationships among cells. 
Data dependencies in the form of DCs have been used 
in recent prior literature  for 
data cleaning \cite{geerts2013llunatic, xu_chu_holistic_2013}, query optimization \cite{kossmann2021data}, and secure databases \cite{vimercati2014outsourcing, brodsky2000secure}. 
Moreover, systems, such as \cite{xuchuDC}, have been
designed to automatically discover DCs in a given database. 
This is the type of DCs considered throughout the paper. \reviseTKDE{We used a data profiling tool, Metanome\cite{metanomeDataProfiling} to identify the complete set of denial constraints. }

\stitle{Function-based Constraints (FCs)} capture the relationships between derived data and its inputs.
As described in Example~\ref{eg:intro2}, the \textit{Salary} in the \textit{Wages} table (see Table~\ref{fig:exampleTables}) is a attribute derived using \textit{WorkHrs} and \textit{SalPerHr} i.e., \textit{Salary} $\coloneqq$ \textit{ fn(WorkHrs}, \textit{SalPerHr}).
In general, given a function $fn$ with $r_1,r_2, \ldots ,r_n$ as the input \revise{values} and $s_i$ as the derived or output \revise{value}, the FC can be represented by $fn(r_1, r_2, \ldots,$ $ r_n) = s_i$.

\section{Full Deniability}
\label{sect:problem_definition}

In this section, we discuss the assumptions in our setting and present the concept of \emph{view} of a database for the querier and formalize an \emph{inference function} with respect to the view and data dependencies.
We formally define  our security model --- \emph{Full Deniability} --- based on the inference function and use it to determine the leakage on sensitive cells. 

\subsection{Assumptions}
\label{sect:assumptions}

We will assume that tuples (and cells in tuples) are independently distributed
except for explicitly specified dependencies that are either learnt automatically or 
specified by the expert.
The database instance is assumed to satisfy the data dependencies. 
The querier, who is the adversary in our setting, is assumed
to know the dependencies and can use them to infer
the sensitive
data values. 
\revise{This assumption leads to a stronger adversary than the standard adversary considered by many algorithms for differential privacy or traditional privacy notions like k-anonymity or access controls, which assumes the adversary knows no tuple correlations (or tuples are independent).} 
A querier is free to run multiple queries and can attempt to make
inferences about sensitive data based on the results of those queries. Two queriers, however, do not collude (i.e., share answers to the queries).
We note that if such collusions were to be allowed, it
would void 
the purpose of having different access control policies for different users.

As queriers are service providers or third parties who are interested in obtaining user data to provide a service and therefore we assume that queriers and data owners do not overlap. 
We also assume that a querier cannot apriori determine if a cell is sensitive or not (i.e., they do not know the access control policies).  
To see why this is important, 
consider a FD defined on the \textit{Employee} table (in Fig.~\ref{fig:exampleTables}) \textit{Zip}$\rightarrow$\textit{State}. Suppose $\vCell{}{11} (State = ``CA")$ is sensitive based on the policy and in order to prevent inferences using the FD, let $\vCell{}{24}$ be hidden. If the querier has knowledge that $\vCell{}{24}$ is hidden due to our approach (and hence know that $\vCell{}{11}$ was sensitive), they can deduce that $\vCell{}{25}$ and $\vCell{}{11}$ have the same value.

\subsection{ Querier View}

For each querier, given the set of policies applicable to the querier, 
the algorithm first determines which cell is sensitive to them.
Such cells are set to \nullvalue  in the view of the database shared with the 
querier.
As noted in the introduction, if only the sensitive cells are set to \nullvalue and all the non-sensitive cells retain their true values, the querier may infer information about the sensitive cells through the various dependencies defined on the database.
It is necessary, therefore, to set some of the non-sensitive cells to \nullvalue in order to prevent leakages due to dependencies.
Henceforth, we will refer to the cells, 
both sensitive and non-sensitive, whose values
will be replaced by \nullvalue as \emph{hidden} cells, denoted by $\vCellSet{}^H$.
We now present the concept of a querier view on top of which queries are answered.

\begin{definition}[Querier View] \label{def:view}
The set of value assignments for a set of cells $\vCellSet{}$ in a database instance $\vDatabase$ with respect to a querier is denoted by $\view(\vCellSet{})$ or simply $\view$ when the set of cells is clear from the context.
The value assignment for a cell could be either the true value of this cell in $\vDatabase$ or \nullvalue value (if it is hidden).
\end{definition}

We also define a concept of 
the \emph{base view} of  database for a querier, denoted by $\view_0$. 
In $\view_0$,  \emph{all} the cells in $\vDatabase$ are set to be \nullvalue. 
We consider the information  leaked to the 
querier based on computing the query results  over the base view $\view_0$ 
as the least amount of information revealed to the querier. For instance, 
the base view may provide querier with information about number of tuples 
in the relation, but, by itself it will not reveal any further information about
the sensitive cells, despite what dependencies hold over the database.
Our goal in developing the algorithm to prevent leakage would be
to determine a view $\view$ for a querier that  hides the minimal number of cells, and yet, leaks no further information 
than  the base view. 
Next, we  define an inference function that captures what the querier can infer about a sensitive cell in a view using dependencies.

\subsection{Inference Function}
\label{sect:inf_fun}

Dependencies such as denial constraints are  defined at schema level, such as the dependency $\vSchemaDep{}$ on Table~\ref{fig:exampleTables}: 
\begin{eqnarray}
\vSchemaDep{}: \forall t_i, t_j \in Emp  \hspace{-0.9em} & \neg(t_i[State] = t_j[State] \land 
 t_i[Role] = t_j[Role]  \nonumber \\ & \land ~~ t_i[SalPerHr] > t_j[SalPerHr]). \nonumber
\end{eqnarray}

Given a database instance $\vDatabase$, the schema level dependencies can be \emph{instantiated} using the tuples.
If the \textit{Employee} Table has 4 tuples, then there are $\binom{4}{2}=6$ number of instantiated dependencies at cell level.
For example, one of the instantiated dependencies for $\vSchemaDep{}$ is 
\begin{eqnarray} \label{eq:id4}
\vDataDep{}{}: \neg((\vCell{}{11} = \vCell{}{18}) \land (\vCell{}{12} = \vCell{}{19}) \land (\vCell{}{14} > \vCell{}{21})) 
\end{eqnarray}
where $\{\vCell{}{11}, \vCell{}{18}, \vCell{}{12}, \vCell{}{19}, \vCell{}{14}, \vCell{}{21}\}$ correspond to $t_2[State]$, $t_3[State]$, $t_2[Role]$, $t_3[Role]$, $t_2[SalPerHr]$, and $t_3[SalPerHr]$ in the \textit{Employee} Table respectively.
From now on, we use $\vDataDepSet{}$ to denote the full set of instantiated dependencies for the database instance $\vDatabase$ at cell level. 
We use $\vPredicates{\vDataDep{}{}}$,  $\vPredicates{\vDataDep{}{},c}$, and $\vPredicates{\vDataDep{}{}\backslash c}$ to represent the set of predicates in the instantiated dependency $\vDataDep{}{}$, the set of predicates in $\vDataDep{}{}$ that involves the cell $c$, and the set of predicates in $\vDataDep{}{}$ that do not involve the cell $c$ respectively.
We also use $Cells(\vDataDep{}{})$ and $Cells(\vPredicate{}{})$ to represent the set of cells in an instantiated dependency and a predicate respectively.
For each instantiated dependency $\vDataDep{}{}\in \vDataDepSet{}$, when every cell $\vCell{}{i}\in Cells(\vDataDep{}{})$ is assigned with a value $x_i\in Dom(\vCell{}{i})$, denoted by 
$\vDataDep{}{}(\ldots,\vCell{}{i}=x_i,\ldots)$, the expression associated with an instantiated dependency can be evaluated to either \true or \false. Note that since the database is assumed to  satisfy all the dependencies, all of the instantiated dependencies must evaluate to \true for any instance of the database.

We use 
the notation $\inference(\vCell{}{} \mid \view)$ to denote the set of values (inferred by the querier) that the cell $\vCell{}{}$ 
can take given
the view  $\view$ but without any knowledge of the set of dependencies.
Likewise, $\inference(\vCellSet{}{} \mid \view )$ denote the cross product of the inferred value sets for cells in the cell set $\vCellSet{}$, i.e., $\inference(\vCellSet{}{} \mid \view)=\times_{\vCell{}{}\in \vCellSet{}} \inference(\vCell{}{} \mid \view)$. Clearly, if in a view, a cell is assigned its original/true value
(and not \nullvalue ) 
then $\inference(\vCell{}{} \mid \view)$ consists of only its true value. We will further assume that:

\begin{assumption}
\label{assumption:correlation}
Let  $\view$ be a view and $\vCell{}{}$  be a cell with value \nullvalue assigned to it in
$\view$. $\inference(\vCell{}{} \mid \view) = \vDomain{\vCell{}{}}$. That is, 
a querier without knowledge of
dependencies, cannot infer any further information about the value of the cell beyond its domain.
\end{assumption}

Knowledge of the dependencies can, however, lead
the querier to make inferences about
the value of the cell. 
The following example illustrates that the querier may  be able to eliminate some domain values as possible assignments of 
 $\vDomain{\vCell{}{}}$.

\begin{example} \label{eg:leak}
Let $\vCell{}{14}$ in Table~\ref{fig:exampleTables} be 
sensitive for a querier and let the view  $\view$ be the same as the original table with  $\vCell{}{14}$
replaced with \nullvalue.
Furthermore, let $\vDataDep{}{}$ (Eqn.~\eqref{eq:id4}) (that refers to
$\vCell{}{14}$) hold.
If the querier is not aware of this dependency $\vDataDep{}{}$, the inferred value set for $\vCell{}{14}$ is the full domain, i.e., $\inference(\vCell{}{14} \mid \view)=Dom(\vCell{}{14})$. 
However, knowledge of
 $\vDataDep{}{}$  leads to the inference that $\vCell{}{14} \leq 200$ since 
the other two predicates ($\vCell{}{11} = \vCell{}{18}, \vCell{}{12} = \vCell{}{19}$) are \true.
\label{eg:cellLeakage}
\end{example}

\begin{definition}[Inference Function]
Given a view $\view$ and an instantiated dependency $\vDataDep{}{}$  for a cell $\vCell{}{i}\in Cells(\vDataDep{}{})$, the inferred set of values for $c_i$ by $\vDataDep{}{}$ is defined as
\begin{eqnarray}
\inference(c_i|\view, \vDataDep{}{}) \hspace{-0.2em}  \coloneqq \hspace{-1.1em} & 
 \hspace{-0.9em} \{x_i  ~|~  \exists (\ldots, x_{i-1}, x_{i+1}, \ldots) \nonumber \\
      &\in \inference(Cells(\vDataDep{}{})\backslash\{c_i\}\mid \view) \nonumber  \\ 
      &s.t.~~\vDataDep{}{}(\ldots, c_i=x_i, \ldots)=\true \}
\end{eqnarray}
where $x_i \in \vDomain{\vCell{}{i}}$.

Given a view $\view$ and a set of instantiated dependencies $\vDataDepSet{}{}=\{\ldots, \vDataDep{}{}, \ldots\}$, the inferred value for a cell $\vCell{}{}$ is the intersection of the inferred values for $\vCell{}{i}$ over all the dependencies, i.e., 
\begin{eqnarray}
    \inference(\vCell{}{i}|\view, \vDataDepSet{}{}) \coloneqq 
    {\bigcap}_{\vDataDep{}{}\in \vDataDepSet{}{}} \inference(\vCell{}{i}|\view, \vDataDep{}{})
\label{eq:inferenceDelta}
\end{eqnarray}

\end{definition}

\subsection{Security Definition}

We can now formally define 
the concept of full deniability of
a view.  Note that given a view $\view$
and a set of dependencies  $\vDataDepSet{}$, the following always holds:
$\inference(\vCell{}{}|\view, \vDataDepSet{}{})\subseteq \inference(\vCell{}{}|\view_0, \vDataDepSet{}{})$. We say that
a $\view$ achieves full deniability if
the two set are identical i.e.,
the query results does not enable the  querier to infer anything further about
the database than what the querier could infer from the $\view_0$ (which, 
as mentioned in Sec. 3.2, is the least 
amount of information leaked to the querier).

\begin{definition}[Full Deniability] 
Given a set of sensitive cells 
$\vCellSet{}^S$ in a database instance $\vDatabase$ and a set of instantiated dependencies $\vDataDepSet{}$, we say that a querier view $\view$ achieves full deniability if for all $\vCell{*}{}\in \vCellSet{}^S$, 
\begin{eqnarray}
\inference(\vCell{*}{}|\view, \vDataDepSet{}{})= \inference(\vCell{*}{}|\view_0, \vDataDepSet{}{}). 
\end{eqnarray}
\end{definition}

\section{Full Deniability with Data Dependencies}
\label{sect:computation}

In this section, we first identify conditions under which denial constraints could result 
in leakage
of sensitive cells (i.e., violation of full deniability) and further consider leakages due to function-based
constraints (discussed in Section~\ref{sect:preliminaries}).

\subsection{Leakage due to Denial Constraints}
\label{subsec:leakage_DC}

An instantiated denial constraint consists of multiple predicates in the form of  $\vDataDep{}{} = \neg (\vPredicate{}{1} \land  \ldots \land \vPredicate{}{N})$ where each predicate is either $\vPredicate{}{N} = \vCell{}{} \; \vOperator \; \vCell{\prime}{}$ or $\vPredicate{}{N} = \vCell{}{} \; \vOperator \; \vConstant$. 
A valid value assignment for cells in $\vCellSet{}(\vDataDep{}{})$ has at least one of the predicates in $\vDataDep{}{}$ evaluating to \textit{False} so that the entire dependency instantiation $\vDataDep{}{}$ evaluates to \textit{True}. 
Based on this observation, we identify a sufficient condition to prevent a querier from learning about a sensitive cell $\vCell{*}{} \in \vCellSet{}^{S}$ in an instantiated DC  $\vDataDep{}{i}$ with value assignments.

As shown in Example~\ref{eg:leak}, for an instantiated DC $\vDataDep{}{}$ with cell value assignments, when all the predicates except for the predicate containing the sensitive cell ($\vPredicate{}{}(\vDataDep{}{} \backslash \vCell{*}{})$) evaluates to \textit{True}, a querier can learn that the remaining predicate $\vPredicate{}{}(\vDataDep{}{}, \vCell{*}{})$ evaluates to \textit{False} even though $\vCell{*}{}$ is hidden.
Thus, it becomes possible for the querier to learn about the value of a sensitive cell from the other non-sensitive cell values.
We can prevent such an inference by hiding additional non-sensitive cells.

\begin{example}
Suppose, in  Example~\ref{eg:leak},  we hide the non-sensitive cell (e.g., $\vCell{}{11}$) 
in addition to $\vCell{}{14}$ (i.e., replace it  with \nullvalue).  \reviseTKDE{Now, the querier will be uncertain of the truth value of $\vCell{}{11} = \vCell{}{18}$}, and as a result, 
cannot determine the truth 
value of the predicate $\vCell{}{14} > \vCell{}{21}$ containing the sensitive cell. Since the predicate, $\vCell{}{14} > \vCell{}{21}$ could either be true or false, 
the querier does not learn anything about
the value of the sensitive cell $\vCell{}{14}$.
\end{example}

We can formalize this intuition into
a sufficient condition that identifies additional
non-sensitive cells to hide 
which we refer to as the \emph{Tattle-Tale Condition} (TTC) \footnote{Tattle-Tale refers to someone who reveals secret about others} 
in order to prevent leakage of sensitive cells, as follow: 

\begin{definition}[Tattle-Tale Condition]\label{def:ttc}
Given an instantiated DC $\vDataDep{}{}$, a view $\view{}{}$, a cell $\vCell{}{} \in Cells(\vDataDep{}{})$, and $\vPredicates{\vDataDep{}{}\backslash \vCell{}{}} \neq \phi$  
\begin{equation}
TTC(\vDataDep{}{}, \view, \vCell{}{})=
    \begin{cases}
        \textit{True}, & \parbox{4cm}{$\forall ~~ \vPredicate{}{i}$ $\in$ $\vPredicates{\vDataDep{}{} \backslash \vCell{}{}}$, \\ $eval(\vPredicate{}{i}, \view{}{})$ = \textit{True}} \\
        \textit{False}, & \textit{otherwise}
    \end{cases}
\label{eq:ttc}
\end{equation}
where \textit{eval($\vPredicate{}{}, \view{}{}$)}  refers to the truth value of
the predicate $\vPredicate{}{}$ in the view $\view{}{} $ using the standard 3-valued
logic of SQL i.e., a predicate evaluates
to true, false, or unknown (if one or both cells are set to \nullvalue). The predicates only compare between the values of two cells or the value of a cell with a constant.
\vspace{-0.3em}
\end{definition}

Note that $TTC(\vDataDep{}{}, \view, c)$ is \true 
if and only if all the predicates except for the predicate\revise{(s)} containing $\vCell{}{}$ ($\vPredicates{\vDataDep{}{},c}$) evaluate to \textit{True} in which case, the querier can
infer that the \revise{one of the predicates} containing $\vCell{}{}$ must be false and, as a result, 
could exploit the knowledge of the predicate\revise{(s)}
to restrict the set of possible
values that $\vCell{}{}$ could take. 
This leads us to a sufficient condition
to achieve full deniability as captured in the following two theorems. In proving the theorems, we will assume that none of the 
predicates in the denial constraints are trivial 
That is, there always exist a domain value for which the predicate can be true or false. This also means that
in the base view $\view_0$ (where all cells are hidden), for any cell $\vCell{}{i} \in cells(\vDataDep{}{})$ and for any predicate $\vPredicate{}{} \in \vPredicates{\vDataDep{}{}, \vCell{}{}}$, there exists a possible assignment for $\vCell{}{i}$ in $\inference(\vCell{}{i}\mid\view_0{}{},\vDataDep{}{})$ such that \textit{eval}($\vPredicate{}{}, \view_0{}{})$ returns \textit{False}. The proof is inclduded in the appendix.

\begin{theorem}
Given an instantiated DC $\vDataDep{}{}$, a view $\view{}{}$, and a sensitive cell $\vCell{*}{} \in \vCells{(\vDataDep{}{})}$ whose value is hidden in this view. If the Tattle-Tale Condition $TTC(\vDataDep{}{}, \view, \vCell{*}{})$ evaluates to False, then the set of inferred values for $\vCell{*}{}$ from $\view$ is the same as that from the base view $\view_0$ (where all the cells are hidden), i.e.,  $\inference(\vCell{*}{}|\view, \vDataDep{}{})=\inference(\vCell{*}{}|\view_0, \vDataDep{}{})$. 
\label{theorem:TTC}
\end{theorem}

\begin{corollary}
\label{cor:TTC}
Given a set of instantiated DCs $\vDataDepSet{}$, a view $\view{}{}$, and a sensitive cell $\vCell{*}{}$ whose value is hidden in this view. If for each of the instantiations $\vDataDep{}{i} \in \vDataDepSet{}$, \textit{TTC}($\vDataDep{}{i}, \view{}{}, \vCell{*}{}$) evaluates to \textit{False} then the set of inferred values $\vCell{*}{}$ from the $\view{}{}$ is same as that from the base view $\view{}{}_0$ i.e., $\inference(\vCell{*}{} \mid \view{}{}, \vDataDepSet{}) = \inference(\vCell{*}{} \mid \view_0, \vDataDepSet{})$.
\end{corollary}

\extend{
\begin{proof}
From Theorem~\ref{theorem:TTC}, we know that for each $\vDataDep{}{i} \in \vDataDepSet{}$ when the TTC is False, we have $\inference(\vCell{*}{}|\view, \vDataDep{}{i})=\inference(\vCell{*}{}|\view_0, \vDataDep{}{i})$.
As each individual set based on individual dependency instantiation are equal in both the released view and base view, the joint set of values in both views computed by the intersection of all the sets should also be equal i.e., ${\bigcap}_{\vDataDep{}{i} \in \vDataDepSet{}{}} \inference(\vCell{*}{}|\view, \vDataDep{}{i})$ = ${\bigcap}_{\vDataDep{}{i} \in \vDataDepSet{}{}} \inference(\vCell{*}{}|\view_0, \vDataDep{}{i})$.
According to Equation~\ref{eq:inferenceDelta}, this joint set is the final inferred set of values for $\vCell{*}{}$ based on $\vDataDepSet{}$ in a given view and as they are equal we have $\inference(\vCell{*}{} \mid \view{}{}, \vDataDepSet{}) = \inference(\vCell{*}{} \mid  \view_0, \vDataDepSet{})$. 
\end{proof}
If the dependency $\vDataDep{}{}$ only contains a single predicate, the Tattle-Tale condition evaluates to True even in $\view_0$ when all the cells are hidden $TTC(\vDataDep{}{}, \view_0, \vCell{}{i})=$\textit{True} in the cases of $\vPredicate{}{}(\vCell{}{i})$ and therefore it is not possible to prevent querier from learning about the truth value of the sensitive predicate. 
}

\subsection{Selecting Cells to Hide}

As shown in Theorem \ref{theorem:TTC}, the Tattle-Tale condition evaluating to \false is the sufficient condition of achieving full deniability requirement.
 $TTC(\vDataDep{}{}, \view{}{}, \vCell{}{})$ evaluates to \false when one of the following holds:
(i) none of the predicates involve the sensitive cell i.e., $\vPredicates{\vDataDep{}{}, \vCell{*}{}} = \phi$ (trivial case); 
(ii) one of the other predicates in $\vPredicates{\vDataDep{}{}\backslash\vCell{*}{}}$ evaluates to \false in $\view{}{}$; or
(iii) one of the other predicates in $\vPredicates{\vDataDep{}{}\backslash\vCell{*}{}}$ involve a hidden cell in $\view{}{}$ and thus evaluates to \unknown.

We define \emph{cuesets}\footnote{These cells give a \textit{cue} about the sensitive cell to the querier.} as the set of cells in an instantiated DC that can be hidden to falsify the Tattle-Tale condition.

\begin{definition}[Cueset]
Given an instantiated DC $\vDataDep{}{}$, a cueset for a cell $c\in cells(\vDataDep{}{})$ is defined as 
\begin{equation}
    cueset(\vCell{}{}, \vDataDep{}{}) 
= Cells(\vPredicates{\vDataDep{}{}\backslash c})
.
\end{equation}
\end{definition}

\revise{
If $\vDataDep{}{}$ only contains a single predicate, we consider the remaining cell in the $cueset(\vCell{}{}, \vDataDep{}{}) = \vCell{}{j}$ given that $\vPredicate{}{}(\vCell{}{}) = \vCell{}{i} \vOperator \vCell{}{j}$.
}

\begin{example}\label{eg:cueset}
In the instantiated DC from Example~\ref{eg:leak}, 
the cueset for $\vCell{}{14}$ based on $\vDataDep{}{4}$ is $cueset(\vCell{}{14}, \vDataDep{}{4})$ = $\{\vCell{}{4}, \vCell{}{11}, \vCell{}{5}, \vCell{}{12}\}$.
\end{example}

We could falsify the Tattle-Tale condition w.r.t. 
a given cell  $\vCell{}{}$ and dependency $\vDataDep{}{}$ by hiding 
any one of the cells in the cueset independent of their values in $\view{}{}$.
The cuesets for a cell $\vCell{}{}$ is defined for a
given dependency instantiation. We can further define 
cueset  for $\vCell{}{}$ for given a set of instantiated DCs $\vDataDepSet{}$ by simply computing the $cueset(\vCell{}{}, \vDataDep{}{})$ for each instantiated dependency in the set $\vDataDep{}{} \in \vDataDepSet{}$.
In order to prevent leakage of $\vCell{}{}$ through $\vDataDep{}{}$, we will  hide
one of the cells in the  $cueset(\vCell{}{}, \vDataDep{}{})$ corresponding to
each of dependency instantiations $\vDataDep{}{} \in \vDataDepSet{}$.

\eat
{\color{blue}
We might define two versions of cueset:
\begin{itemize}
\item  an instance/view oblivious cueset: 
\begin{equation}
    cueset(\vCell{}{}, \vDataDep{}{}) 
= cells(\vPredicates{
\vDataDep{}{})} - \vPredicates{\vDataDep{}{},c})
\end{equation}
where $cells(\vPredicates{\cdots})$ denotes the cells in the set of predicates. 
This view oblivious cueset does not care whether the predicates have assigned values and can be evaluated False or not. If designing algorithms with respect to this type of cueset, it is fine to let the querier to know what is the set of sensitive cells.
\item 
a view aware cueset: 
if  a view $\view$ falsifies some of the predicates in $\vPredicates{\vDataDep{}{}}\backslash \vPredicates{\vDataDep{}{},c}$ (i.e.no cells in this predicate are hidden in $\view$ and the predicate evaluates False), 
then there is no cueset for $c$ and $\vDataDep{}{}$; 
otherwise, the cueset is defined in the same way as the view oblivious cueset. The cueset of $c$ given $\vDataDep{}{}$ and $\view$ is 
\begin{eqnarray}
cueset(\vCell{}{}, \vDataDep{}{},\view)  =  \nonumber \\ 
\begin{cases}
    N.A., & \text{if } \exists \vPredicate{}{} \in \vPredicates{\vDataDep{}{}}\backslash \vPredicates{\vDataDep{}{},c} \text{ s.t. } eval(\vPredicate{}{}, \view)=False.\\
    cueset(\vCell{}{}, \vDataDep{}{}), & \text{otherwise}.
  \end{cases}
\end{eqnarray}
\end{itemize}
}

This alone, however, might not still falsify the 
tattle-tale condition to achieve full-deniability. 
Leakage can  occur indirectly since the value of the
cell, say $\vCell{}{j}$ chosen from the   
  $cueset(\vCell{*}{}, \vDataDep{}{i})$ to hide 
(in order to protect leakage of a sensitive cell $\vCell{*}{}$) 
could, in turn, be 
inferred due to additional dependency instantiation, say $\vDataDep{}{j}$. If this dependency instantiation does contain 
 $\vCell{*}{}$ (as in that case $\vCell{*}{}$ is already  hidden and therefore it cannot be used to infer any information about $\vCell{}{j}$), such a leakage can, in turn, lead to leakage of $\vCell{*}{}$ as shown in the following example.

Achieving full deniability for the sensitive cells requires us to recursively
select cells to hide from the cuesets of not just sensitive cells, but also, from the cuesets of all the hidden cells.  
This recursive hiding of cells  
terminates when the cueset of a newly hidden cell includes an already hidden cell.
The following theorem states that after the recursive hiding of cells in cuesets has terminated, the querier view achieves full deniability. The proof is included in the appendix.

\begin{theorem}[Full Deniability for a Querier View]\label{theorem:fdcell}
Let $\vDataDepSet{}$ be the set of dependencies,
$\vCellSet{}^S$ be the sensitive cells for the querier  and $\vCellSet{}^S \subseteq \vCellSet{}^H$ be the set of hidden cells resulting in a  $\view{}{}$ for the querier.
$\view{}{}$ achieves full deniability 
if\; $\forall \vCell{}{i} \in \vCellSet{}^H$,
$\forall \vDataDep{}{} \in \vDataDepSet{}$,
$\forall$ $\text{non-empty  }$ $cueset(\vCell{}{i}, \vDataDep{}{})$ $\in cuesets(\vCell{}{i}, \vDataDepSet{})$, 
there exists a $\vCell{}{j} \in \vCellSet{}^H$ such that $\vCell{}{j} \in cueset(\vCell{}{i}, \vDataDep{}{})$.
\end{theorem}

\subsection{Leakage due to Function-based Constraints} 
\label{subsec:leakage_FC}

To study the leakages due to function-based constraints (FCs), we define the property of invertibility associated with functions.

\begin{definition}[Invertibility] 
Given a function $fn(r_1, r_2, \ldots,$ $ r_n) = s_i$, we say that $fn$ is invertible  if it is possible to infer knowledge about the inputs ($r_1, r_2, \ldots, r_n$) from its output $s_i$. Conversely, if  $s_i$ does not lead to any inferences about ($r_1, r_2, \ldots, r_n$), we say that it is non-invertible 
\end{definition}

The \textit{Salary} function, in Example~\ref{eg:intro2}, is invertible as 
given the Salary of an employee, a querier can determine the minimum value of 
\textit{SalPerHr} for that employee given that the maximum number of work hours in a week 
is fixed. Complex user-defined functions (UDFs) (e.g., sentiment analysis code which outputs the sentiment of a person in a picture), oblivious functions, secret sharing, and many aggregation functions are, however, non-invertible.
Instantiated FCs can be represented similar to denial constraints. 
 For example, 
an instantiation of the dependency $\vSchemaDep{}:$ \textit{Salary} $\coloneqq$ \textit{fn(WorkHrs, SalPerHr)} is: \reviseTKDE{$\vDataDep{}{}: \neg(\vCell{}{6} = 20 \land \vCell{}{7} = 40 \land \vCell{}{31} \neq 800)$ where $\vCell{}{6}, \vCell{}{7}, \vCell{}{31}$} corresponds to Alice's \textit{WorkHrs}, \textit{SalPerHr} and \textit{Salary} respectively.

For instantiated FCs, if the sensitive cell corresponds to an input to the function, and the function is
not invertible, then leakage cannot occur due to such an FC. Thus, the $TTC(\vCell{*}{}, \vDataDep{}{}, \view{}{})$ returns \false when the function is non-invertible.
For all other cases, the leakage can 
occur in the exact same way as in denial constraints. 
We thus, need to to ensure the Tattle-Tale Condition for all the instantiations of a FC evaluates False.

\stitle{Cueset for Function-based Constraints.}
The cueset for a FC $\vDataDep{}{}$ is determined depending on whether the derived value ($s_i$)
or input value ($\{ \ldots, r_j, \ldots \}$) is sensitive and the invertibility property of the function $fn$. 
\begin{equation*}
cueset(\vCell{}{}, \vDataDep{}{})=
    \begin{cases}
        {\{ \vCell{}{i} \} \; \forall \vCell{}{i} \in \{ \ldots, r_j, \ldots \}}, & {\text{if $\vCell{}{} = s_i$}} \\
        { \{ s_i \} } & \hspace{-9.5em}  {\text{$fn$ \textit{ is invertible and if } $\vCell{}{} \in \{ \ldots, r_j, \ldots \} $}} \\
        { \text{$\phi$} } & \hspace{-10.5em} {\text{$fn$ \textit{is non-invertible and if} $\vCell{}{} \in \{ \ldots, r_j, \ldots \} $}}
    \end{cases}
\label{eq:fc_cuesets}
\end{equation*}

As the instantiation for FC is in DC form and their Tattle-Tale Conditions and cueset determination are almost identical, in the following section we explain the algorithms for achieving full deniability with DCs as extending it to handle FCs requires only a minor change (disregard cuesets when one of the input cell(s) is sensitive and function is non-invertible).

\reviseTKDE{
\stitle{Remark.}
We extend the invertibility notion to a more general model, i.e., \emph{$(m, n)$-invertibility}, that can capture the partial leakage due to function-based constraints. The details for this notion and computing partial leakage according to $(m,n)$-invertibility can be found in supplementary materials.
}

\section{Algorithm to Achieve Full Deniability}
\label{sect:algorithm}

In this section, we present an algorithm 
to determine the set of cells to hide to achieve
full-deniability based on  Theorem~\ref{theorem:fdcell}. 
Full-deniability can trivially be achieved by sharing the 
base view $\view_0$ where all cell values are replaced with $\nullvalue$.
Our goal is to ensure that we hide the minimal number of cells possible while achieving full deniability.

\subsection{Full-Deniability Algorithm}

\begin{algorithm}[t]
    \SetAlgoLined
    \DontPrintSemicolon
    \KwInput{User $\vUser{}{}$, Data dependencies $\vDataDepSet{}$, A view of the database $\view$}
    \KwOutput{A secure view $\view_S$}
        $\vCellSet{}^{S}$ = \textbf{SensitivityDetermination}($\vUser{}{}$, $\view$) \\
        $\vCellSet{}^H$ = $\vCellSet{}^{S}$,  $\view_S = \view$  \\
        $cuesets$ = \textbf{InferenceDetect}($\vCellSet{}^H$, $\vDataDepSet{}, \view$) \\
        \While{$cuesets \neq \phi$}
        {
            \For{\textit{cs} $\in$ \textit{cuesets}}
            {
                \uIf{\textit{cs}.overlaps($\vCellSet{}^H$)}
                {
                    \textit{cuesets}.remove(\textit{cs})
                }
            }
            \textit{toHide} = \textbf{InferenceProtect} ($cuesets$) \\
            $\vCellSet{}^H$.addAll(\textit{toHide})  \\
            $cuesets$ = \textbf{InferenceDetect}($toHide$, $\vDataDepSet{}, \view$) \\
        }
        \For{$\vCell{}{i} \in $ $\vCellSet{}^H$}
        {
            Replace $\vCell{}{i}$\textit{.val} in $\view_S$ with \textit{NULL}
        }
        \Return{$\view_S$}
    \caption{Full Algorithm }
    \label{alg:FullAlgo}
\end{algorithm}

Our approach (Algorithm~\ref{alg:FullAlgo}) takes as input a user $\vUser{}{}$, a set of schema level dependencies $\vDataDepSet{}$, and a view of the database $\view$ (initially set to the original database). 
The algorithm first determines the set of sensitive cells $\vCellSet{}^S$
(\emph{Sensitivity Determination} function 
for $\vUser{}{}$ and $\view$).
Sensitivity determination identifies the policies applicable to a querier using the subject conditions in policies and marks a set of cells as sensitive thus assigning them with \nullvalue in the view. 
The set of 
sensitive cells are added into a set of hidden cells ($hidecells$) which will be finally hidden in the secure view ($\view_S$) that is shared with the user $\vUser{}{}$. Next, the algorithm generates the 
cuesets for cells in $hidecells$ 
using $\vDataDepSet{}$ and $\view$ (\emph{Inference Detection}, line 3).
Given the cuesets, the algorithm 
chooses a set of 
cells to hide such that the selected cells cover each of the cuesets 
(\emph{Inference Protection}).
This process of cueset identification protection continues iteratively
as new  hidden cells get added. 
The algorithm terminates when for all of the cuesets there exists a cell that is already hidden.
Finally, we replace the value of $hidecells$ in $\view_S$ (initialized to $\view$) with \nullvalue and return this secure view to the user (lines 13-16).
The following theorem states that the algorithm successfully implements the recursive hiding of cells in $\vCellSet{}^H$ which is required for generating a querier view that achieves full deniability (as discussed in Theorem~\ref{theorem:fdcell}).

\begin{theorem}
When Algorithm~\ref{alg:FullAlgo} terminates, $\forall \vCell{}{i} \in \vCellSet{}^H$, $\forall \vDataDep{}{} \in \vDataDepSet{}$, for all $cueset(\vCell{}{i}, \vDataDep{}{})$ that is non-empty, there exists $\vCell{}{j} \in cueset(\vCell{}{i}, \vDataDep{}{})$ such that $\vCell{}{j} \in \vCellSet{}^H$ (i.e., Algorithm 1 has recursively hidden $\geq$ 1 cell from all the non-empty cuesets of cells in $\vCellSet{}^H)$.
\end{theorem}

\commentrequired{
\begin{proof}
    By contradiction, we suppose there exists a cueset $cs \in cueset(\vCell{}{i}, \vDataDep{}{})$ in which no cell is not hidden.
    This means the cueset $cs$ has no overlap with the hidden cell set $\vCellSet{}^H$.
    Then by lines 6-7 in Algorithm 1, the cueset $cs$ exists in the cueset list $cueset(\vCell{}{i}, \vDataDep{}{})$, which indicates that the While loop will not terminate.
    This contradicts the pre-assumed condition.
\end{proof}
}

\subsection{Inference Detection}
\label{subsec:cueset_detection}

\begin{algorithm}[t]
    \SetAlgoLined
    \DontPrintSemicolon
    \KwInput{A set of sensitive cells $\vCellSet{}^{S}$, Schema-level data dependencies $\vDataDepSet{}$, A view of the database $\view$}
    \KwOutput{A set of cuesets $cuesets$}
    \SetKwFunction{FMain}{InferenceDetect}
    \SetKwProg{Fn}{Function}{:}{}
    \Fn{\FMain{$\vCellSet{}^{S}$, $\vDataDepSet{}$, $\view$}}{
    
        cuesets = \{ \} \\
        \For{$\vCell{*}{} \in \vCellSet{}^{S}$}{
            $S_{\vDataDepSet{}}$ = \{ \} \Comment{Set of instantiated dependencies.} \\
            \For{$\vSchemaDep{} \in \vDataDepSet{}$}
            {
                $S_{\vDataDepSet{}}$ = $S_{\vDataDepSet{}}$ $\cup$ \textbf{DepInstantiation}($\vSchemaDep{}$, $\vCell{*}{}$, $\view$)  \\
            }
            
            \For{$\vDataDep{}{} \in$ $S_{\vDataDepSet{}}$}
            {   
                \uIf{$| \vPredicates{\vDataDep{}{}} | = 1$ }
                {
                    cueset = \{$\vCell{}{k}$\} \Comment{Note: $\vPredicate{}{}(\vCell{*}{}) = \vCell{*}{} \vOperator \vCell{}{k}$}
                }
                \uElseIf{ $TTC(\vDataDep{}{}, \view, \vCell{*}{}) = False$}
                {
                    \textbf{continue} 
                }
                \Else
                {   
                   cueset = $cells(\vPredicates{\vDataDep{}{}\backslash\vCell{}{})}$ 
                }
                cuesets.add(cueset)
            }
        }
        \Return{cuesets} 
    }
    \caption{Inference Detection}
    \label{alg:cuedetector}
\end{algorithm}

Inference detection (Algorithm~\ref{alg:cuedetector}) takes as input the set of sensitive cells ($\vCellSet{}^S$), the set of schema-level dependencies ($\vDataDepSet{}$), and a view of the database ($\view$) in which sensitive cells are hidden by replacing with \null and others are assigned the values corresponding to the instance.
For each sensitive cell $\vCell{*}{}$, we consider the given set of dependencies $\vDataDepSet{}$ and instantiate each of the relevant dependencies $\vSchemaDep{}$ using the database view $\view$ (lines 5-7). 
The \emph{DepInstantiation} function returns the corresponding instantiated dependency $\vDataDep{}{}$.
For each such dependency instantiation, if it is a dependency containing a single predicate i.e., $\vDataDep{}{} = \neg(\vPredicate{}{})$ where $\vPredicate{}{} = \vCell{*}{} \vOperator \vCell{}{k}$, we add the non-sensitive cell ($\vCell{}{k}$) to the cueset (lines 9-10).
If the dependency contains more than a single predicate, we determine if there is leakage about the value of the sensitive cell by checking the Tattle-Tale Condition (TTC) for the sensitive cell $\vCell{*}{}$ (line 11)\footnote{While not shown in the algorithm for simplicity, when an input cell is sensitive in an FC instantiation, if the FC is non-invertible we ignore its cuesets as they are empty.}. 
If $TTC(\vDataDep{}{}, \view, \vCell{*}{})$ evaluates to \textit{False}, we can skip that dependency instantiation as there is no leakage possible on $\vCell{*}{}$ due to it (line 12). 
However, if $TTC(\vDataDep{}{}, \view, \vCell{*}{})$ evaluates to \textit{True}, we get all the cells except for $\vPredicate{}{}(\vCell{*}{})$ (line 14)\footnote{If we wish to relax the assumption that queriers and data owners do not overlap stated in Section~\ref{sect:assumptions}, we can do so here by only including the cells in the cueset that do not belong to the querier. We show algorithms to achieve so and prove the correctness of this modification in Section~\ref{sec:assumption_relaxation}.}.
After iterating through all the dependency instantiations for all the sensitive cells, we return \textit{cuesets} (line 19).

Note that in our inference detection algorithm, 
we did not choose the non-sensitive cell $\vCell{\prime}{}$ in $\vPredicate{}{}(\vCell{*}{}) = \vCell{*}{} \vOperator \vCell{\prime}{} $ as a candidate for hiding. We illustrate below using a counter-example why hiding $\vCell{\prime}{}$ might not be enough to prevent leakages.

\begin{example}
Consider a relation with 3 attributes $A_1, A_2, A_3$ and 
3 dependencies among them ($\vSchemaDep{1}: A_1 \rightarrow A_2$, $\vSchemaDep{2}: A_2 \rightarrow A_3$, $\vSchemaDep{3}: A_1 \rightarrow A_3$).
Let there be two tuples in this relation $t_1: 1 (\vCell{}{1}), 2 (\vCell{}{2}), 2 (\vCell{}{3})$ and $t_2: 1 (\vCell{}{4}), 2 (\vCell{}{5}), 2 (\vCell{}{6})$.
Suppose $\vCell{}{6}$ is sensitive. As leakage of the sensitive cell is possible through the dependency instantiation $\vDataDep{}{2}: \neg((\vCell{}{2} = \vCell{}{5}) \land (\vCell{}{3} = \vCell{}{6}))$, $\vCell{}{5}$ is hidden. 
In the next iteration of the algorithm, to prevent leakages on the hidden cell $\vCell{}{5}$ through dependency instantiation $\vDataDep{}{1}: \neg((\vCell{}{1} = \vCell{}{4}) \land (\vCell{}{2} = \vCell{}{5}))$, $\vCell{}{2}$ is also hidden. 
Note that $\vCell{}{2}$ is in the same predicate as $\vCell{}{5}$ in $\vDataDep{}{1}$. 
However, the querier can still infer the truth value of the predicate $\vCell{}{2} = \vCell{}{5}$ as \true based on the two non-hidden cells, $\vCell{}{1}$ and $\vCell{}{4}$, and the dependency instantiation \reviseTKDE{$\vDataDep{}{3}: \neg((\vCell{}{1} = \vCell{}{4}) \land (\vCell{}{3} = \vCell{}{6}))$}.
The querier also learns that $\vCell{}{3} = \vCell{}{6}$ evaluates to \true in $\vDataDep{}{2}$ which leads to them inferring that $\vCell{}{6} = 2$ (same as $\vCell{}{3}$) and complete leakage.
\end{example}

To prevent any possible leakages on the sensitive cell $\vCell{*}{}$ and its corresponding predicate $\vPredicate{}{}(\vCell{*}{})$, we only consider the solution space where a cell from a different predicate ($\vPredicates{\vDataDep{}{}\backslash\vCell{*}{}}$) is hidden.

\commentrequired{
\stitle{Query-based method.}
\label{sect:optimization_query}
For each dependency and each sensitive cell, inference detection instantiates the dependency to generate $| \vDatabase | -1$ instantiations. The algorithm then iterates over each instantiation and checks the Tattle-tale condition and if satisfied generates a cueset.
The inference detection algorithm will be time and space-intensive given a substantial number of dependencies and/or sensitive cells.  
To improve upon this, we propose a query-based technique for implementing inference detection.

Instead of generating one instantiation per sensitive cell and dependency, this method produces one query for all the sensitive cells.
First, this method retrieves the tuples containing sensitive cells, sets the values of sensitive cells to \textit{NULL} and stores them in a temporary table called \textbf{\textit{temp}}.
Next, the Tattle-tale condition check is turned into a join query between this \textbf{\textit{temp}} table and the original table.

}

\commentrequired{
The join condition in this query is based on the tuples being unique (\textit{T1.tid $\neq$ T2.tid}).
Furthermore, this query checks for each relevant attribute in the tuple whether it is sensitive i.e., it is set to \nullvalue in the \textit{temp} table (\textit{T2.Zip is NULL}), or whether the corresponding predicate from the dependency evaluates to \true (\textit{T1.Zip=T2.Zip}).
The \textit{WHERE} condition in this query is only satisfied if all the predicates in a dependency instantiation except for the sensitive predicate evaluate to \true.
Thus, the result of this join query contains all instantiations for which the Tattle-tale condition evaluates to \true from which the cuesets can be readily identified.
}

\subsection{Inference Protection}
\label{subsec:inference_protection}

After identifying the cuesets for each sensitive cell based on their dependency instantiations, we now have to select a cell from each of them to hide to prevent leakages.
The first strategy for cell selection, described in Algorithm \ref{alg:random}, randomly selects a cueset and a cell from it to hide (if no cells in it have been hidden already). We use this approach as our first baseline (\baselineOne) in Section~\ref{sect:experiments}.
The second strategy for cell selection, described in Algorithm \ref{alg:MVC} utilizes Minimum Vertex Cover (MVC)~\cite{cormen2009introduction} to minimally select the cells to hide from the list of cuesets.
In this approach, each cueset is considered as a hyper-edge and the cell selection strategy finds the minimal set of cells that covers all the cuesets.
MVC is known to be NP-hard \cite{dinur2005hardness} and therefore we utilize a simple greedy heuristic based on the membership count of cells in various cuesets.
Algorithm~\ref{alg:MVC} takes as input the set of cuesets and returns the set of cells to be hidden to prevent leakages.
First, we flatten all the cuesets into a list of cells and insert this list into a dictionary with the cell as the key and their frequency count as the value (lines 4-5).
Next, we select the cell from the dictionary with the maximum frequency and add it to the set of cells to be hidden and remove any cuesets that contain this cell (lines  7-10).  
These steps are repeated until all the cuesets are covered i.e., at least one cell in it is hidden, and finally, we return the set of cells to be hidden.

\begin{algorithm}[t]
    \SetAlgoLined
    \DontPrintSemicolon
    \SetKwFunction{FMain}{InferenceProtect}
    \SetKwProg{Fn}{Function}{:}{}  
    \KwInput{Set of cuesets \textit{cuesets}}
    \KwOutput{A set of cells selected to be hidden \textit{toHide}}
    \Fn{\FMain{$cuesets$}}{
    
        \textit{toHide} = $\{ \}$\\
    
        \While{cuesets $\neq \phi$}
        {
            \textit{cuesetCells} = \textbf{Flatten}(\textit{cuesets})  \\
            \textit{dict}[$\vCell{}{i}, freq_i$] = \textbf{CountFreq}(\textbf{GroupBy}(\textit{cuesetCells}))\\
            \textit{cellMaxFreq} = \textbf{GetMaxFreq}(\textit{dict}[$\vCell{}{i}, freq_i$]) \\
            \textit{toHide}.add($cellMaxFreq$) \\
            \For{\textit{cs} $\in$ \textit{cuesets}}
            {
                \uIf{\textit{cs}.overlaps(\textit{toHide})}
                {
                    \textit{cuesets}.remove(\textit{cs})
                }
            }
        }
        \Return{toHide}
    }
    \caption{Inference Protection (MVC)}
    \label{alg:MVC}
\end{algorithm}

\subsection{Convergence and Complexity Analysis} 
\label{sect:complexity}

Algorithm~\ref{alg:FullAlgo} starts with $s$ number of hidden cells.  At each iteration, we consider that each hidden cell (including cells that are hidden in previous iterations) is expanded to $f$ number of cuesets on average by the Inference Detection algorithm (Algorithm~\ref{alg:cuedetector}). Among the cuesets, the average number of cells that are hidden, such that it satisfies full deniability, is given by $\frac{f}{m}$ where $m$ is the coverage factor determined by minimum vertex cover (MVC). Then, at the end of $i$th iteration, the number of average hidden cells will be $s_i=s (\frac{f}{m})^i$, and the average number of cuesets will be $cs_i = sf(\frac{f}{m})^{i-1}$. As $s_i$ is bounded by the total number of cells in the database, denoted by $N$, the number of iterations ($T$) to converge is bounded by $\log_{f/m}(N/s)$, when $f>m$ (which was verified in our experiments).  

Given $|\Delta|$ which is the number of schema-level dependencies, we can estimate the time complexity with respect to I/O cost.
At $i$th iteration of Algorithm~\ref{alg:FullAlgo}, the I/O cost of (i) the dependency instantiation is $\mathcal{O}( |\Delta|(N+ s_i))$ (where inference detection is implemented using the query-based method given sufficient, i.e. $\Theta(N)$, memory) and (ii) minimum vertex cover (MVC) with an I/O cost of $\mathcal{O}(cs_i)$. 
Hence, the overall estimated I/O cost $\sum_{i=1}^T \mathcal{O}(|\Delta|(N+ s_i)) + \mathcal{O}(cs_i)$ in which is equivalent to $\mathcal{O}(N)$ given $T\leq \log_{f/m}(N/s)$ and thus is linear to the data size. 

\extend{
The cost of the dependency instantiation for the $i$th iteration depends on the I/O cost of the join query which is  $\mathcal{O}(N+ s_i)$ when given sufficient (i.e., $\Theta(N)$) memory. This query is executed $|{\Delta}|$ times. Hence, the cost for the dependency instantiation is $\mathcal{O}(|{\Delta}|(N +s_i))$.

Hence, the total estimated I/O cost for $T$ iterations can be derived as follows given $T\leq \log_{f/m}(N/s)$.

\begin{eqnarray}
&    &\sum_{i=1}^T (|{\Delta}|(N + s_i)) + c_i \nonumber  \\ 
&= & |{\Delta}|(N + s \sum_{i=1}^T (f/m)^i)  + sf \sum_{i=1}^{T} (f/m)^{i-1} \nonumber \\ 
&\leq & |{\Delta}| (N + s (f/m)^{T+1}) + sf (f/m)^T \nonumber \\ 
&=& |{\Delta}|(N + s (N/s) (f/m)) + sf (N/s)  \nonumber \\
&=& N |{\Delta}| (1 + f/m) + fN  \nonumber 
\end{eqnarray}

We complement the complexity analysis with the required sufficient memory storage discussion.
For (i) dependency instantiation, the join query between two tables of size $N$ and $s_i$, we need memory size $\Omega(N+s_i) = \Omega(N)$ since $s_i \leq N$.
In (ii) the algorithm of computing MVC, all cuesets are read into the memory, which requires the memory size $\Omega(c_i) = \Omega(N*m) = \Omega(N)$ for constant $m$.
Thus we need $\Omega(N)$ memory to finish all operations in our system implementation, which is feasible in practice.
We also note that this complexity analysis only holds with $\Theta(N)$ size of memory, in which case the cost of memory operations is much cheaper than the overhead of I/O operations.
Given $\Omega(N^2)$ memory, which can be \emph{impractical}, all the operations can then be finished within memory and the total computational cost is bounded by $\mathcal{O}(N^2)$, according to the following analysis.

If all operations are taken within memory, then the cost of dependency instantiation is bounded by $\mathcal{O}(Ns_i)$ and the computational cost of the MVC algorithm is bounded by $\mathcal{O}(c_i^2)$. Then we derive the following bound similarly.
\begin{eqnarray}
&    &\sum_{i=1}^T (N|{\Delta}| s_i) + c_i^2 \nonumber  \\ 
&= & N|{\Delta}|  s \sum_{i=1}^T (f/m)^i  + s^2f^2 \sum_{i=1}^{T} (f/m)^{2i-2} \nonumber \\ 
&\leq & N|{\Delta}|   s (f/m)^{T+1} + s^2f^2 (f/m)^{2T} \nonumber \\ 
&=& N|{\Delta}|  s (N/s) (f/m) + s^2f^2 (N/s)^2  \nonumber \\
&=& N^2 |{\Delta}| (f/m) + f^2N^2  \nonumber 
\end{eqnarray}
}

\extend{
\subsection{Wrapper for Scaling out Full-Deniability Algorithm}
\label{app:wrapper_algo}
The complexity analysis above shows that, given sufficient memory, full deniability algorithm is linear to the size of the database. On larger databases, the memory requirement becomes unsustainable  due to the substantial number of dependency instantiations and cuesets.
We present a wrapper which partitions the database in order so that our algorithm is able to run with a smaller memory footprint.

The high-level idea of the wrapper algorithm is illustrated in Figure \ref{fig:btm}.
Algorithm~\ref{alg:binning_then_merging} partitions the full database into a number of \emph{bins}, where $b$ is the bin size parameter. It then calls the \emph{Full Algorithm} (presented in Section 5.1 and denoted by \texttt{runMain()} in Algorithm \ref{alg:binning_then_merging}) on each of these bins in order to generate a view per bin that satisfies full deniability.
As the full algorithm is executed on smaller bins, the memory requirement is much lower than the entire database.
Next, it merges $m$ number of these bins, where $m$ is the merge size parameter, and executes \emph{Full Algorithm} 
on the merged bins.
The wrapper iterates over the merged bins until there is only 1 bin left. It then executes \emph{Full Algorithm} on this last bin which is full database and the final view that satisfies full deniability is returned.
As each of the bins has achieved full deniability, the number of relevant dependency instantiations and cuesets will be much lower in the merged bin compared to running the full algorithm on the entire database.
The output view generated by Algorithm~\ref{alg:binning_then_merging} trivially satisfies full-deniability as the \emph{Full Algorithm} is executed on each of the individual bins as well as the full database in the final step. 
}

\begin{figure}[t]
    \centering
    \includegraphics[width=\linewidth]{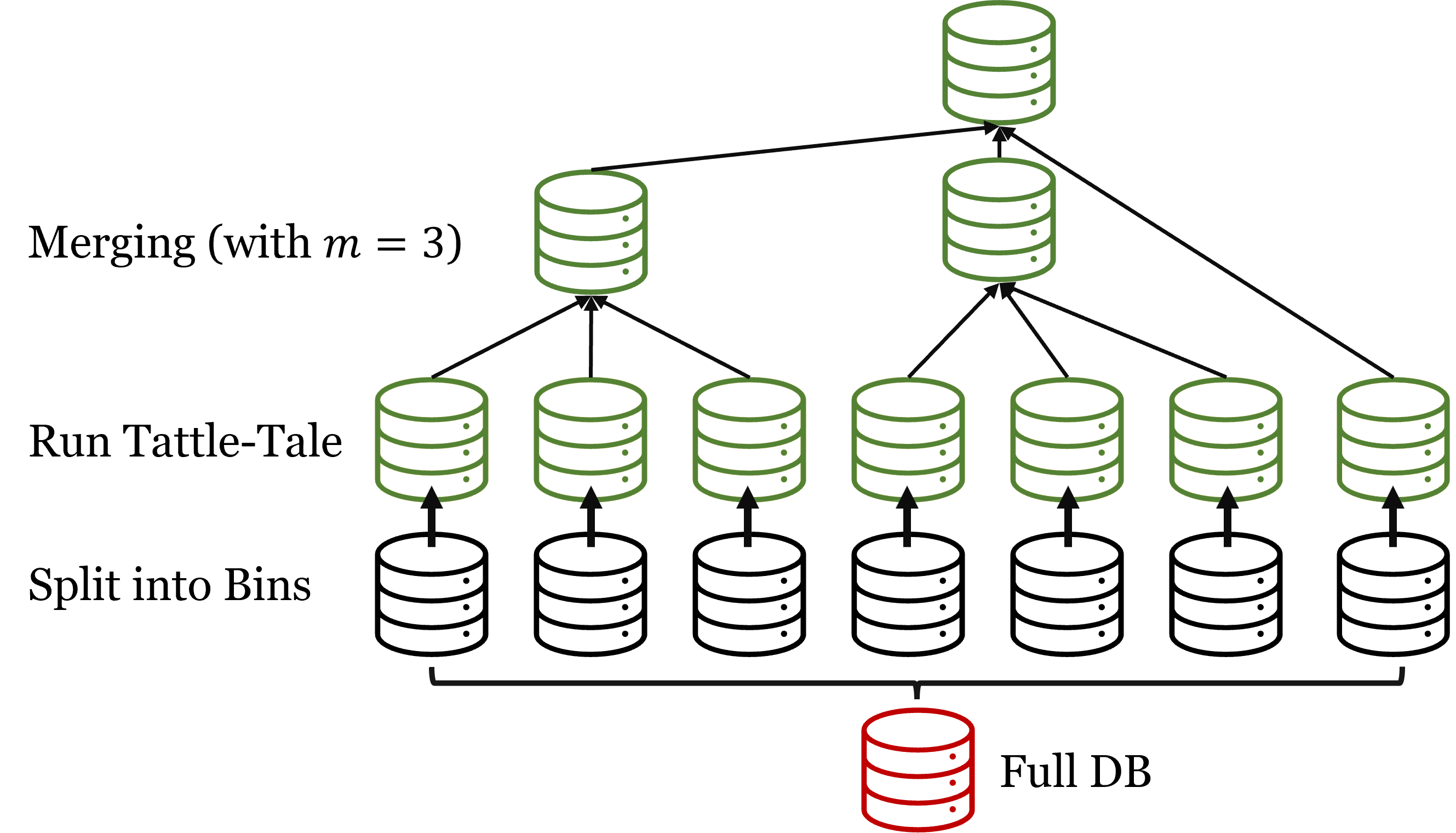}
    \caption{\reviseA{An Illustration of the Binning-then-Merging Algorithm (with binning size $b=7$ and merging size $m=3$).}}
    \label{fig:btm}
\end{figure}

\begin{algorithm}[t]
    \SetAlgoLined
    \DontPrintSemicolon
    \KwInput{User $\vUser{}{}$, Data dependencies $\vDataDepSet{}$, A view of the database $\view$, Bin size $b$, Merge size $m$}
    \KwOutput{A secure view $\view_{S}$}
    \SetKwFunction{FMain}{BinningThenMerging}
    \SetKwProg{Fn}{Function}{:}{}    
    \Fn{\FMain{$\vUser{}{}$, $\vDataDepSet{}$, $\view$, $b$, $m$}}{
        $\view_1, \dots, \view_{k} \gets$ \textbf{Binning}($\view, b$) \Comment{$k \coloneqq \frac{\|\view\|}{b}$, no. of bins.} \\
        binQueue = [$\view_1, \dots, \view_{k}$] \\
        mergeQueue = \{ \} \\
        
        \While{$\|binQueue\| \neq 1$ or  mergeQueue $\neq \emptyset$}{
            $\view_i \gets$ binQueue.\textbf{pop}() \\
            mergeQueue.\textbf{push}(\textbf{runMain}($\vUser{}{}$, $\vDataDepSet{}$, $\view_i$)) \\
            \uIf{$\| mergeQueue \| \geq m$ or $\|binQueue\| = 0$ }{
                $\view_j \gets$ \textbf{Merge}(mergeQueue) \\
                binQueue.\textbf{push}(\textbf{runMain}($\vUser{}{}$, $\vDataDepSet{}$, $\view_j$)) \\
                mergeQueue.\textbf{clear}()\\
            }
        }
        \Return{binQueue.\textbf{pop}()}
        
    }
    \caption{\reviseA{Binning-then-Merging Scaler}}
    \label{alg:binning_then_merging}
\end{algorithm}

\section{\reviseTKDE{Weaker Security Model}}
\label{sec:weaker_security_model}

Achieving full deniability on  a view can lead to hiding a number of non-sensitive cells to prevent leakages. In this section we describe how to relax full deniability to a weaker security model which we call, \emph{k-percentile deniability}, in order to potentially hide fewer cells and thus improve utility.

\subsection{\texorpdfstring{$k$}{k}-Percentile Deniability}

The weaker security notion of $k$-Percentile Deniability is defined as follows.

\begin{definition}[k-percentile Deniability]
\label{Def:k-den}
Given a set of sensitive cells $\vCellSet{}^S$ in a database instance $\vDatabase$ and a set of instantiated dependencies $\vDataDepSet{}$, we say that a querier view $\view$ achieves k-percentile deniability if for all $\vCell{*}{} \in \vCellSet{}^S$,
\begin{eqnarray}
\quad \vCount{\inference(\vCell{*}{}|\view, \vDataDepSet{}{})} ~~ \geq (k \cdot \vCount{\inference(\vCell{*}{}|\view_0, \vDataDepSet{}{})})
\end{eqnarray}
where $\frac{1}{\vCount{\inference(\vCell{}{}|\view_0, \vDataDepSet{}{})}} \leq k \leq 1 $.
\end{definition}

Note that if $k=1$, then k-percentile deniability is the
same as  full deniability, where the set of values inferred by the adversary from view $\view$ is the same as the set from the base view. With $k < 1$, it allows for a bounded amount of leakage. 
We also note that the security models
used in prior works is subsumed by the notion of 
 $k$-percentile deniability as defined above.
 For instance, the model used in \cite{brodsky2000secure} 
 ensures that the querier cannot reconstruct the exact value of the sensitive cell using data dependencies, which can be viewed as a special case of  $k$-percentile deniability with the value of 
 $k=\frac{2}{\vCount{\inference(\vCell{}{}|\view_0, \vDataDepSet{}{})}}$, i.e., the number of values sensitive cell can take is more than 1.

\subsection{Algorithm to Achieve \texorpdfstring{$k$}{k}-Percentile Deniability}
In $k$-percentile deniability or simply \emph{k-den}, we quantify the leakage on the sensitive cell in a given view $\view$ and the set of instantiated data dependencies $\vDataDepSet{}$. 
Unlike in \emph{full deniability}, where any inference is considered as leakage, in \emph{k-den} the decision to hide additional cells is only made if the set of possible values inferred by the querier is larger than the given threshold ($k$). 
In every iteration, the output of the inference detection algorithm could be illustrated as a tree rooted in the cell to be hidden, with cuesets forming its fan-out.
Given this tree structure, we can \emph{prune} some of the cuesets at the first level based on the \emph{k} value.
 
The set of inferred values for a $\vCell{*}{}$ given by the $\inference(\vCell{}{} \mid \view, \vDataDepSet{})$ (defined in Section~\ref{sect:inf_fun}) can be represented as follows

\begin{equation} 
\inference(\vCell{}{} \mid \view, \vDataDepSet{}) = 
    \begin{cases}
        minus\_set, & \vDomain{\vCell{}{}} \textit{ is discrete} \\
        \vDomain{\vCell{}{}} - [low, high] &  \vDomain{\vCell{}{}} \textit{ is continuous}
    \end{cases}
\label{eq:cellView}
\end{equation}

When attribute for the cell $\vCell{}{}$ is discrete, the operator $\vOperator$ in $\vPredicate{}{}(\vCell{}{})$ is limited to either $\neq$ or $=$. Therefore, we represent the inferred set of values by a set, called minus\_set, containing the values that cannot be assigned to the cell in the view $\view$. 
On the other hand, when the attribute for the cell $\vCell{}{}$ is continuous, the operator $\vOperator$ could be one of the following: \{$>, \geq, <, \leq\}$ and therefore we use a range, denoted by $(low, high)$ to represent the set of values.
The details of the algorithm to compute the set of inferred values for a cell can be found in supplementary materials.

This function computes the exact leakage on a sensitive cell with respect to various instantiated dependencies.
We utilize this to implement $k$-den where for each sensitive cell after we detect the cuesets (line 3 in Algorithm~\ref{alg:FullAlgo}), we compare the leakage on a sensitive cell due to the instantiated dependencies (associated with the cuesets).
The $k$ parameter, specified as a fraction of the maximum domain size of a sensitive cell, provides a bound on the acceptable leakage on a sensitive cell.
If the sensitive cell $\vCell{*}{}$ has a discrete domain and $\vCount{\vCell{*}{}.minus\_set}$ $\leq  \vCount{\vDomain{\vCell{*}{}}}  \times (1 - k)$  evaluates to \true, we do not hide any cells from any of its cuesets.
On the other hand, for a sensitive cell $\vCell{*}{}$ with a continuous domain we check if $high - low$ $\geq  \vCount{\vDomain{\vCell{*}{}}}  \times (k)$  evaluates to \true.
The difference between $low$ and $high$ gives the actual domain size after taking into account leakages due to dependencies.

The algorithm to implement $k$-den only needs to hide non-sensitive cells from the cuesets until the leakage is below the \emph{k} value for a sensitive cell.
Algorithm \ref{alg:kprune} utilizes the previously described tree representation of the output of the inference detection with cuesets as fan-out from the hidden cells. At the first level of this tree (i.e., for each sensitive cell), we calculate the set of inferred values (using Equation~\ref{eq:cellView}), based on the cuesets in its fan-out (lines 1-12). 
Since the pruning only happens at the first level of the tree, the algorithm implements the full-deniability algorithm at higher levels (i.e., $>$ 1, lines 8-11).
We sort the cuesets in descending order of their set of inferred to values to the parent sensitive cell and select cells to hide from them until the $k$-deniability is met (lines 12-21).
Note that this k-pruning step is only executed in the first fan-out level -- this ensures that the final solution generated by $k$-deniability is an improvement over the full deniability model.

\begin{theorem}
The algorithm to achieve k-percentile deniability (i.e. algorithm \ref{alg:kprune}) always performs as well as (or better than) the algorithm to achieve full-deniability (i.e. algorithm \ref{alg:FullAlgo}). 
\label{theorem:k-den}
\end{theorem}

\begin{proof}
    We note that the KPrune algorithm implicitly simulates the full-deniability algorithm.
    It does not immediately prune the cuesets or the cells to hide from the fan-out tree generated by the full-deniability algorithm (since this can change the result of running the greedy minimum vertex cover).
    Instead, we collect those cuesets that can be pruned but actually prune out them after simulating the overall full-deniability algorithm.
    Therefore, the KPrune algorithm won't hide more cells than the algorithm to achieve full-deniability.
\end{proof}

\begin{algorithm}[t]
    \SetAlgoLined
    \DontPrintSemicolon
    \SetKwFunction{FMain}{KPrune}
    \SetKwFunction{FSub}{isDeniable}
    \SetKwProg{Fn}{Function}{:}{}  
    \KwInput{Last level hidden cells $trueHide$, Current level hidden cells to prune $toHide$, Current level $level$, Leakage parameter $k$}
    \KwOutput{An updated minimum set of hidden cells in this level that satisfy k-deniability \textit{trueHide}}
    \Fn{\FMain{$trueHide$, $toHide$, $level$, $k$}}{
            $bestCuesets = \{ \}$ \Comment{Cuesets cannot be pruned.} \\
    
            \For{$cell \in trueHide$}{
                cellCuesets = cell.getCuesets() \\
                $cell$.leakage = \textbf{InferredValues}(cell, cellCuesets) \\
                \uIf{\textbf{isDeniable}(cell, $k$)}{
                    \textbf{continue}
                }    
                \For{$cs \in$ cellCuesets}{
                    \uIf(\Comment{Simulate full-den.}){$level > 1$}{  
                            $bestCuesets$.add($cs$) \\
                    }
                }
                
                \uIf(\Comment{KPrune condition.} ){$level = 1$}{
                    
                    cellCuesets.\textbf{Sort}(leakageToParent, `desc') \\
                    \While{\textbf{not} \textbf{isDeniable}(cell, $k$)}{
                        $lcs$ = cellCuesets.head \Comment{Max leakage.} \\
                        $bestCuesets$.add($lcs$) \\
                        cellCuesets.remove($lcs$) \\
                            \Comment{Recalculate the leakage of the cell.} \\
                        $cell$.leakage = \textbf{InferredValues}(cell, cellCuesets) \\
                    }               
            }
            
            }
            \For{$bestCS \in $ $bestCuesets$}{
                \Comment{Update $trueHide$ based on the pruning.} \\
                $trueHide$ = $trueHide$ $\cup$ ($toHide$ $\cap$ $bestCS$.cells)
            }
            
            \Return{trueHide}
        }
    \Fn{\FSub{$cell$, $k$}}{
        
        \uIf{$\vCount{\vDomain{\vCell{*}{}}}  - \vCount{cell.leakage}  \geq  k \cdot \vCount{\vDomain{\vCell{*}{}}}$  }{
            \Return{True} \Comment{Based on k-deniability.} \\
        }
        \Return{False}
    }
    \caption{KPrune: Achieving \emph{k}-Deniability}
    \label{alg:kprune}
\end{algorithm}

In Section~\ref{sect:experiments}, we show through experiments that the algorithm that achieves $k$-percentile deniability only  marginally improves on full deniability even with low values of $k$ (i.e., complete leakage).
Therefore this approach is not useful in improving the utility in realistic settings. It is possible that in more complex domains with large number of sensitive cells, $k$-percentile deniability is more effective and this needs to be studied further.

\section{\reviseTKDE{Relaxing Security Assumptions}}
\label{sec:assumption_relaxation}

\reviseTKDE{
In this section, we explore relaxing an important assumption stated in Section~\ref{sect:assumptions} about the adversary, that the adversary cannot apriori determine whether a cell is sensitive or not.
There may be scenarios where the adversary can accurately guess the relative sensitivity of the attributes in a database schema. 
For example, in an employee table \textit{Salary} is more likely to be sensitive than \textit{Zip Code} and if both are hidden in a tuple the adversary can guess that one was due to policy and the other due to the algorithm.
This situation can be handled by our algorithm with a slight modification under the assumption that any tuple in the database instance could contain a sensitive cell. 
This means that while the adversary knows that salary is more likely to be sensitive, they do not know salaries of exactly which employees are sensitive.
}

\begin{algorithm}[t]
    \SetAlgoLined
    \DontPrintSemicolon
    \KwInput{Map$<$sensitive cell $\vCell{*}{}$: Set of cuesets $cuesets>$, A view of the database $\view$}
    \KwOutput{A set of tuples to hide $toHide$}
    \SetKwFunction{FMain}{InferenceProtection$*$}
    \SetKwProg{Fn}{Function}{:}{}    
    \Fn{\FMain{Map}}{
        \textit{toHide} = $\{ \}$ \Comment{Return set initialization.} \\
    
        \While{Map.cuesets $\neq \phi$}
        {
            \textit{cuesetCells} = \textbf{Flatten}(\textit{Map.cuesets})  \\
            \textit{dict}[$\vCell{}{i}, freq_i$] = \textbf{CountFreq}(\textbf{GroupBy}(\textit{cuesetCells}))\\
            \textit{cellMaxFreq} = \textbf{GetMaxFreq}(\textit{dict}[$\vCell{}{i}, freq_i$]) \\
            \textit{toHide}.add($cellMaxFreq$) \Comment{Greedy heuristic.} \\
            \For{\textit{cs} $\in$ \textit{Map.cuesets}}
            {
                \uIf{\textit{cs}.overlaps(\textit{toHide})}
                {
                    \textit{Map.cuesets}.remove(\textit{cs})
                }
            }
        }
        \reviseA{
        \textit{additionalHiddenCells} = $\{ \}$ \Comment{Hiding additional cells.} \\
        \For{$c_h \in$ toHide}{
        \textit{tid} = $c_h$.getTupleID() \Comment{Hidden cell's tuple ID.} \\
        \For{$c^* \in$ Map.sensitiveCells}{
            \textit{sensitiveAttr} = $c^*$.attributeID \Comment{Sensitive cell's attribute.} \\
            \textit{additionalHiddenCells}.add($\view$.get(\textit{tid}, \textit{sensitiveAttr}))  \\
        }
        }
        \textit{toHide} = \textit{toHide} $\cup$ \textit{additionalHiddenCells} \\
        }
        \Return{toHide}
    }
    \caption{\reviseA{Modified Inference Protection}}
    \label{alg:tuple_hiding}
\end{algorithm}

The key idea behind this modified algorithm is to hide the sensitive cell in a tuple where only the non-sensitive cell is hidden.
From the previous example, we would also hide the \textit{Salary} attribute of a tuple (even when it is not sensitive) if our algorithm chooses to hide \textit{Zip Code}.
Therefore the adversary cannot be certain whether all the hidden cells under \textit{Salary} attribute were done so by policy or the algorithm.
We slightly modify the original Inference Protection algorithm (Algorithm \ref{alg:MVC}) and propose Algorithm \ref{alg:tuple_hiding} in order to handle this relaxed assumption.

First, the original Inference Detection algorithm (Algorithm~\ref{alg:cuedetector}) identifies the cuesets based on dependency instantiations as an input to the Inference Protection algorithm.
Second,  the original Inference Protection algorithm will select at least 1 cell from each cueset to hide.
Third, the steps in the modified Inference Protection Algorithm (lines 13-21) go through the set of hidden cells and for each of them check if they belong to a non-sensitive attribute.
If it does, then add the cells under the sensitive attribute from the corresponding tuple to the set of hidden cells.

\eat{
\begin{theorem}
Given a view $\view^*_{M}$ which is generated by $FDAM$ and has achieved full deniability such that $\vCellSet{}^H$ is the set of hidden cells and $\vCellSet{}^S$ in this view. Consider that views $\view_1, \ldots, \view_i, \ldots, \view_N$ generated by replacing the \nullvalue for $\vCell{}{} \in \vCellSet{guess}^* \subseteq \vCellSet{}^H$ with a value from $\vDomain{\vCell{}{}}$ such that $\view_i$ does not violate any dependencies ($\vCellSet{guess}^*$ is adversary's prior knowledge about the sensitive cells). 
The adversary does not learn any additional information about sensitive cells $\vCellSet{}^S \in \view^*$ after executing $FDAM$ on all of the above views.

\begin{eqnarray}
\inference(\vCell{*}{}|\{\view_1, \ldots, \view_i, \ldots, \view_N \}, \vDataDepSet{}{})= \inference(\vCell{*}{}|\view_0, \vDataDepSet{}{}),
\forall \vCell{*}{} \in \vCellSet{}^S 
\end{eqnarray}

\end{theorem}
}

We note that the assumption of equal likelihood of tuple containing sensitive cell can be further relaxed by adopting a probabilistic approach (motivated by OSDP \cite{kotsogiannis2020one}) in which certain non-sensitive cells are randomly hidden to prevent adversary from inferring if it was part of a sensitive cell's cueset.
However, such an approach will be a non-trivial extension and is an interesting future direction to explore.
In supplementary materials (App. \ref{app:adv_data_owner}), we also discuss how to relax the assumption that adversaries and data owners cannot overlap.

\section{Experimental Evaluation}
\label{sect:experiments}

In this section, we present the experimental evaluation results for our proposed approach to implementing full-deniability. 
First, we explain our experimental setup including details about the datasets, dependencies, baselines used for comparison, evaluation metrics, and system setup.
Second, we present the experimental results for each of the following evaluation goals: 1) comparing our approach against baselines in terms of utility, performance, and the number of cuesets generated; 2) evaluating the impact of \emph{dependency connectivity}; 3) testing the scalability of our system; 4) validating $k$-percentile deniability presented \reviseTKDE{in Section \ref{sec:weaker_security_model} and the modified inference protection algorithm in Section \ref{sec:assumption_relaxation}; 5) evaluating the query-driven utility in a case study when query workloads are presented; and 6)} testing effectiveness against real-world adversaries.

\subsection{Evaluation Setup}
\label{subsec:exp_setup}

\noindent \textbf{Datasets}.
We perform our experiments on 2 different datasets.
Some statistics of the datasets are summarized in the supplementary materials.
The first one is \emph{Tax dataset}~ \cite{bohannon2007conditional}, a synthetic dataset with 10K tuples and 14 attributes, where 10 of them are discrete domain attributes and the rest are continuous domain attributes.
Every tuple from the tax table specifies the tax information of an individual with information such as name, state of residence, zip, salary earned, tax rate, tax exemptions etc.
The second dataset is the \emph{Hospital dataset} \cite{xuchuDC} which is a 100K dataset where all of the 15 attributes are discrete domain attributes.
We select a subset of this dataset (which includes the first 10K tuples of the dataset), called \emph{Hospital10K}, for the experiments included in the paper. 
\extend{
We then conduct a scalability experiment that makes use of the binning-then-merging wrapper on the original Hospital dataset, i.e. 100, to show the scalability of our system.
It is also notable that both datasets have a large domain size, as shown in Table \ref{tab:datasets}.
The active domain size in the table refers to the domain of the attributes participating in the data dependencies that we consider in the experiments.}

\noindent \textbf{Data Dependencies}. 
For both datasets, we identify a large number of denial constraints by using a data profiling tool, Metanome \cite{metanomeDataProfiling}.
Many of the output DCs identified by Metanome were soft constraints which are only valid for a small subset of the database instance.
After manually analyzing and pruning these soft DCs, we selected 10 and 14 hard DCs for the Tax dataset and the Hospital dataset respectively.
We also added an FC based on the continuous domain attribute named \textit{``tax"} which is calculated as a function $``\text{tax} = fn(\text{salary}, \text{rate})"$. 
Since the Hospital dataset does not have continuous domain attributes, we cannot create a function-based constraint on it and just use the 14 DCs for evaluation.
If any of them were soft DCs, we updated/deleted the violating tuples to turn them into hard DCs. 
The data dependencies used for experiments can be found in supplementary materials.

\noindent \textbf{Policies} control sensitivity of cells. 
The number of sensitive cells is equivalent to the number of policies and it helps us in precisely controlling the number of sensitive cells in experiments using policies.
\reviseTKDE{We randomly sample each policy by first sampling a tuple ID among all the tuples and an attribute from a selected group of attributes without replacement, until obtaining a certain number of policies determined by a control parameter.}
\reviseTKDE{For each experiment (with the same set of control parameters), we generate 4 different access control views with different policies to represent 4 users. We execute our algorithm independently over these 4 views and report the mean and standard deviation.}

\noindent \textbf{Metrics.}
We compare our approach against the baselines using the following metrics: 1) \textit{Utility}: measures the number of total cells hidden; \reviseTKDE{2) \textit{Workload-driven utility, i.e., visibility percentage}: measures the percentage of visible cells in queries from a workload;}
3) \textit{Performance}: measures the run time in seconds. 
Besides, we study the fan-out of the number of cuesets, the attack precision of real-world adversaries, and the distribution of the hidden cells in access control and inference control views.

\eat{\color{blue}
The second one ($U_2$) is a weighted utility metric. We assign weights to each attribute based on a pre-set query workload.
The metric $U_2$ will be higher if less cells according to the highly weighted attributes are hidden.
}

\eat{\begin{itemize}
    \item Utility metrics: (i) how to measure utility: \# of additional cells being hidden; (ii) add weights to the attributes based on some queries? (will this change our algorithm in a way we define weight and how to factor this in the MVC)
    
    \item Performance metrics: (i) time consumption; (ii) the fan-out of cuesets/hidden cells \todo{define what is the fanout if it will be used as a metric.} \todo{Be consistent about the fanout, iteration, and level.}

    \item Privacy metrics: describe imputation algorithm used, and how do you measure the leakage/the number of correctly cells. 
\end{itemize}}

\noindent \textbf{System Setup}.
We implemented the system in Java 15 and build the system dependencies using Apache Maven.
We ran the experiments on a machine with the following configuration: Intel(R) Xeon(R) CPU E5-4640 2.799 GHz, CentOS 7.6, with RAM size 64GB. We chose the underlying database management system MySQL 8.0.3 with InnoDB.
For each testcase, we perform 4 runs and report the mean and standard deviation.

\extend{
\noindent
\textbf{Reproducibility.}
We open-source our codebase on GitHub\footnote{\href{https://github.com/zshufan/Tattle-Tale}{https://github.com/zshufan/Tattle-Tale}}. This codebase includes the implementation of our system as well as scripts to set up databases,  generate testcases, run end-to-end experiments, and plot the empirical results.
For experiment reproducibility instructions please follow the guidelines in the README file in the GitHub repository.
}

\noindent \textbf{Baselines}.
In the following experiments, we test our approach which implements Algorithm~\ref{alg:FullAlgo}, denoted by \ourapproach against baselines.
\revise{
To the best of our knowledge, there exist no other systems which solve the same problem and therefore we have developed 2 different baseline strategies for comparison. 
In each baseline method, we replace one of the key modules in our system, determining cuesets and selecting cells to hide from the cueset,  with a na\"ive strategy but without compromising the full deniability of the generated querier view.
}

\squishlist

    \item \textit{Baseline 1: Random selection strategy for hiding }(\baselineOne):
    which replaces the minimum vertex cover approach with an inference protection strategy that randomly selects cells from cuesets to hide.

    \item \textit{Baseline 2: Oblivious cueset detection strategy }(\baselineTwo):
    which disregards Tattle-Tale Condition and uses an inference detection strategy that creates as many dependency instantiations as the number of tuples in the database for each dependency and generates cuesets for all of them.
\squishend

\eat{    .... 
    (1. full-den without MVC, 2. by using partial property of the cueset)
    (1. just the sensitive cells, 2. randomly hide similar number of cells according to the results from full-den, 3. by using partial property of the cueset)}

\begin{figure}[t]
\centering
\includegraphics[width=.48\linewidth]{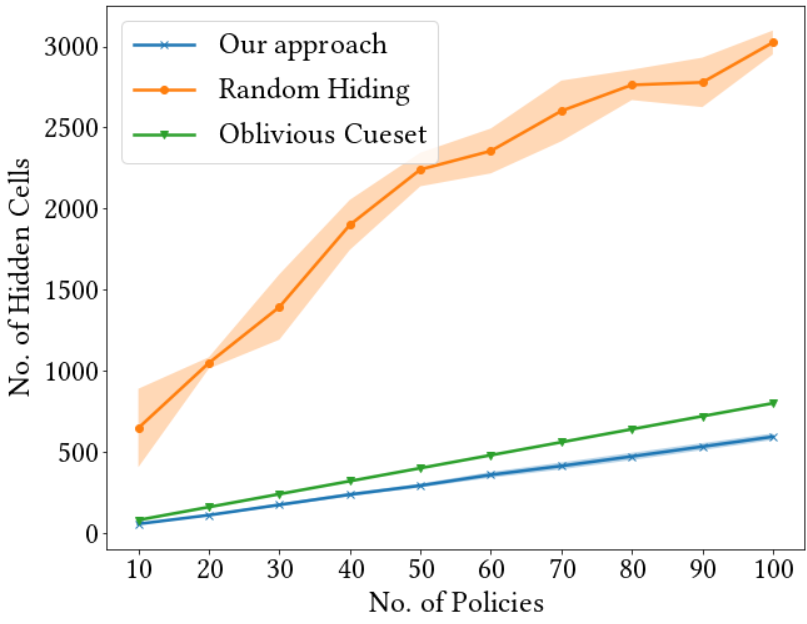}
\includegraphics[width=.48\linewidth]{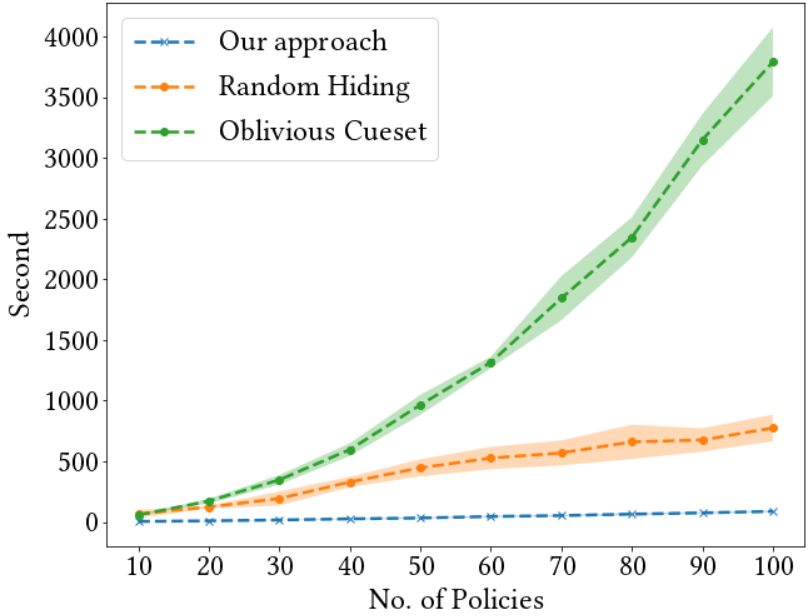}
\caption{(a) Data utility (b) Performance.  Experiments done on Tax dataset for \ourapproach, \baselineOne, and \baselineTwo.}
\label{fig:utility_efficiency_tax}
\end{figure}

\begin{figure}[t]
\centering
\includegraphics[width=.48\linewidth]{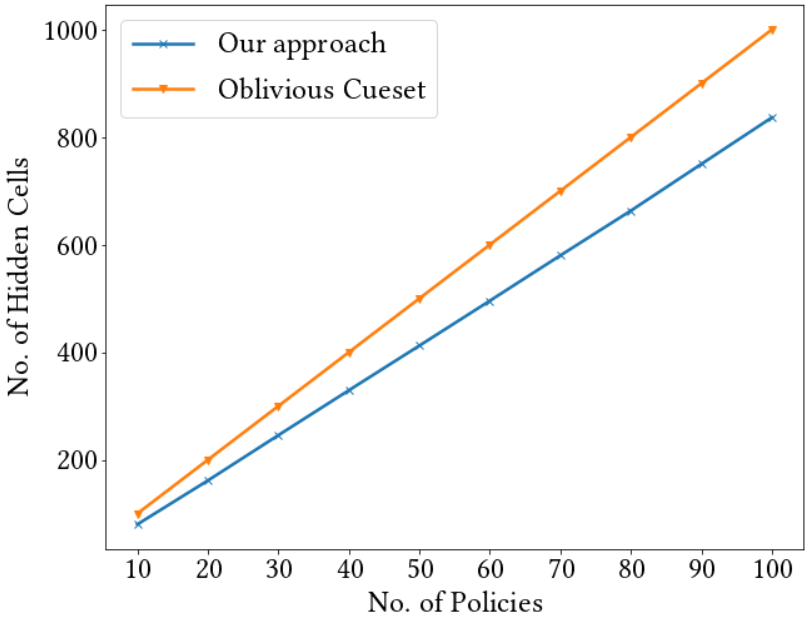}
\includegraphics[width=.48\linewidth]{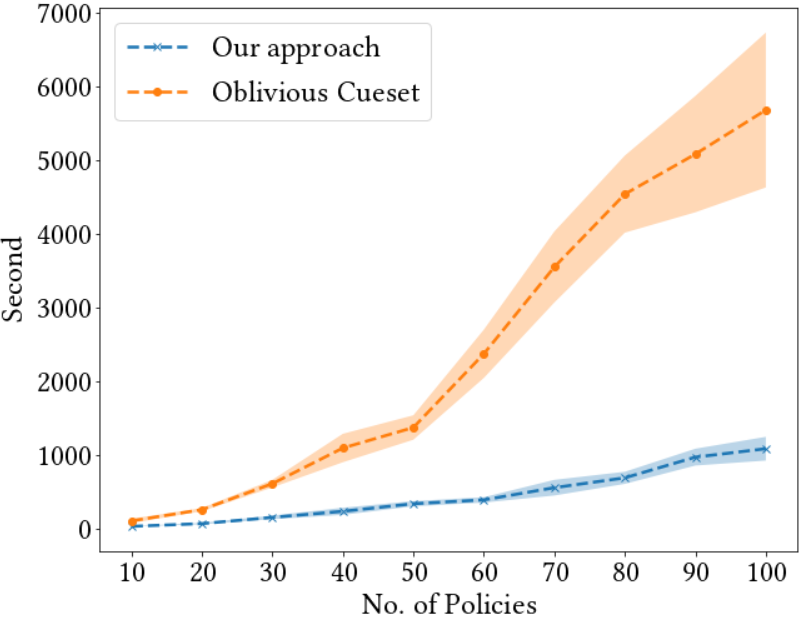}
\caption{(a) Data utility (b) Performance.  Experiments on Hospital10K dataset for \ourapproach, and \baselineTwo.}
\label{fig:e2e_hospital}
\end{figure}

\subsection{Experiment 1: Baseline Comparison}
\label{sect:baseline_comp}

We compare our approach against the aforementioned baselines and measure the utility as well as performance (see Figure \ref{fig:utility_efficiency_tax}(a)).
We increase the number of policies from 10 to 100 (step=10) where each sensitive cell participates in at least 5 dependencies.
This ensures that there are sufficient inference channels through which information about sensitive cells could be leaked.
The number of cells hidden by \ourapproach increases linearly  w.r.t the increase in number of policies/sensitive cells compared to \baselineOne (5.3$\times$\ourapproach) and \baselineTwo(1.4$\times$\ourapproach).
\baselineOne performs the worst because it randomly hides cells without checking the membership count of a cell in cuesets (as with using \textit{MVC} in Algorithm~\ref{alg:MVC}).
The performance of \baselineTwo is better because it uses the same inference protection strategy as \ourapproach. However, it generates a larger number of cuesets as it doesn't check the Tattle-Tale Condition for the dependency instantiations (like in Algorithm~\ref{alg:cuedetector})  and therefore has to hide more cells to ensure full deniability.

We also compare the performance (run time in seconds) against number of policies of these 3 approaches (see Figure\ref{fig:utility_efficiency_tax}(b)).
The run time of \ourapproach is almost linear w.r.t the increase of the number of policies.
On the other hand, \baselineTwo is exponential w.r.t number of policies, because it generates $\mathcal{O}(\vCount{\vSchemaDepSet{}} \times n^2)$ cuesets where $n$ denotes the number of tuples in $\vDatabase$ and it is expensive to run inference detection on such a large number of cuesets.
In \baselineOne, we restrict the execution to the fifth invocation of the inference detection algorithm (Algorithm \ref{alg:cuedetector}) i.e., if the execution doesn't complete by then, we force stop the execution. 
In order to study this further, we analyzed the total number of cuesets generated by \baselineOne vs.\ \ourapproach (see Figure \ref{fig:TE_tax}) in each invocation of Inference Detection. 
Due to the usage of MVC optimization in Inference Protection, \ourapproach terminates after a few rounds where as with \baselineOne the number of cuesets generated in each invocation keeps increasing.
We also note that \ourapproach is more stable in different test cases and has a lower standard deviation on number of cuesets and hidden cells compared to \baselineOne.

We show the supplementary evaluation results on the Hospital10K dataset.
Figure \ref{fig:e2e_hospital} presents the end-to-end comparison between \ourapproach and \baselineTwo, and supports our claim. 
\reviseTKDE{
In supplementary materials, we show experimental results with more sensitive cells (i.e., access control policies).
Interestingly if the access control view is highly sensitive (e.g., 10\% cells of the view are marked as NULL) and the sensitive cells are distributed over different columns, the sensitive cells can cancel out the channels leading to inference to each other. Therefore, in this case, our experimental results show that few additional cells are required to hide to achieve inference control. 
}

\begin{figure}[t]
\centering
\includegraphics[width=.5\linewidth]{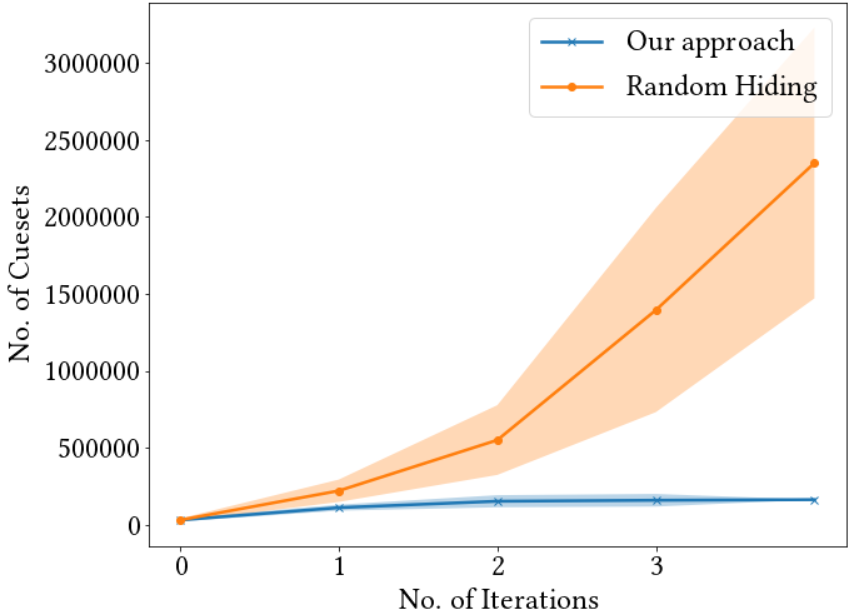}
\includegraphics[width=.47\linewidth]{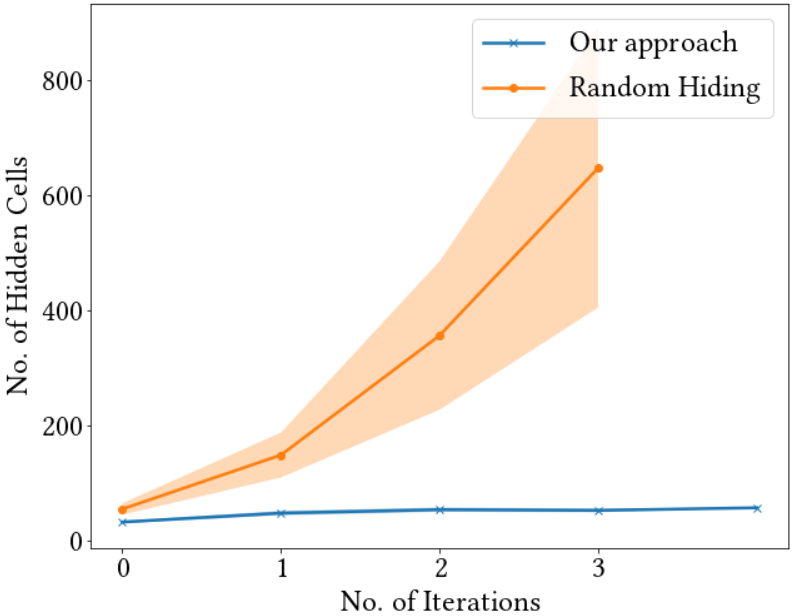}
\caption{(a) Number of cuesets generated in each invocation of Inference Detection (b) Number of cells hidden in each invocation of Inference Protection. Experiments run with 10 sensitive cells on Tax dataset.}
\label{fig:TE_tax}
\end{figure}

\subsection{Experiment 2: Dependency Connectivity}
\label{sect:dep_connectivity}

In the next set of experiments, we study the impact of dependency connectivity on the utility as well as performance.
The relationship between dependencies and attributes can be represented as a \emph{hypergraph} wherein the attributes are nodes and they are connected via data dependencies. 
We define the \emph{dependency connectivity} of a node, i.e., an attribute, in this graph based on the summation of the degree (number of edges incident on the node) as well as the degrees of all the nodes in its closure.
Using dependency connectivity, we categorize attributes on \textit{Tax} dataset into three groups: low, medium, and high where attributes in high, low, and medium groups have the highest, lowest, and average dependency connectivity respectively.
In Tax dataset, the high group contains 3 attributes (e.g. State), while the medium group has 3 attributes (e.g. Zip) and the low group includes 4 attributes (e.g. City).

The results (see Figure \ref{fig:sensitivity_group_tax}) show that when sensitive cells are selected from attributes with higher dependency connectivity, \ourapproach hides more cells than when selecting sensitive cells with lower dependency connectivity.
\extend{
The results are verified on both the Tax dataset and Hospital10K dataset (as shown in Figure \ref{fig:sensitivity_group_tax}(a) and Figure \ref{fig:sensitivity_group_tax}(b)).
This is because higher dependency connectivity leads to a larger number of dependency instantiations and therefore a larger number of cuesets from each of which at least one cell should be hidden.
Figure \ref{fig:sensitivity_group_hospital} demonstrates the evaluation among the dependency connectivity groups, on both datasets.
}

\begin{figure}[t]
\centering
\includegraphics[width=.47\linewidth]{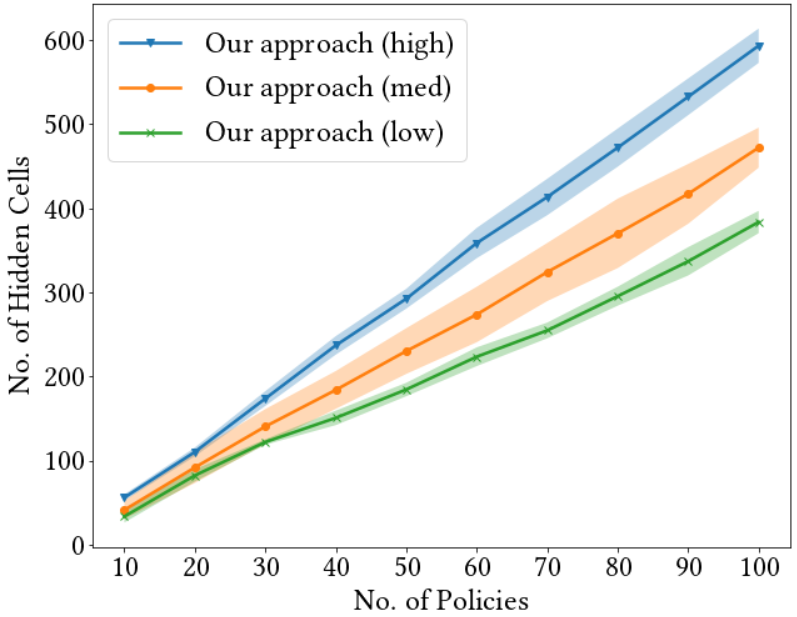}
\includegraphics[width=.47\linewidth]{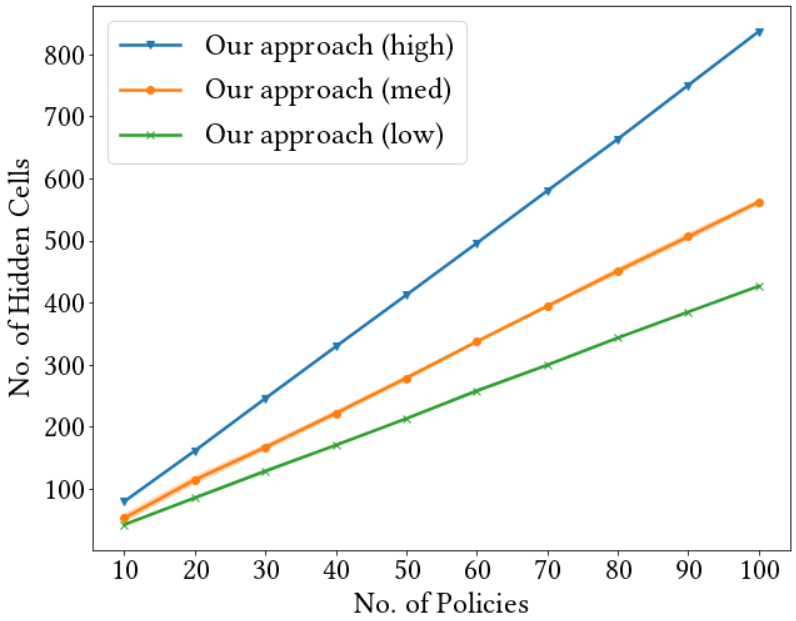}
\caption{Data utility experiments run with sensitive cells selected from (low, medium, high) dependency connectivity attributes in (a) Tax dataset (b) Hospital10K dataset.}
\label{fig:sensitivity_group_tax}
\end{figure}

\begin{figure}[t]
\centering
\includegraphics[width=.45\linewidth]{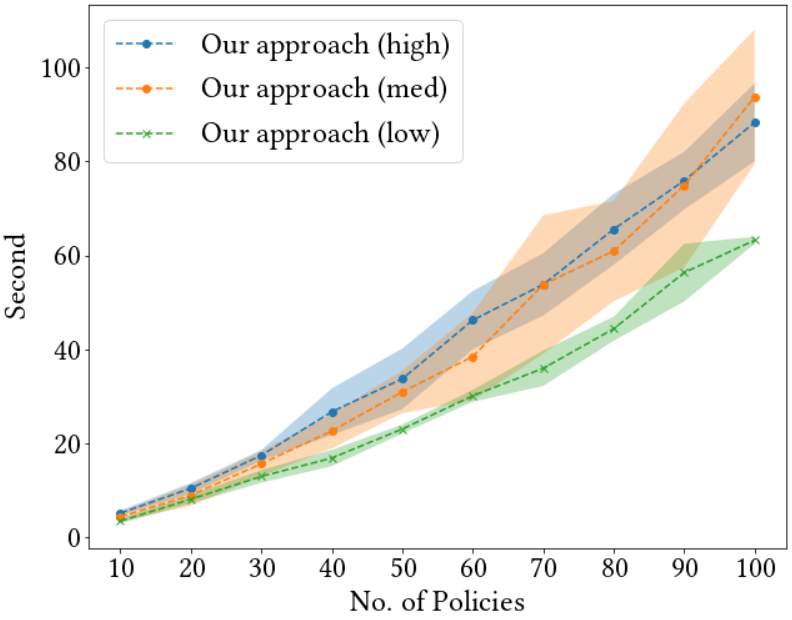}
\includegraphics[width=.48\linewidth]{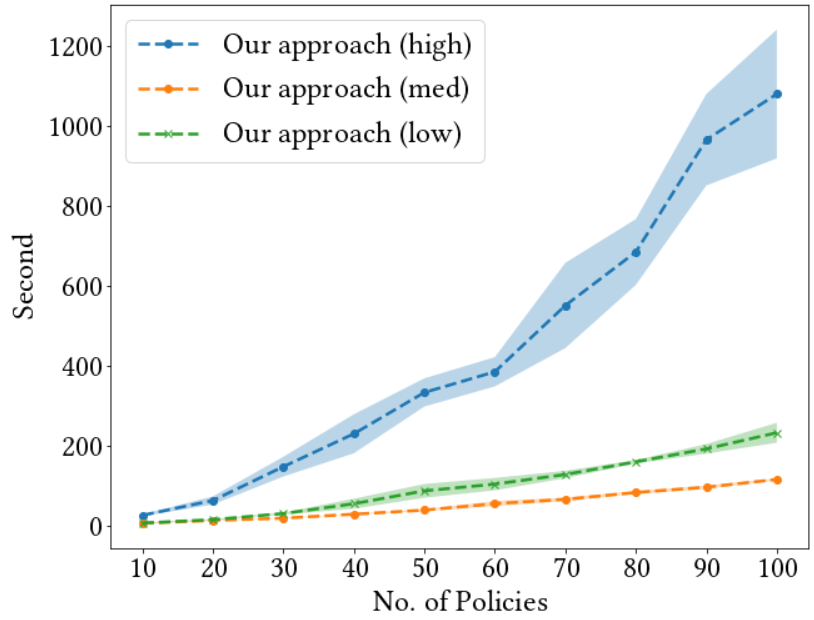}
\caption{Performance experiments run with sensitive cells selected from (low, medium, high) dependency connectivity attributes in (a) Tax dataset (b) Hospital10K dataset.}
\label{fig:sensitivity_group_hospital}
\end{figure}

\extend{
\subsection{Experiment 3: Scalability Experiments}
\label{subsec:scalability_exp}

The results of the scalability experiments are shown in Figure \ref{fig:scalablity}.
The $y$ axis records the time consumption while the $x$ axis denotes the size of the database (spanning from 10K tuples to 100K tuples).
We consider two different settings for selecting sensitive cells, 1) randomly sample a fixed number of sensitive cells regardless of the database size, and 2) incrementally sample a fixed ratio of sensitive cells w.r.t the database size.
The results of these two settings are presented in Figure \ref{fig:scalablity}(a) and Figure \ref{fig:scalablity}(b), resp.
In both cases, we set the bin size as 10K tuples and the merging size as 5.
In the first setting, the number of sensitive cells is set as 30 whereas, in the second setting, the ratio of sensitive cells to the total number of cells is 30 cells per 10K tuples.
We note that the starting point of the plot ($x =$ 10K tuples) corresponds to the experiments presented in Section \ref{sect:experiments} i.e., running our main algorithm on the dataset of size 10K (as there is only 1 bin).
As shown in Figure \ref{fig:scalablity}, the time consumption scales near-linearly (depending on the data itself) to the size of the datasets.
}

\begin{figure}[t]
\includegraphics[width=.5\linewidth]{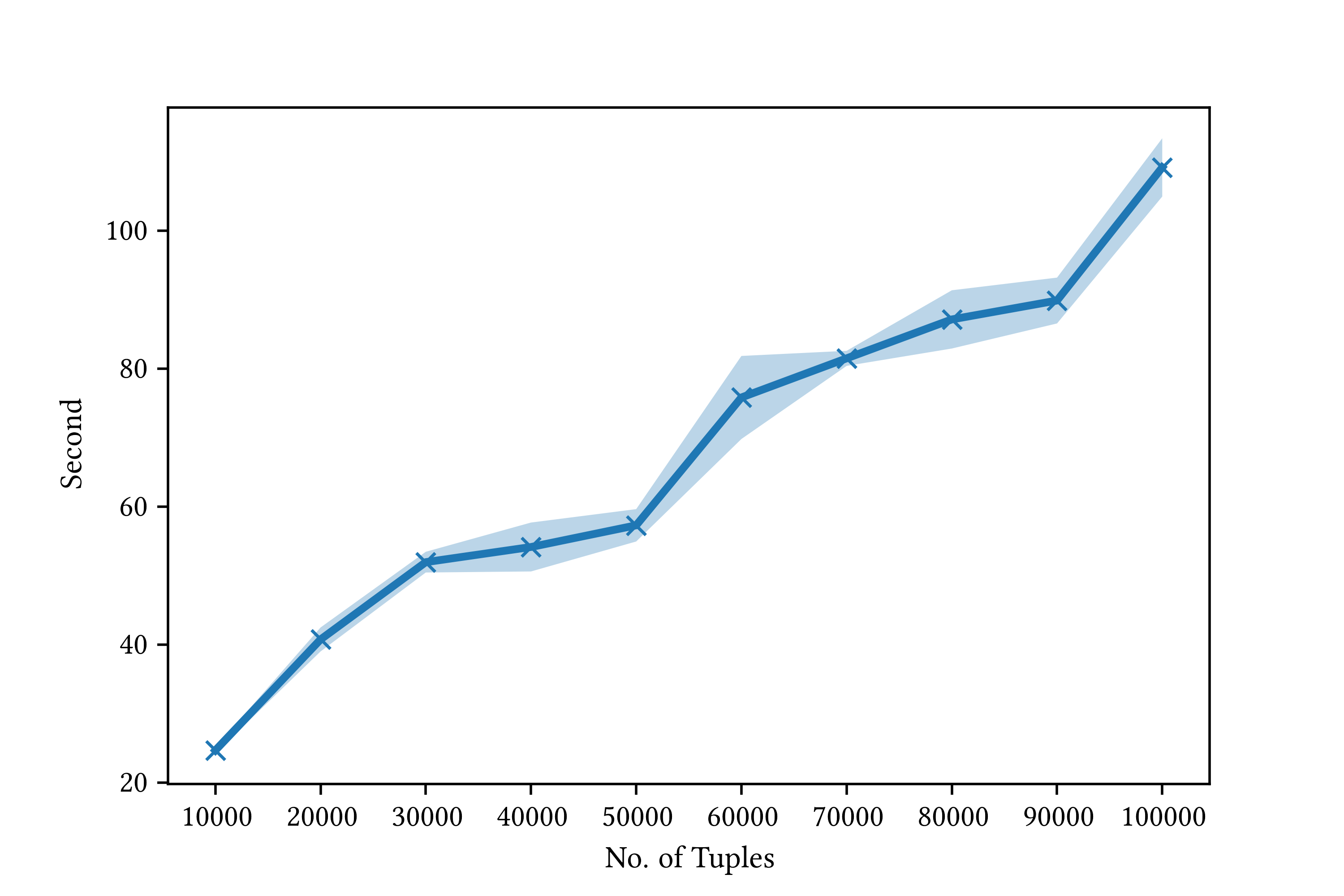}\hfill
\includegraphics[width=.5\linewidth]{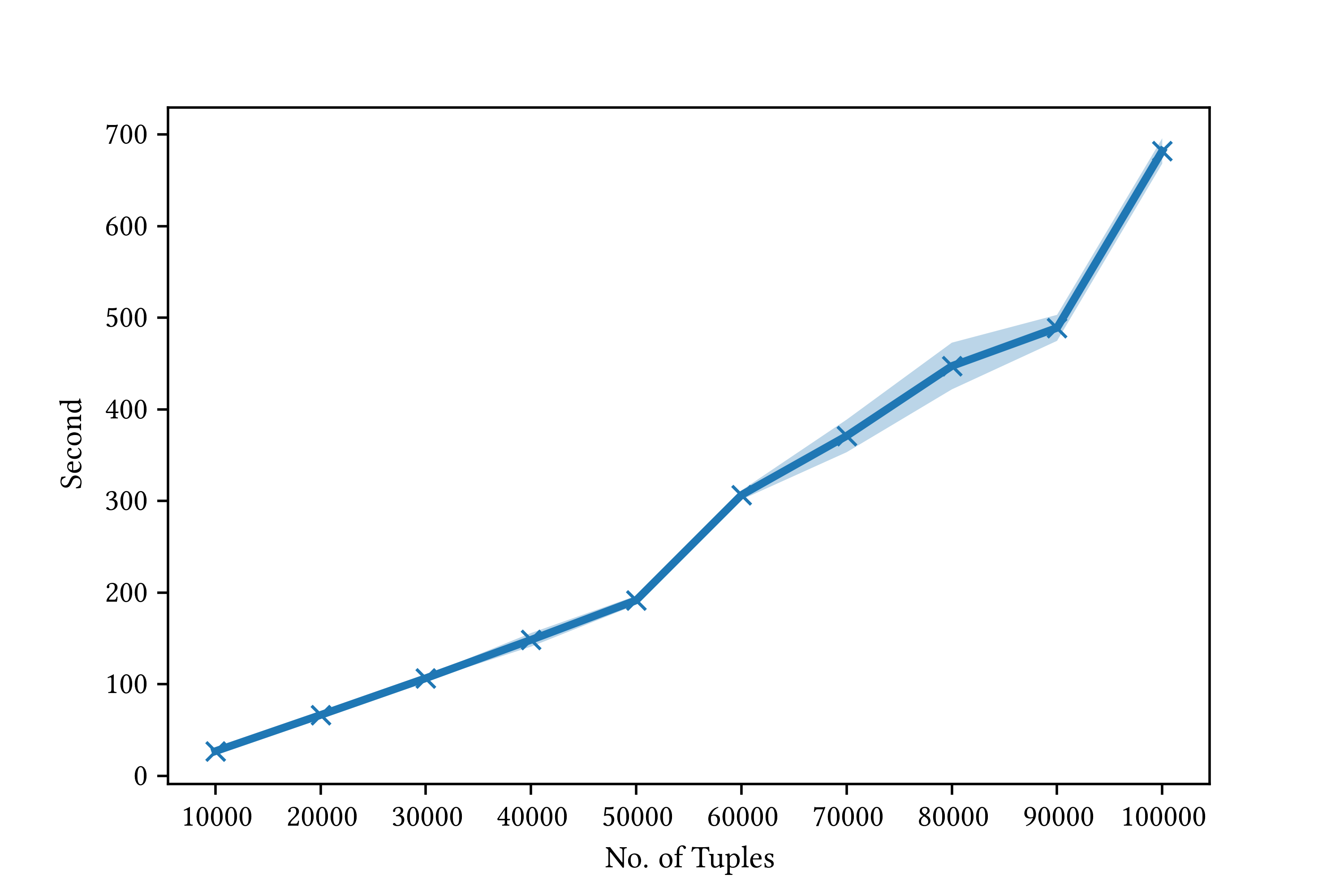}
\caption{\reviseA{
(a) The results for randomly sampling a fixed number of sensitive policies (b) The results for incremental sampling a fixed ratio of sensitive policies. Evaluation was done using the Binning-then-Merging Wrapper Algorithm on the Hospital dataset.}}
\label{fig:scalablity}
\end{figure}

\subsection{Experiment 4: \texorpdfstring{$k$}{k}-Percentile Deniability}
\label{sect:cleaning}

We implemented our system with a relaxed notion of security, $k$-percentile deniability, where $k$ is a relative parameter based on the domain size of the sensitive cell.
We analyze the utility of the system when varying $k$ and measure the utility.
For the results shown in Figure \ref{fig:k-den}(a), the sensitive cell is selected from ``State'' which is a discrete attribute with high dependency connectivity.
Clearly, when $k=0$, i.e., full leakage, the unconstrained case will only hide sensitive cells and when $k=1$ i.e, full deniability, the system hides the maximum number of cells.
When $k=0.5$, i.e., the inferred set of values is half of that of the base view, the system hides almost the same number of cells as $k=1$ i.e., full deniability.
When $k=0.1$, i.e, the inferred set of values is $\frac{1}{10}$ of that of the base view, 
our system hides $\approx15\%$ fewer cells than the one that implements full deniability.
On the \emph{Hospital} dataset, the utility improvement was marginal with k set to the smallest value possible (besides full leakage) i.e., $k=\frac{1}{\vCount{\vDomain{\vCell{*}{}}}}$.
The approach that implements full deniability is able to provide high utility with a stronger security model on both datasets compared to the one that implements $k$-percentile deniability.
\reviseTKDE{
We measure the runtime performance of $k$-deniability for different $k$ values and compare the results with full-deniability.
As shown in Figure \ref{fig:k-den}(b), algorithms to achieve $k$-deniability take longer time to complete than the full-deniability algorithms, because $k$-deniability algorithms reduce the fan-out of the cuesets in the first iteration, but more iterations are thus taken to converge.
For different tested $k$ values, the more we relax the $k$ constraint, the less execution time the algorithm will take, because fewer cuesets, thus a smaller fan-out, are considered in calculating leakage.
}

\begin{figure}[t]
\centering
\includegraphics[width=0.43\linewidth]{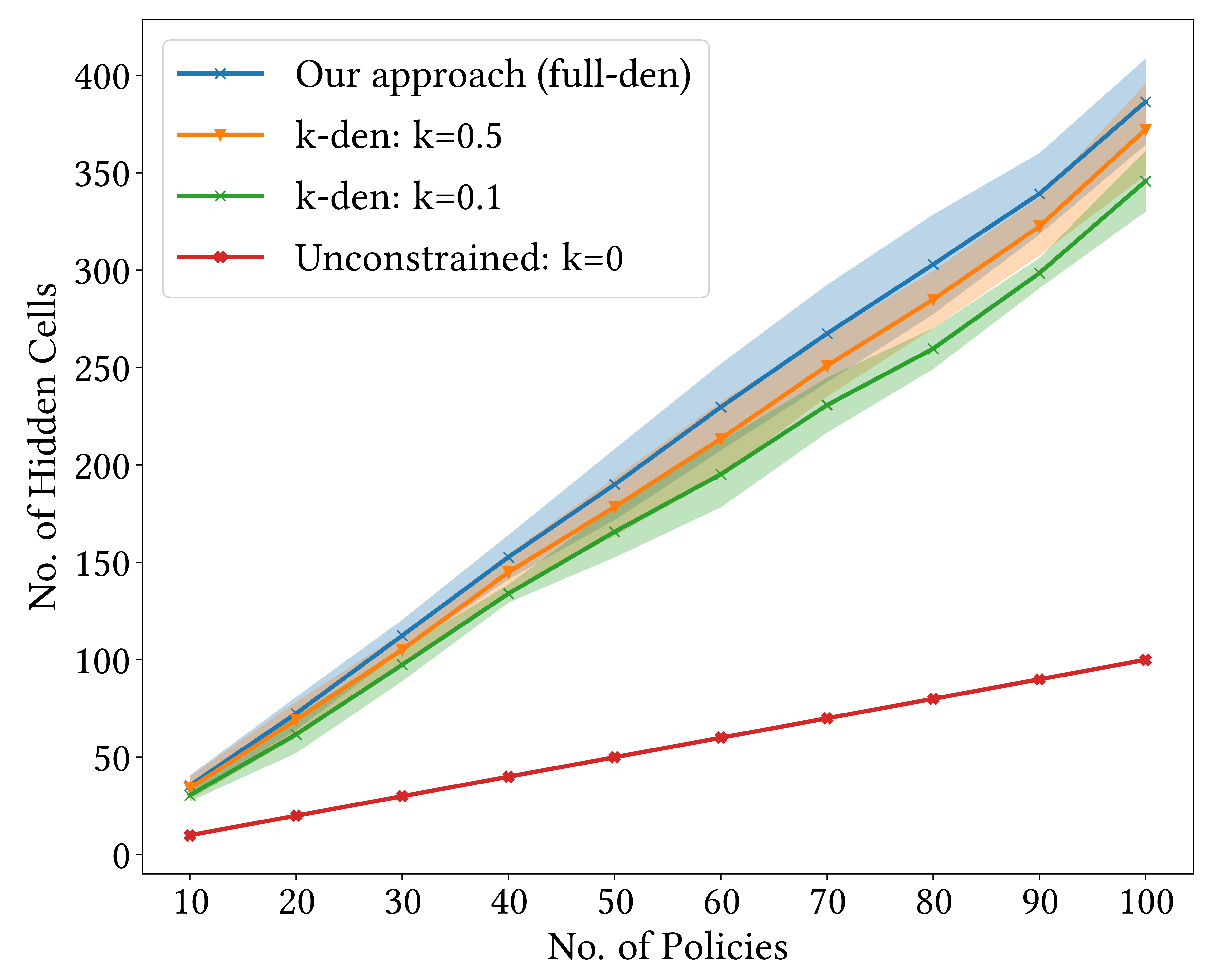}
\includegraphics[width=0.4\linewidth]{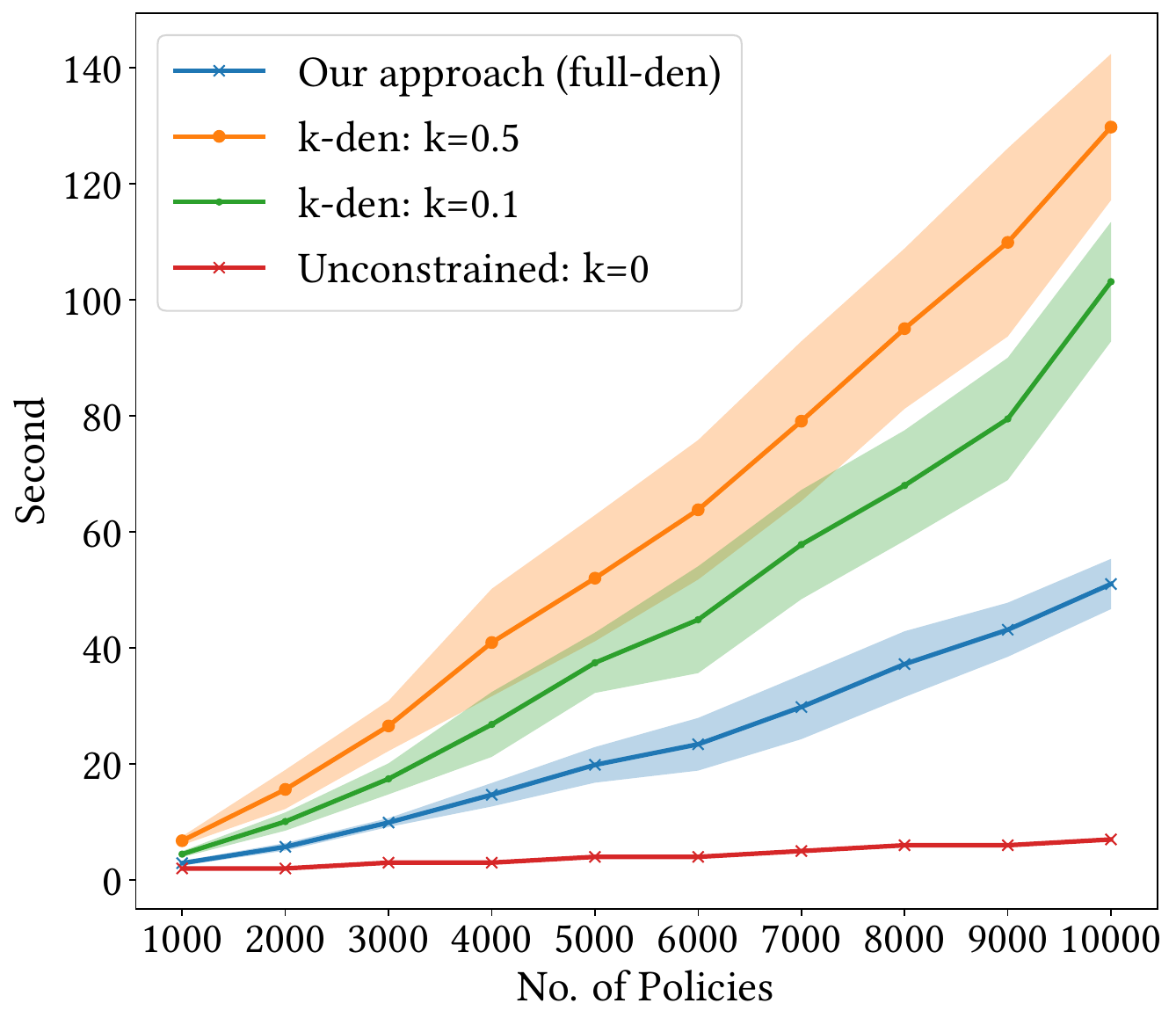}
\caption{(a) Data utility on Tax dataset. Experiments done with full deniability and k-deniability (varying values of k); \reviseTKDE{(b) Performance on Tax dataset (varying values of k).}}
\label{fig:k-den}
\end{figure}

\subsection{\reviseTKDE{Experiment 5: Modified Inference Protection}}
\label{subsec:modified_ic_exp}

\reviseTKDE{
We compare the modified inference protection algorithm (Algorithm \ref{alg:tuple_hiding}) on the Tax dataset against the original inference protection (Algorithm \ref{alg:MVC}) to achieve full-deniability.
As shown in Figure \ref{fig:modified_TT}, the price of relaxing the assumption comes at the cost of lower utility (up to 1.3x cells hidden) and efficiency (2-3x seconds taken to converge, with a non-linear growth). 
}

\begin{figure}[t]
    \centering
    \includegraphics[width=0.23\textwidth]{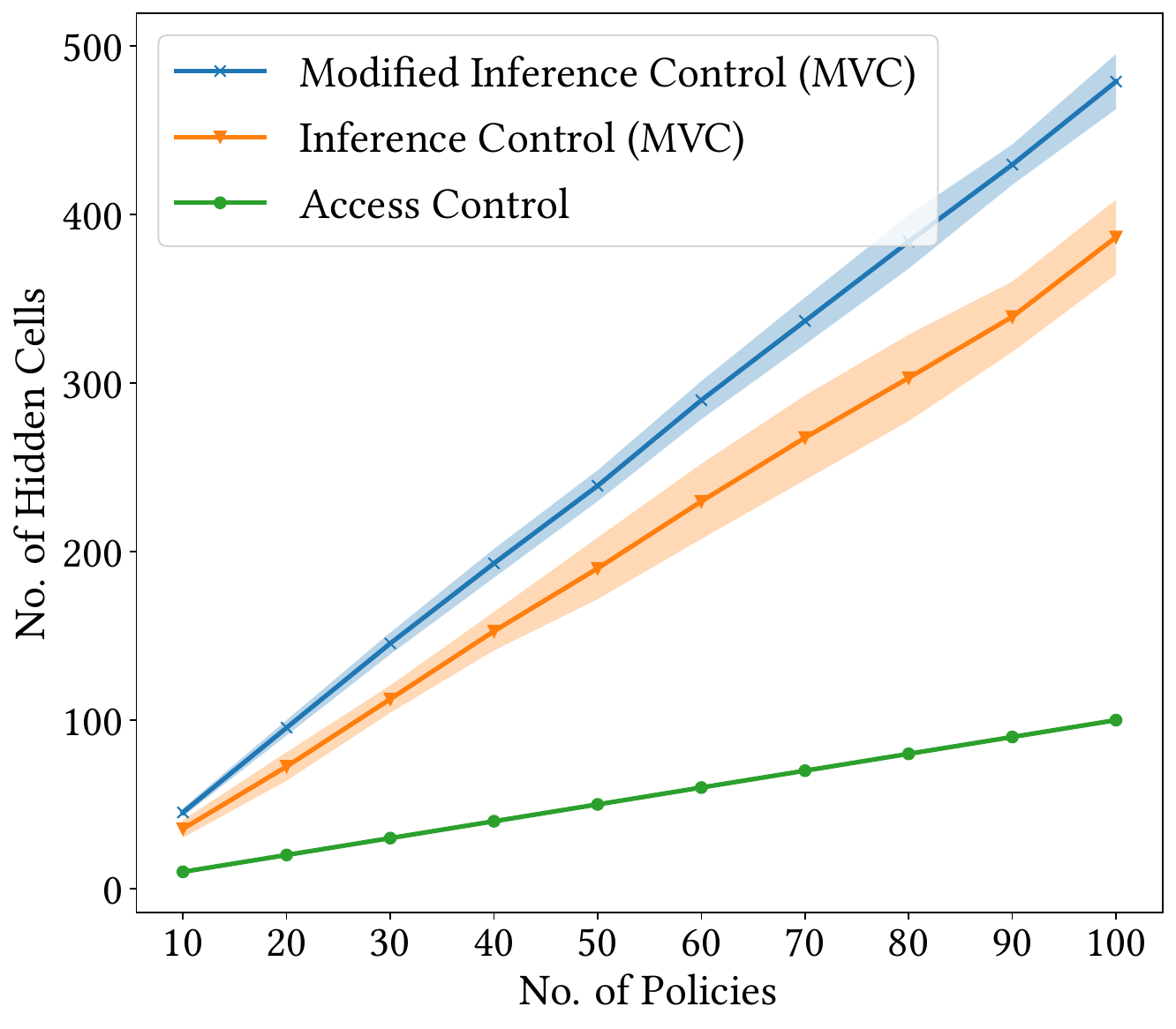}
    \includegraphics[width=0.23\textwidth]{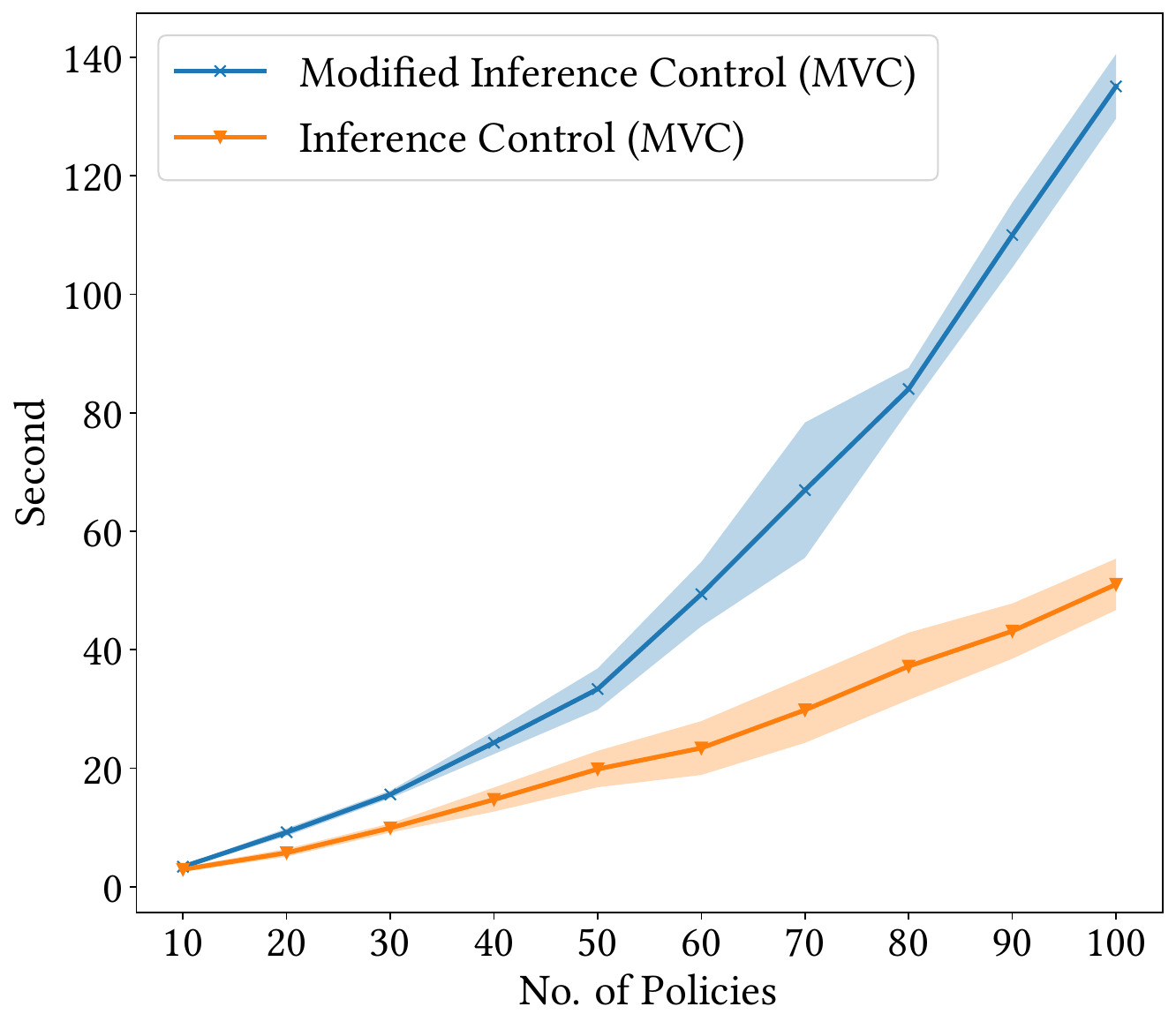}
    \caption{\reviseTKDE{Modified Inference Protection: (a) Data utility (b) Performance on Tax dataset for modified inference protection, inference protection, and access control.}}
    \label{fig:modified_TT}
\end{figure}

\subsection{\reviseTKDE{Experiment 6: Case Study over Query Workloads}}
\label{subsec:query_case_study}

\reviseTKDE{
We further study how inference control algorithms can affect the utilities of query workloads, especially when a large portion of the database view is marked as NULL by access control policies.
We first investigate the distribution of the hidden cells (NULL's) across the views.
We take the run with the access control view with 1,000 policies and execute our approach (w. MVC) to generate the inference control view and use these two views throughout this case study.
Since this study involves a large number of sensitive cells, the baseline methods presented earlier, time out before converging and we only compare the inference control view based on our approach with the access control view.
We present in Figure \ref{fig:heatmaps} the heatmap, where a darker color represents more cells hidden in this column, and the density distributions of data that support the visualization.
The distributions of NULL cells are similar in both views -- most additional hidden cells in the inference control view are concentrated in the first 3 attributes that are directly correlated with the access control policies.
Some but fewer additional cells from other columns are hidden as well in the inference control view, while none of the cells are hidden from the attributes not participating in dependencies.

\stitle{Evaluating workload-driven utility metric.} Next, we evaluate the utility of the database views over two types of query workloads: randomized range queries over one column and cross columns.
In particular, for the first case, we randomly generate 1,000 set queries per column with randomly sampled range specifications (w. 300-1100 cells, varying).
For the cross-column queries, we consider every possible pairwise combination of the attributes and similarly generate 1,000 queries for each combination. The range queries cover both attributes in each combination. 
As mentioned, we consider visibility (i.e., percentage of non-NULL cells in the query result) as the utility metric in this case study.

Figure \ref{fig:workload_case_study} shows the empirical results.
We take the workload that executes 1,000 queries on the ``Rate'' column to present the results in Figure \ref{fig:workload_case_study}(a) and (b) for access control and inference control views, resp.
We use histograms to show the number of queries that has a certain percentage of visibility.
As observed, most queries remain high visibility ($\sim$93-96\% cells visible) in both views, indicating good utility for downstream analytics.

We then present results for cross-column queries in Figure \ref{fig:workload_case_study}(c) and (d) as heatmaps.
Each block in the heatmap represents the average visibility percentage among 1,000 queries executed over this attribute combination. While the overall visibility is over 95\% for both views, a darker color in the heatmap suggests more cells are visible from the query.
The similarity between the two heatmaps indicates that inference control does not affect the query-driven utility much compared to the access control views.
}

\begin{figure}[t]
    \centering
    \includegraphics[width=0.47\linewidth]{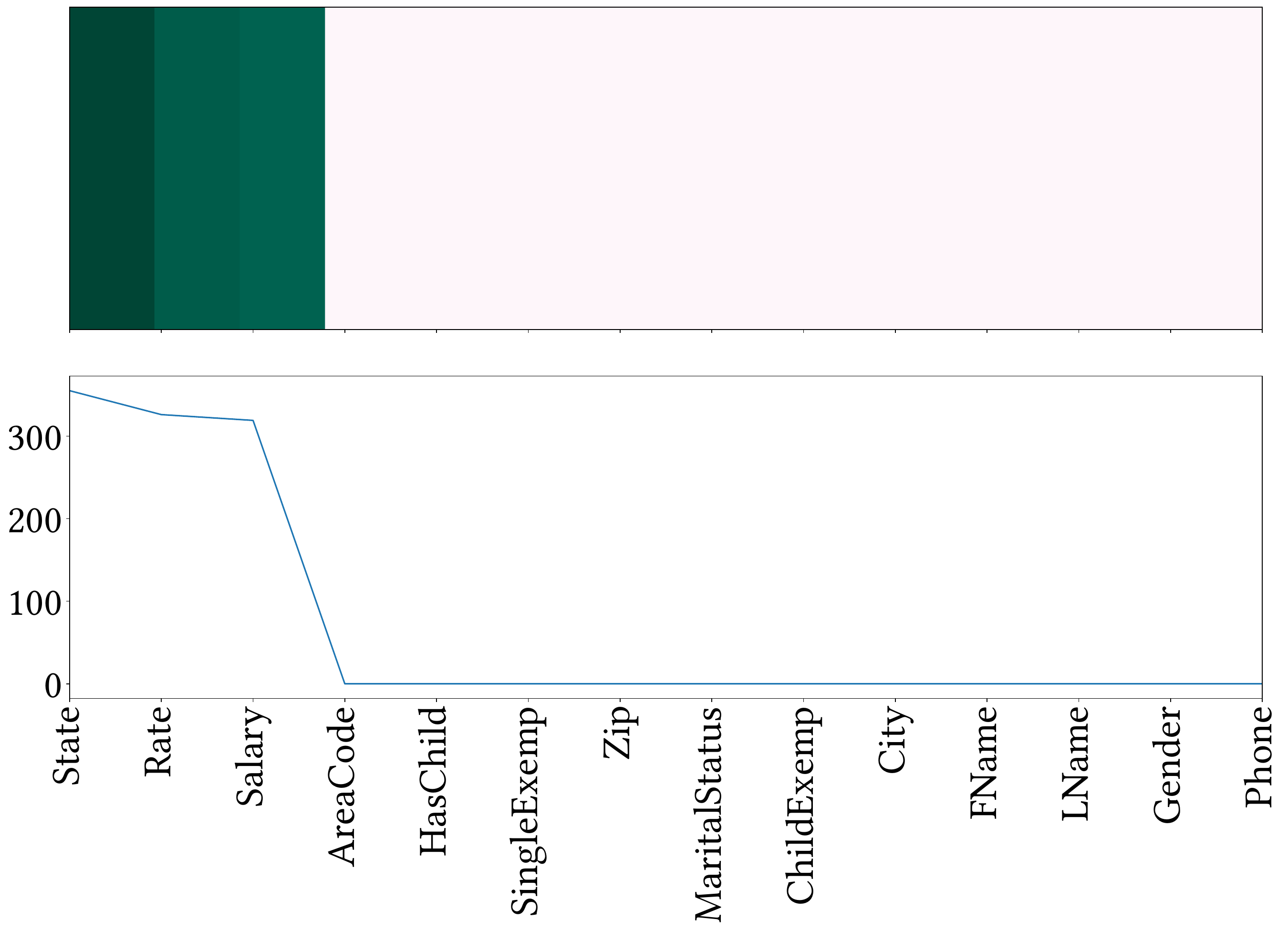}
    \includegraphics[width=0.47\linewidth]{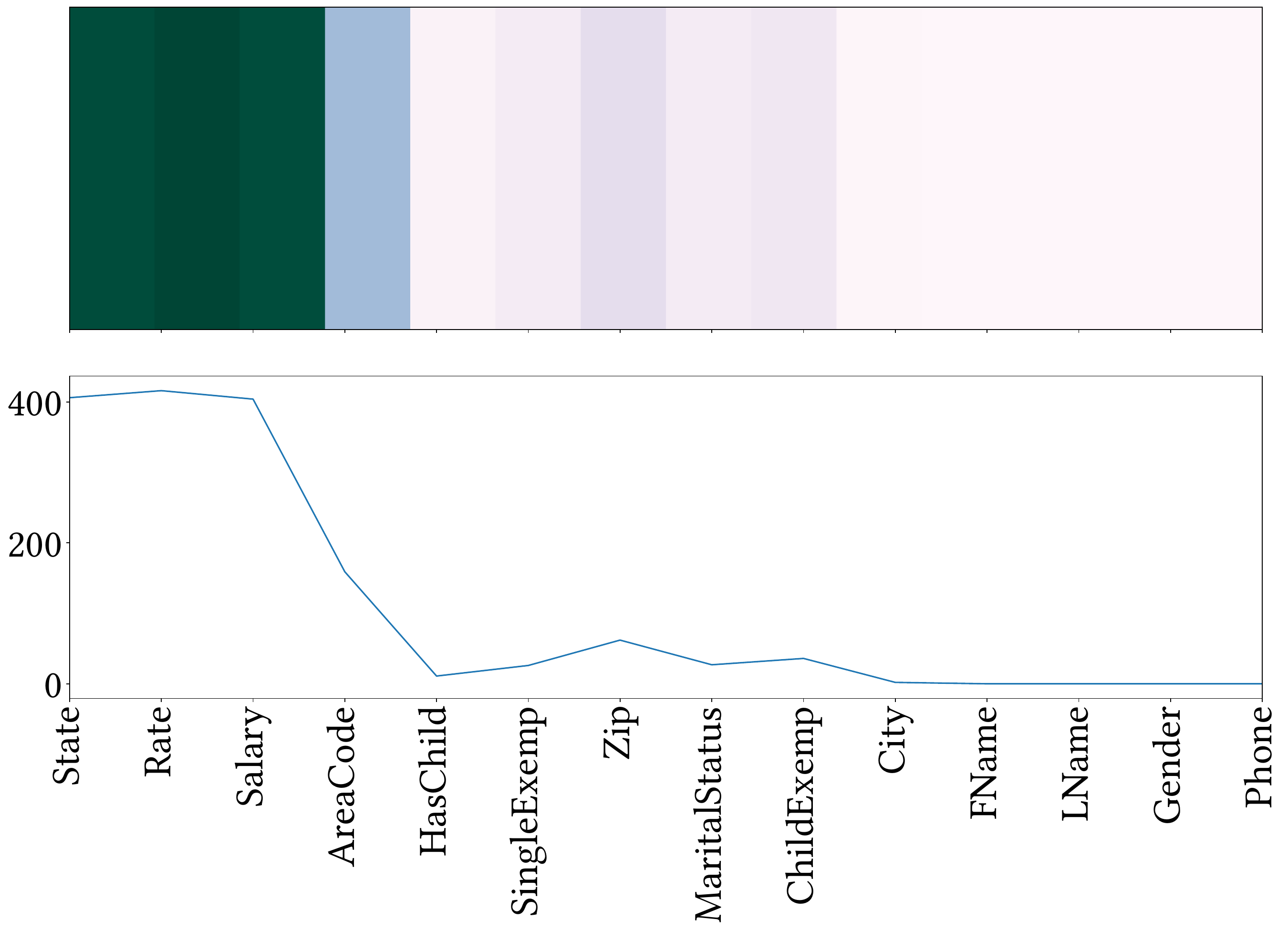}
    \caption{\reviseTKDE{Distribution of NULL's: (a) as policies in access control view; (b) as hidden cells in the inference control view.}}
    \label{fig:heatmaps}
\end{figure}

\begin{figure*}[t]
    \centering
    \includegraphics[width=0.33\linewidth]{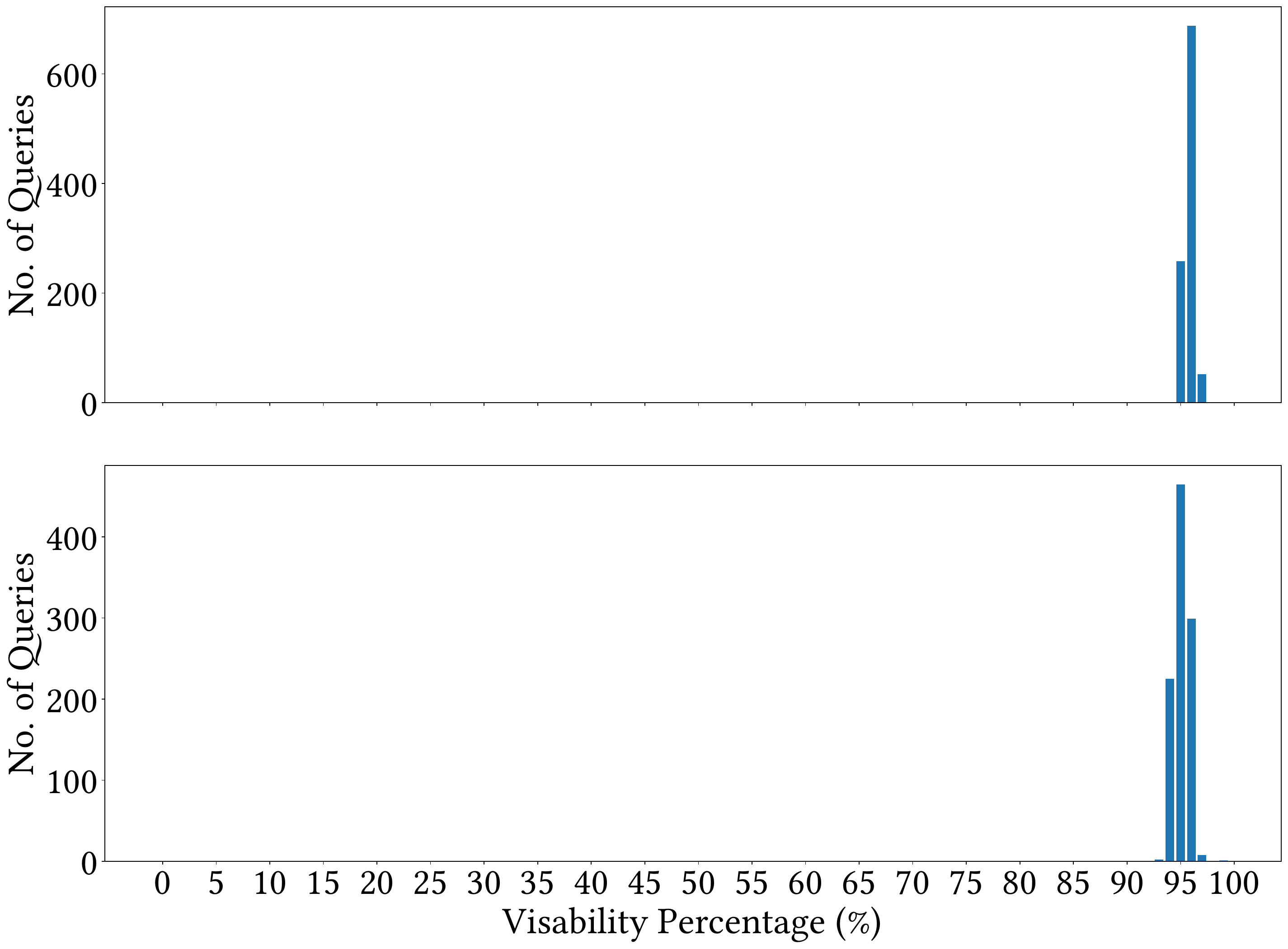}
    \includegraphics[width=0.3\linewidth]{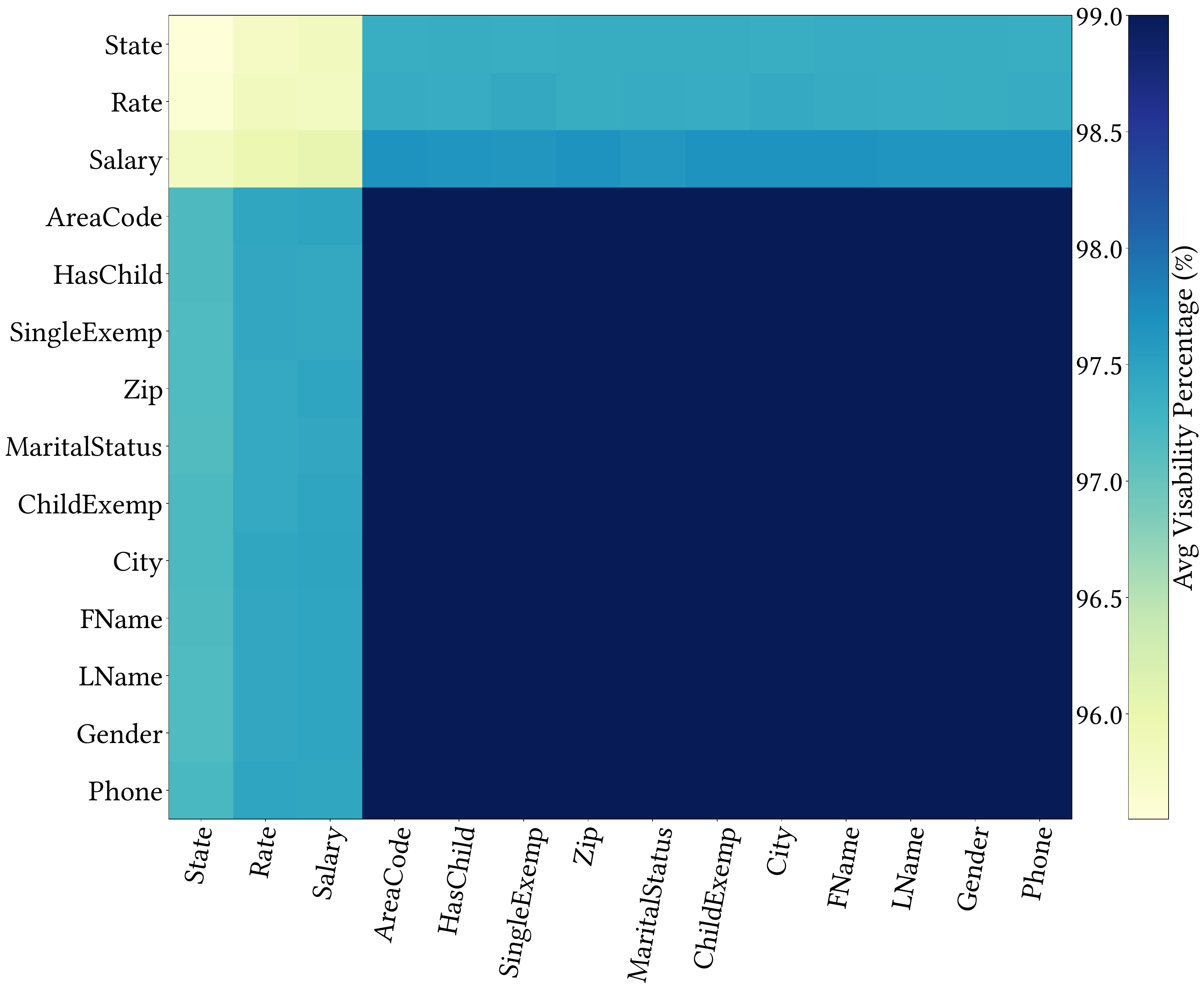}
    \includegraphics[width=0.3\linewidth]{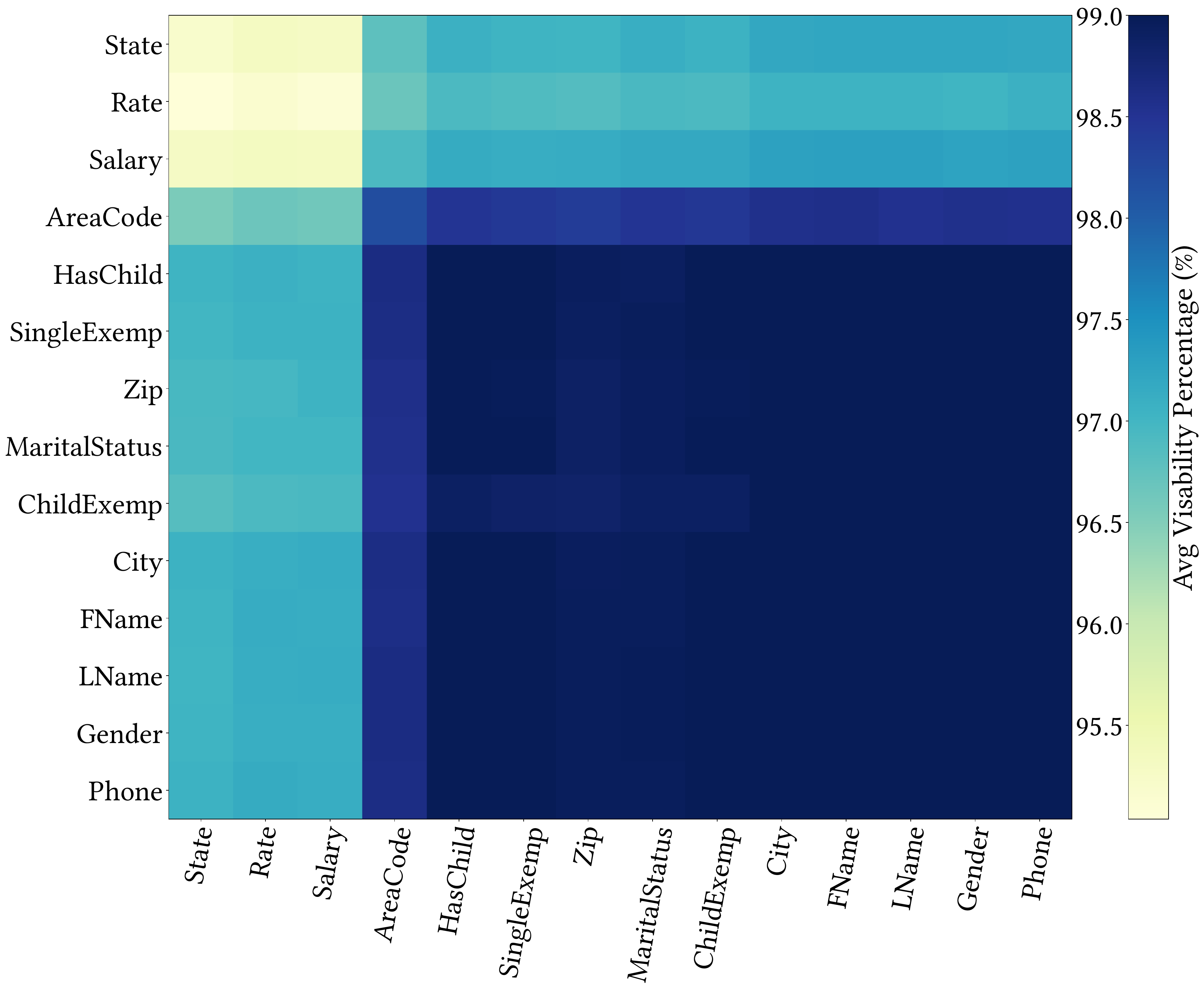}
    \caption{\reviseTKDE{Workload-driven utility: (a) Upper left: visibility percentage for queries in the workload over the access control view; (b) Bottom left: visibility percentage for queries in the workload over the inference control view; (c) Middle: average visibility percentage in cross-column workload over the access control view; (d) Right: average visibility percentage in cross-column workload over the inference control view.}}
    \label{fig:workload_case_study}
\end{figure*}

\subsection{\reviseTKDE{Experiment 7: Case Study against Real-World Adversaries}}

\begin{figure}[t]
    \centering
    \includegraphics[width=.7\linewidth]{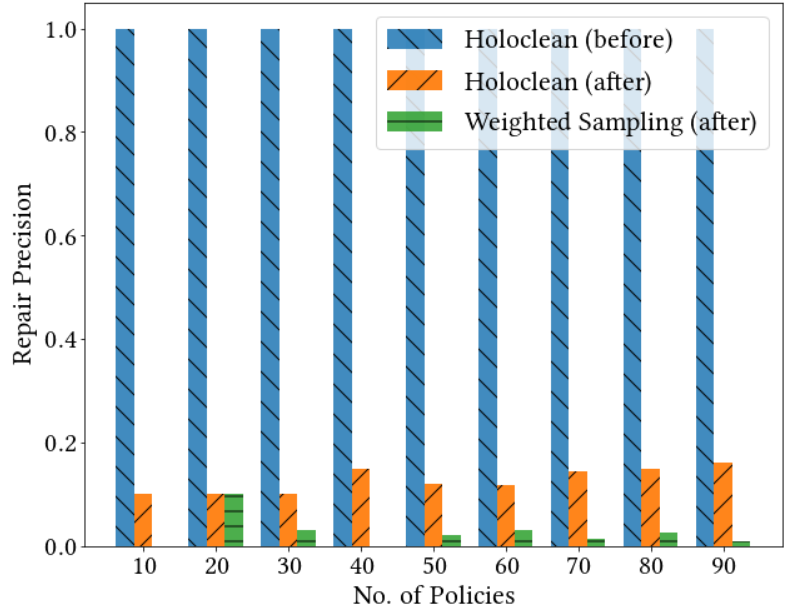}
    \caption{Against real-world adversaries: Reconstruction precision of sensitive cells with two types of adversaries.}
    \label{fig:holoclean}
\end{figure}

A potential limitation of our security model is based on the assumption that no correlations exist between attributes and tuples i.e., they are independently distributed other than what is explicitly stated through dependencies 
(that is either learnt automatically or
specified by the expert). 
However, typically in databases, other correlations do exist which can be exploited to infer the values of the hidden cells. These correlations can be also learned by the database designer using dependency discovery tools or data analysis tools.
If the correlations are very strong (e.g. hard constraints with no violations in the database), we call them out as constraints and consider them in our algorithms. For weak correlations, or soft constraints that only apply to a portion of the data, we do not consider them. Otherwise, everything in the database will become dependent, in which case our algorithm would be too conservative and hide more cells than necessary based on these soft constraints.

Therefore, we study the effectiveness of \ourapproach against inference attacks, i.e.,  to what extent can an adversary reconstruct the sensitive cells in a given querier view.
We consider two types of adversaries. The first type of adversary uses \textit{weighted sampling} where for each sensitive cell $\vCell{*}{}$, the adversary learns the distribution of values in $\vDomain{\vCell{*}{}}$ by looking at the values of other cells in the view.
The querier, then tries to infer the sensitive cell value by sampling from this learned distribution. 
The second type of adversary utilizes a state-of-the-art data cleaning system, Holoclean~\cite{rekatsinas2017holoclean}, which compiles data dependencies, domain value frequency, and attribute co-occurrence and uses them into training a machine learning classifier.
The adversary then leverages this classifier to determine values of sensitive cells by considering them as missing data in the database.
The sensitive cell for this experiment is selected from ``State'' which is a discrete attribute with high dependency connectivity. 
We consider the $10$ dependencies and drop the FC because Holoclean doesn't support it.
We increase the number of policies from 10 to 90 and input the querier view (in which the values of hidden cells are replaced with \nullvalue) to both adversaries.
We measure the effectiveness by $\text{repair precision} = \dfrac{\# \text{correct repairs}}{\# \text{total repairs}}$ (where a \textit{repair} is an adversary's guess of the value of a hidden cell) and therefore lower the \textit{repair precision} of the adversary is, the more effective \ourapproach is.

The results ``Holoclean (before)'' in Figure \ref{fig:holoclean}
show that when only sensitive cells are hidden, an adversary such as Holoclean, is able to correctly infer the sensitive cells.
When additional cells are hidden by \ourapproach, indicated by ``Holoclean (after)'', the maximum precision of Holoclean is $0.15$. 
On the other hand, the weighted sampling employed by the other type of adversary, indicated by ``Weighted Sampling (after)'', could reconstruct between 3\% and 10\% of the sensitive cells. 
Note that Holoclean uses the learned data correlations (and attribute co-occurrence, domain value frequency) in addition to the explicitly stated data dependencies. However, it only marginally improves upon weighted sampling given the view generated by \ourapproach.

\eat{
\begin{itemize}
    \item Utility plots: 
    each plot: no. of hidden cells v.s no. of policies: full-den, full-den-woMVC, full-den-oblivious: mean + std \\
    tax plot, hospital plot 
    \item Performance plots:
    each plot: runtime vs. no. of policies: three lines\\
    tax and hospital\\
    two more plots for fan-out of cueset and hidden cells
    \item Privacy tables: 
    holoclean-constraint-only, percentage of sensitive cells recovered v.s. no of polices: 
    each table: full-den; 
    baseline (hide no additional cells); 
    baseline (hide similar no. of cells randomly);
    baseline (full-den-oblivious) \\
    holoclean-all features for the same set of algos     
    for both tax and hospital 
\end{itemize}
}

\section{Related Work}
\label{sect:related}

The challenge of preventing leakage of sensitive data from query answers 
has been studied in many prior works on inference control~\cite{farkas2002inference}. 
Early work by Denning~\etal~\cite{denning1985commutative} designed 
commutative filters to ensure answers returned by a query 
are equivalent to that which would be returned based on the authorized view for the user.
This work, however, did not consider
 data dependencies.  
\reviseTKDE{We categorize them based on when and how inference control is applied and what security model is used.}

\reviseTKDE{\stitle{Design-time Prevention Methods}
which mark attributes that lead to inferences on sensitive data items as sensitive.}
Qian \etal~\cite{QianInference} developed a tool to analyze potential leakage due to foreign keys in  order to elevate the clearance level of data if such leakage is detected. 
Delugachi \etal     ~\cite{DelugachHinkeInf} generalized 
the work in  \cite{QianInference} and developed an approach based on analyzing a conceptual graph to identify potential leakage from  more general types of data associations (e.g., part-of, is-a). 
Later works such as \cite{yip1998data}, however, established that inference rules for detecting inferences at database design time are incomplete and hence are not a viable approach for preventing leakage from query answers. 
\reviseTKDE{Design time approaches for disclosure control have successfully been used in restricted settings such as identifying
the maximal set of non-sensitive data to outsource such that it prevents inferences about sensitive data \cite{vimercati2014outsourcing, haddad2014access,oktay2015semrod, oktay2017secure}, however, do not extend to our setting.
}

\reviseTKDE{\stitle{Query-time Prevention Methods}
that reject queries which lead to inferences on sensitive data items.}
Thuraisingham \cite{THURAISINGHAM1987479} developed a \emph{query control approach} in
the context of 
 Mandatory Access Control (MAC) wherein 
 policies specify the security clearances for the users (subject) and the security classification/label for the data.
 ~\cite{THURAISINGHAM1987479} presented an inference engine to determine if query answers can lead to leakage (in which 
case the query is rejected). 
While \cite{THURAISINGHAM1987479} assumed a prior existence of an inference detection engine, Brodsky \etal \cite{brodsky2000secure} developed a framework, DiMon,  based on chase algorithm for constraints expressed as Horn clauses.
DiMon takes in current query results, the user's query history, and Horn clause constraints to determine the additional data that may be inferred by the subject. Similar to \cite{THURAISINGHAM1987479}, 
if inferred data is beyond the security clearance of the subject then their system refuses the query. 
\reviseTKDE{Such work (that identifies if a query leaks/does not leak data) differs from ours since it cannot be used  directly 
to identify  a maximal
secure answer  that does not lead to any inferences  --- the problem we study in this paper.
Also, the above work on query control is based on a much weaker security model  compared to the full-deniability model we use. It only prevents 
 an adversary from reconstructing the exact value of a sensitive cell but cannot prevent them from learning new information about the sensitive cell.
 }

\reviseTKDE{\stitle{Perfect Secrecy Models}
that characterizes inferences on any possible database instance as leakage.}
The most relevant of these works is from
 Miklau \& Suciu \cite{miklau2007formal} who study the challenge of preventing information disclosure for a secret query given a set of views.
 Our problem setting is different as we check for a given database instance whether it is possible to answer the query hiding as few cells as possible while ensuring full deniability.
\reviseTKDE{Applying their approach to our problem setting will
be extremely pessimistic as most queries will be rejected on a database with a non-trivial number of dependencies.  
}

\reviseTKDE{\stitle{Randomized Algorithms for Inference Prevention} that suppress too many cells and does not look at dependencies as inference channels}
The most relevant of these are Differential Privacy (DP) mechanisms promise to protect against an adversary with any prior knowledge and thus have wide applications nowadays \cite{DBLP:journals/fttcs/DworkR14,yu2022thwarting,zhang2023DProvDB}. 
In our problem setting of access control, called the \textit{Truman model} of access control \cite{rizvi2004extending}, the data is either hidden or shared depending upon whether it is sensitive for a given querier.
In such a model, the expectation of a querier is that the result doesn't include any randomized answers. 
Weaker notions of DP such as One-sided differential privacy (OSDP) 
 \cite{kotsogiannis2020one} aims to prevent inferences on sensitive data by using a randomized mechanism when sharing non-sensitive data. 
  However, such techniques offer
  only probabilistic guarantees (and cannot implement
  security guarantees such as full deniability), and therefore may allow some non-sensitive data to be released even when their values could lead to leakage of a sensitive cell.
  These techniques also lead to suppression of a large amount of data (suppresses  approx. 91\% non-sensitive data at $\epsilon = 0.1$ and approx. 37\% at $\epsilon = 1$). 
The current model of OSDP only supports hiding at the row level and is designed for scenarios where the whole tuple is sensitive or not. It is non-trivial to extend to suppress cells with fine-grained access control policies considered in our setting.
Furthermore, \reviseTKDE{most DP-based mechanisms (including OSDP) assume that no tuple correlations exist even through explicitly stated data dependencies.}

\reviseTKDE{\stitle{Inference Control in Other Access Control Settings.}
Among these, \cite{JEBALI20221} studies the problem of secure data outsourcing in the presence of functional dependencies.
Access control policies are modelled using confidentiality constraints which define what combination of attributes should not appear together in a partition. 
They use a graph-based approach built upon on functional dependencies to detect possible inference channels. 
The goal is to then derive optimal partitioning so as to prevent inferences through these functional dependencies while efficiently answering queries on distributed partitions. 
Vimercati et al~\cite{vimercati2014outsourcing} also studied the problem of improper leakage
due to data dependencies in data fragmentation. Similar to \cite{JEBALI20221}, they mark attributes as sensitive (using
confidentiality constraints) and block the information flow from non-sensitive attributes to
sensitive attributes through dependencies. 
In general, the works in this category look at sensitivity at the level of attributes and not at the level of cells through fine-grained access control policies, studied in our work. In our work, we enforce fine-grained access control policies and allow minimal hiding of additional cells to prevent inferences.

}

\section{Conclusions and Future Work}
\label{sect:conclusions}

We studied the inference attacks on access control protected data through data dependencies, DCs and FCs. We developed a new stronger security model called \emph{full deniability} which prevents a querier from learning about sensitive cells through data dependencies. We presented conditions for determining leakage on sensitive cells and developed algorithms that uses these conditions to implement full deniability. The experiments show that we are able to achieve full deniability for a querier view without significant loss of utility for two different datasets.

In future, extending the security model to not only consider hard constraints explicitly specified in the form of denial constraints but also soft constraints that exist as correlations between data items poses a significant challenge. 
The invertibility model in FCs could also be expanded to model the probabilistic relationship between input and output cells, replacing the current deterministic model.
In addition to considering non-binary leakage as in \emph{k-percentile deniability}, one could release non-sensitive values randomly  (like in OSDP~\cite{kotsogiannis2020one}) instead of hiding all of them to prevent leakage. However, this requires addressing the challenges of any inadvertent leakages through dependencies when sharing such randomized data and also maintaining the validity of the database w.r.t dependencies.
We also envision our approach to preventing inference control being relevant to other areas of access control research (such as cryptographic models~\cite{HCA2016}, and web applications~\cite{blockaid2022}) and applications (such as Internet of Vehicles~\cite{ABE}).

\eat{In future, we would like to extend the security model to not only consider hard constraints explicitly specified in the form of data dependencies but also soft constraints that exist as correlations between data items. 
The invertibility model in FCs could also be extended to model probabilistic relationship between input and output cells, instead of being deterministic as in the current model. 
Improving utility while implementing full deniability is also an open challenge.
In $k$-percentile deniability, the improvement in utility as a factor of $k$ needs to be studied further as factor of different properties of the dataset such as type of attributes, dependency connectivity, dependency instantiations, and possible number of cuesets for a given sensitive cell.
Along with further exploration of $k$-percentile deniability considered in our paper, one could also consider releasing non-sensitive values (like in OSDP) randomly instead of hiding all. However, this requires addressing challenges of any inadvertent leakages through dependencies when sharing such randomized data. }

\section*{Supplementary Material}

Due to space constraints, we defer omitted proofs, algorithms, discussions, and some experimental details to the supplementary materials of this paper.

\ifCLASSOPTIONcompsoc
  \section*{Acknowledgments}
\else
  \section*{Acknowledgment}
\fi

This work was supported by NSF under Grants 2032525, 1952247, 2008993,
and 2133391, and by NSERC through a Discovery Grant. This material was based on research sponsored by DARPA
under Agreement Number FA8750-16-2-0021. The U.S. Government is
authorized to reproduce and distribute reprints for Governmental purposes
not withstanding any copyright notation there on. The views and conclusions
contained here in are those of the authors and should not be interpreted as necessarily representing the official policies or endorsements, either expressed or
implied, of DARPA or the U.S. Government.
We thank the reviewers for their detailed comments which helped to improve the paper during the revision process. 

\ifCLASSOPTIONcaptionsoff
  \newpage
\fi

\bibliographystyle{IEEEtran}  
\bibliography{references}

\vspace{-0.5in}
\begin{IEEEbiography}[{\includegraphics[width=1.1in,height=1.3 in,clip,keepaspectratio]
{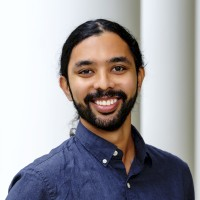}}]{Primal Pappachan} is an Assistant Professor in the Department of Computer Science at Portland State University. He received a M.S. in Computer Science from University of Maryland, Baltimore County in 2014 and a Ph.D. in Computer Science from University of California, Irvine in 2021. Afterwards, he was a postdoctoral scholar in the College of Information Sciences and Technology at Pennsylvania State University. His research interests are in the intersection of data management and privacy, particularly data protection methods such as access control, differential privacy, and privacy policies.

\end{IEEEbiography}

\vspace{-0.5in}
\begin{IEEEbiography}[{\includegraphics[width=1.1in,height=1.3 in,clip,keepaspectratio]
{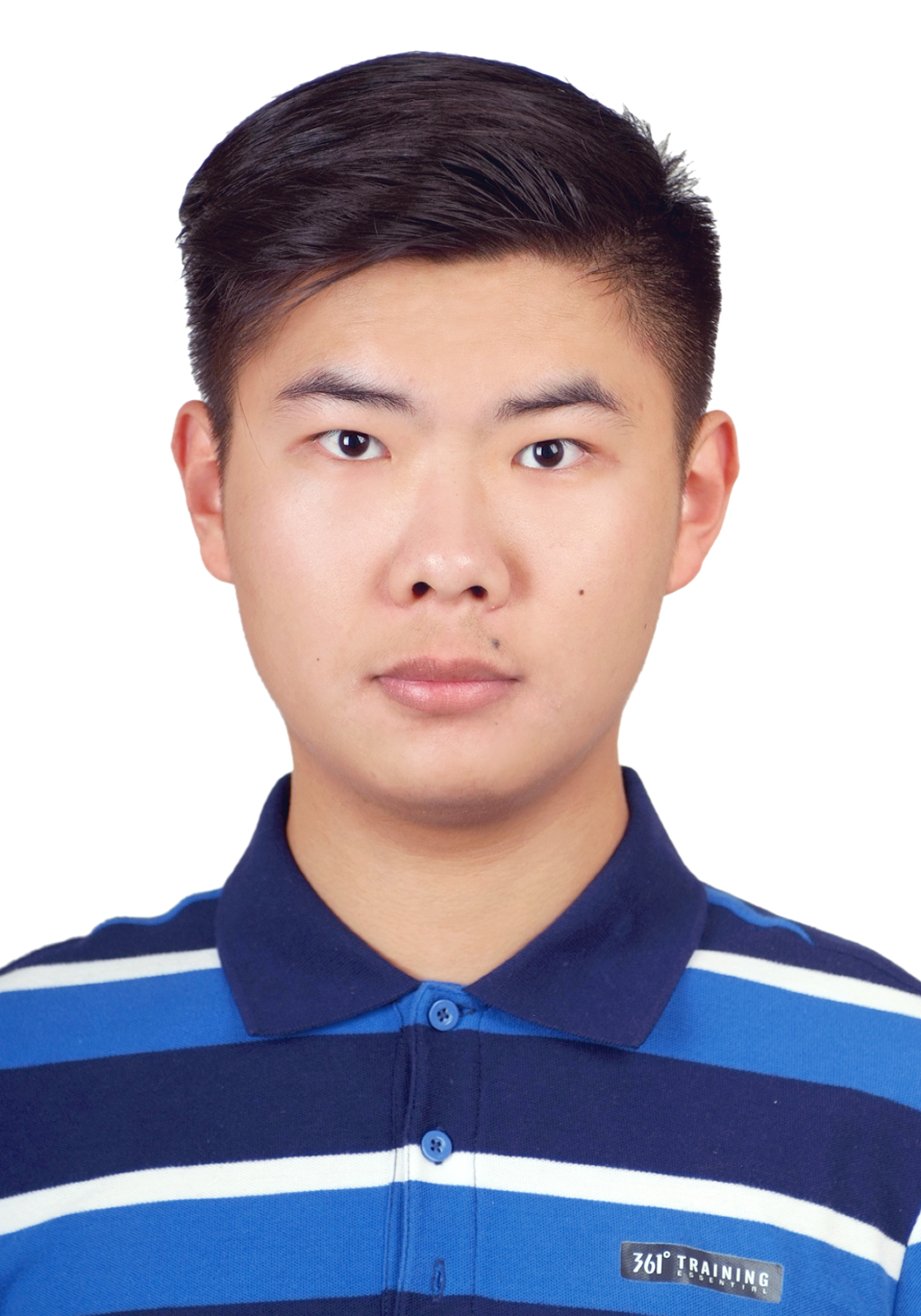}}]{Shufan Zhang} 
received the M.Math degree from the University of Waterloo, Waterloo, ON, Canada, in 2022. He is currently working toward the Ph.D. degree in computer science at the University of Waterloo. His research interests include computer security and data privacy, on both theory and system aspects, as well as their intersections with database systems and machine learning. He was selected as one of the Rising Stars in Data Science by UChicago and UCSD.

\end{IEEEbiography}

\vspace{-0.5in}
\begin{IEEEbiography}[{\includegraphics[width=1.1in,height=1.3 in,clip,keepaspectratio]
{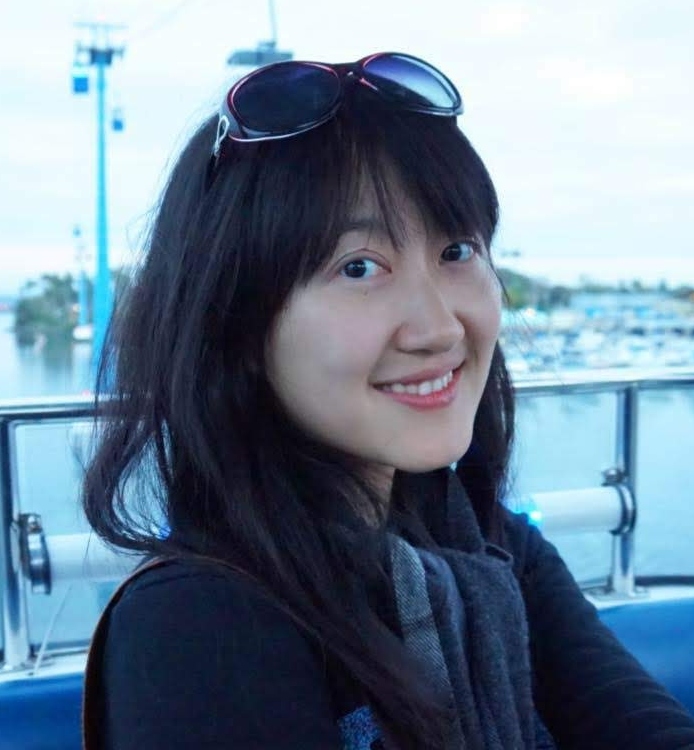}}]{Xi He} is an Assistant Professor in the Cheriton School of Computer Science at the University of Waterloo, and Canada CIFAR AI Chair at the Vector Institute. Her research focuses on the areas of privacy and security for big data, including the development of usable and trustworthy tools for data exploration and machine learning with provable security and privacy guarantees. She has given tutorials on privacy at VLDB 2016, SIGMOD 2017, and SIGMOD 2021. She is a recipient of the Meta Privacy Enhancing Technologies Research Award in 2022 and Google Ph.D. Fellowship in Privacy and Security in 2017. Her book ``Differential Privacy for Databases,'' co-authored by Joseph Near, was published in 2021. Xi graduated with a Ph.D. from the Department of Computer Science, Duke University, and a double degree in Applied Mathematics and Computer Science from the University of Singapore.
\end{IEEEbiography}

\vspace{-0.5in}
\begin{IEEEbiography}[{\includegraphics[width=1.1in,height=1.3 in,clip,keepaspectratio]
{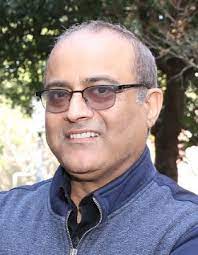}}]{Sharad Mehrotra} received the PhD degree in
computer science from the University of Texas,
Austin, Austin, Texas, in 1993. He is currently a
professor with the Department of Computer
Science, University of California, Irvine, Irvine,
California. Previously, he was a professor with the
University of Illinois at Urbana Champaign,
Champaign, Illinois. He has received numerous
awards and honors, including the 2011 SIGMOD
Best Paper Award, 2007 DASFAA Best Paper
Award, SIGMOD Test of Time Award, 2012, DASFAA ten year best paper awards for 2013 and 2014, 1998 CAREER
Award from the US National Science Foundation (NSF), and ACM ICMR
Best Paper Award for 2013. His primary research interests include the
area of database management, distributed systems, secure databases,
and Internet of Things.

\end{IEEEbiography}

\clearpage
\appendices
\extend{

\begin{table}[t]
\caption{A Summary of Notations}
\label{tab:notation}
\begin{tabular*}{\columnwidth}{@{}l|l}
\toprule 
{\em Notation} & {\em Definition}  \\
 \hline \\[-1em]
$\vDatabase$  & A database instance \\ 
$\vCell{}{}$ & A cell in a database relation \\
$\vCellSet{}, \vCellSet{}^H$ & Set of cells, hidden cells \\
$\vSchemaDep{}, \vSchemaDepSet{}$ & A schema level data dependency / set \\
$\vDataDep{}{}, \vDataDepSet{}$ & An instantiated data dependency / set \\
$\vCells{\vDataDep{}{}}$ & Cells involved in $\vDataDep{}{}$ \\
$\vPredicates{\vDataDep{}{}}$ & The set of predicates associated with a DC \\
$\vPredicates{\vDataDep{}{}, \vCell{}{}}$ & The set of predicates in $\vDataDep{}{}$ that involves the cell $c$ \\
$\vPredicates{\vDataDep{}{}\backslash c}$ & The set of predicates in $\vDataDep{}{}$ without the cell $c$ \\
$\view(\vCellSet{})$ & Set of value assignments for cells in $\vCellSet{}$\\
$\inference(\vCell{}{} \mid \view, \vDataDep{}{})$ & Inference function for the cell $\vCell{}{}$ \\
\bottomrule
\end{tabular*}
\end{table}

\section{Additional Proofs}

\extend{
\begin{proof}[Proof of Theorem 1]
Based on the definition in Section~\ref{sect:problem_definition}, we have $\inference(\vCell{*}{}|\view,\vDataDep{}{})$ $\subseteq$ $\inference(\vCell{*}{}|\view_0,\vDataDep{}{})$. 
Next, we show that for any possible value assignment to $\vCell{*}{}$ in the base view, $x^* \in \inference(\vCell{*}{}|\view_0,\vDataDep{}{})$, $x^*$ is also in $\inference(\vCell{*}{}|\view,\vDataDep{}{})$ when the TTC is \false. 
We prove this based on the two cases when $TTC(\vDataDep{}{}, \view, \vCell{*}{})$ evaluates to \textit{False} (see Equation~\ref{eq:ttc}).

\textit{Case 1:} At least one of the predicates $\vPredicate{}{i}$ in $\vPredicates{\vDataDep{}{}\backslash \vCell{*}{}}$ evaluate to \false based on the true cell value assignments in $\view$ i.e., \textit{eval}$(\vPredicate{}{i}, \view{}{})$ = \textit{False}. 
Therefore, the sensitive cell $\vCell{*}{}$ can take any value $x^* \in \inference(\vCell{*}{}|\view_0, \vDataDep{}{})$ to ensure $\vDataDep{}{}$ to be \true, i.e. $\neg(\cdots \land \false \land \cdots)=\true$ always.

\textit{Case 2:} 
At least one of the predicates $\vPredicate{}{i}$ in $\vPredicates{\vDataDep{}{}\backslash \vCell{*}{}}$ evaluate to \textit{Unknown} based on $\vCell{}{j} \in \vCells{\vPredicate{}{i}}$ being hidden in $\view$. 
Based on the assumption stated earlier, we know that there exists $x_j\in \inference(\vCell{}{j} \mid \view_0,\vDataDep{}{}) \subseteq \view(\vCell{}{j})$  that leads to $\vPredicate{}{i}$ evaluating to \false. 
Hence, for any $x^* \in \inference(\vCell{*}{}| \view_0, \vDataDep{}{})$, there exists 
$x_j \in \inference(\vCell{}{j}\mid \view)$ 
for any $\vCell{}{j} \in Cells(\vDataDep{}{})\backslash\{c^*\}$
such that 
\begin{eqnarray}
\vDataDep{}{}(\ldots,c_j=x_j, c^*=x^*,\ldots) \nonumber  \\ 
&\hspace{-3em} =  \neg(\cdots \land \vPredicate{}{i}(c_j=x_j,\ldots) \land \cdots) \nonumber \\
&\hspace{-5.5em}= \neg(\cdots \land \false \land \cdots)=\true \nonumber
\end{eqnarray}

Combining two cases proves the theorem.
\end{proof}
}

\extend{
\begin{proof}[Proof of Theorem 2]
We prove this by induction on the number of dependency instantiations $n$.
\textbf{Base case:} When $n=0$, then there are no dependency instantiations that apply to $\vCell{}{i} \in \vCellSet{}^{H}$ and hence in the shared querier $\view{}{}$, there exists no possibility of leakage for all $\vCell{}{i}$.
\textbf{Induction step:} Suppose the Theorem~\ref{theorem:fdcell} is \true for $n=k$ i.e., $\view{}{}$ achieves  full deniability when there are $k$ dependency instantiations. 
Now we consider the case when there exists $n=k+1$ dependency instantiations.
If $\vDataDep{}{k+1}$ does not include any $\vCell{}{i} \in \vCellSet{}^H$ then by default $\view{}{}$ achieves full deniability.
Suppose $\vDataDep{}{k+1}$ includes a cell $ \vCell{}{i} \in$  $\vCellSet{}^H$ and there is leakage on $\vCell{}{i}$ despite it being hidden because $TTC(\vDataDep{}{k+1}, \vCell{}{i}, \view{}{})$ is \true. In the rest of the proof we show that such a leakage is impossible.

If  $\vDataDep{}{k+1}$ only contains a single predicate of the form $\vPredicate{}{}(\vDataDep{}{k+1}) = $ $\vCell{}{i}\vOperator\vConstant$), then since
$\vCell{}{i}$ is hidden, such a dependency cannot cause leakage. 
On  the other hand, if the single predicate is of the form  $\vPredicate{}{}(\vDataDep{}{k+1}) = $
$\vCell{}{i} \vOperator \vCell{}{j} $, then $cueset(\vCell{}{i}, \vDataDep{}{})$ contains $\vCell{}{j}$.
As a result,  $\vCell{}{j}$ must also be hidden (by the property of $\vCellSet{}^H$ described in the theorem)
and thus, again such a dependency cannot lead to leakage. 
If additional predicate(s) exists in $\vDataDep{}{k+1}$ and $TTC(\vDataDep{}{k+1}, \vCell{}{i}, \view{}{})$ is \true, it must be the case that all the other predicates ($\vPredicates{\vCell{}{}(\vDataDep{}{k+1})\backslash\vCell{}{i}}$) are
\true. 
Thus, by the property of $\vCellSet{}^H$ described in the theorem, there must exist another cell $\vCell{}{j}$ from $\vPredicates{\vCell{}{}(\vDataDep{}{k+1})\backslash\vCell{}{i}}$ such that it is also hidden.
As $TTC(\vDataDep{}{k+1}, \vCell{}{i}, \view{}{})$ returned \true even though $\vCell{}{j}$ from its cueset is hidden, it must be the case, that there exists
another dependency instantiation for which 
tattle tale condition for such a $\vCell{}{j} $ holds.
But as the induction step established full deniability for all the dependency instantiations up to $k$ and therefore no leakage, it cannot exist. 
Hence, the $\view$ achieves full deniability if $\vCellSet{}^H$ satisfies the condition in the theorem.
\end{proof}
}

\section{Additional Algorithms}

\begin{figure}[t]
\centering
\includegraphics[width=\linewidth]{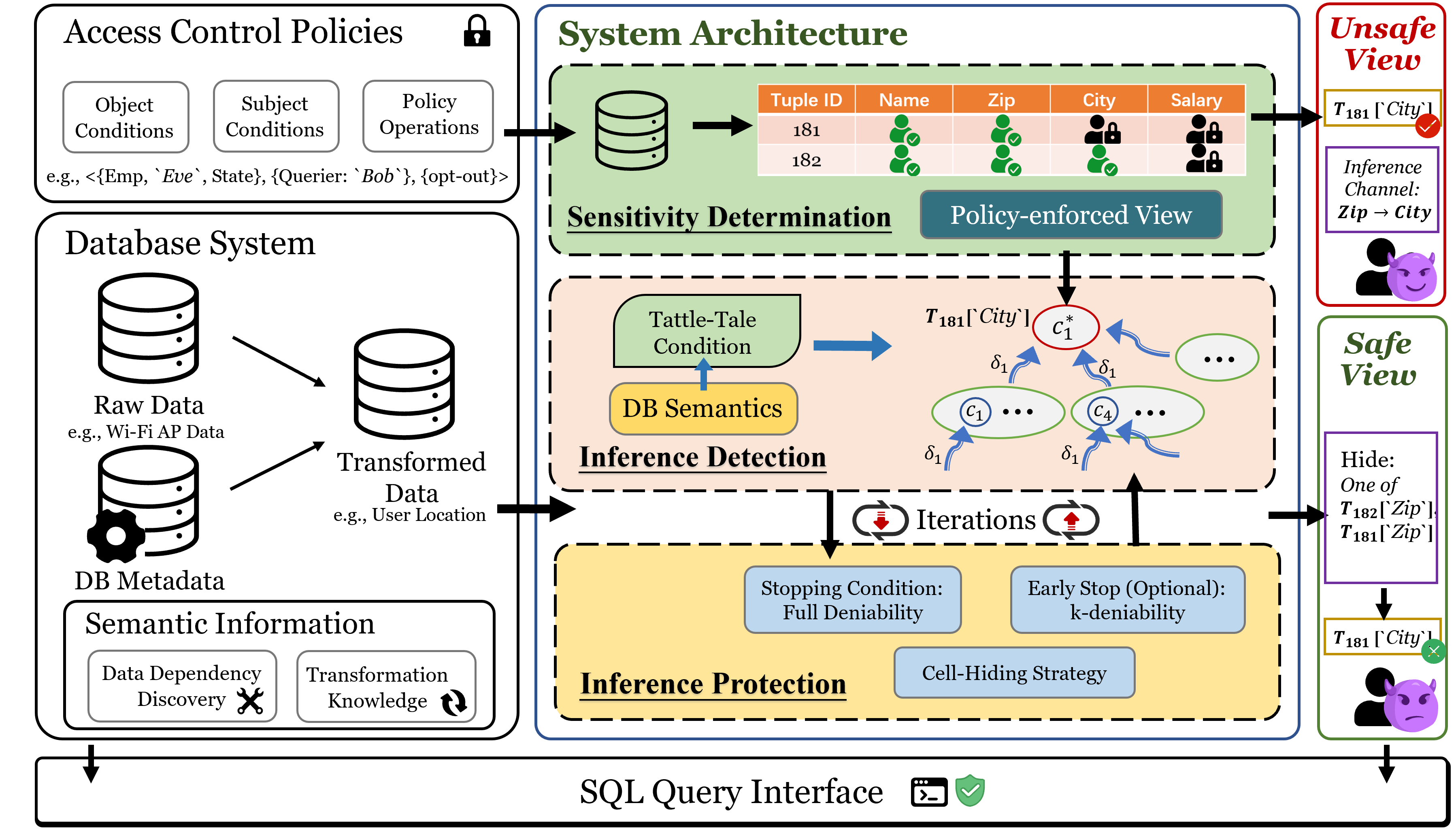}
\caption{System Architecture}
\label{fig:systemarchitecture}
\end{figure}

\stitle{Inference Protection based on Random Hiding}: Algo \ref{alg:random}.

\begin{algorithm}[t]
    \SetAlgoLined
    \DontPrintSemicolon
    \SetKwFunction{FMain}{InferenceProtect}
    \SetKwProg{Fn}{Function}{:}{}  
    \KwInput{Set of cuesets \textit{cuesets}}
    \KwOutput{A set of cells selected to be hidden \textit{toHide}}
    \Fn{\FMain{$cuesets$}}{
    
        \textit{toHide} = $\{ \}$ \Comment{Return list initialization.} \\
    
        \While{cuesets $\neq \phi$}
        {
            $cs$ = $cuesets$.getRandom() \\
            \uIf{$toHide \cap cs$}{
                     cuesets.remove(cs)\\            }
            \Else{
            toHide = toHide $\cup$ cs.getRandom() \\ 
            }
        }
        \Return{toHide}
    }
    \caption{Inference Protection (Random Hiding)}
    \label{alg:random}
\end{algorithm}

\stitle{Algorithm to compute the inferred set of values}: Algo \ref{alg:computeState}.

In Algorithm~\ref{alg:computeState}, we describe how to compute the set of inferred values for a cell based on a given database view and set of instantiated data dependencies.
For each dependency instantiation, we retrieve the predicate containing sensitive cell $\vPredicate{}{}(\vCell{*}{})$. Based on the operator $\vOperator$ \footnote{We consider ``$\leq$" and ``$<$" (similarly ``$\geq$'' and ``$>$") operators as identical in the above algorithm for simplification. 
} in this predicate, we either put the value of non-sensitive cell $\vCell{}{k}$ in the given $\view$ to \textit{minus\_set} (when $\vDomain{\vCell{*}{}}$ is discrete) or determine \textit{low} and \textit{high} (when $\vDomain{\vCell{*}{}}$) is continuous).

\begin{algorithm}[t]
    \SetAlgoLined
    \DontPrintSemicolon
    \KwInput{A target cell $\vCell{*}{}$, A database view $\view$, Set of instantiated data dependencies $\vDataDepSet{}$}
    \KwOutput{Set of inferred values for $\vCell{*}{}$ based on $\view$ and $\vDataDepSet{}$}
    \SetKwFunction{FMain}{InferredValues}
    \SetKwProg{Fn}{Function}{:}{}    
    \Fn{\FMain{$\vCell{*}{}$, $\view$, $\vDataDepSet{}$}}{
        minus\_set = \{ \}   \\
        low, high = \textbf{min}($\vDomain{\vCell{*}{}}$), \textbf{max}($\vDomain{\vCell{*}{}})$ \\
        \For{$\vDataDep{}{}$ $\in$ $\vDataDepSet{}$}
        {
            $\vCell{*}{}$, $op$, $\vCell{}{k}$ = $\vPredicate{}{}(\vCell{*}{}, \vDataDep{}{}) $  \\
            \Switch{$ op $}
            {
                \uCase{ ``$\leq$'' or ``$<$''}{  
                \uIf(\Comment{Value of $\vCell{}{k}$ in $\view$}){low $<$ $\vCell{}{k}$.val} 
                { 
                low = \textbf{min}($\vCell{}{k}$.val, high)
                }
                }
                \uCase{``$\geq$'' or ``$>$''}{ 
                \uIf{high $>$ $\vCell{}{k}$.val}
                {
                high = \textbf{max}($\vCell{}{k}$.val, low)
                }  
                }
                \uCase{``$\neq$''} 
                {
                    low = $\vCell{}{k}$.val, high = $\vCell{}{k}$.val \\
                    minus\_set= $Dom$($c^*$) $-$ $c_k$.val
                }
                \uCase{``$=$''} 
                {
                    minus\_set=minus\_set $\cup$ {$\vCell{}{k}$.val}
                }
            }
        }
        \uIf{$\vDomain{\vCell{*}{}}$ \textit{is discrete}}
        {
            \Return{$minus\_set$}
        }
        \Else (\Comment{$\vDomain{\vCell{*}{}}$ \textit{is continuous}}){
            \Return{$\vDomain{\vCell{*}{}} - [low, high]$}
        }
    }
    \caption{Computing Inferred Set of Values}
    \label{alg:computeState}
\end{algorithm}

\section{Additional Discussion}
\label{sect:attacks}

\subsection{Security Against Attacks Based on Knowledge of Algorithm}

Unlike algorithms used for achieving differential privacy (e.g., 
adding laplace noise to query output), the algorithm (denoted by FDA) used to generate a view with full deniability is
deterministic, and as a result, it could be vulnerable to 
attacks based on knowledge of the algorithm. In particular, 
an adversary, with knowledge of (a) the output of the FDA algorithm 
executed over the real dataset,  and (b) the FDA algorithm could reexecute
FDA against database instances that are consistent with the algorithm's 
output on the real data and compare the outputs to narrow down the set of
possible database instances that might correspond to the real data.  Such
an attack has been referred to as {\em reverse engineering attack} in
\cite{xiao2010transparent}. Reverse engineering attacks that eliminate 
possible database instance could, in turn, violate full deniability.

Our goal in this section is to show that reverse engineering attack does not
lead to elimination of any viable database instance in FDA. In other words,
let $\mathcal{D}_{a}$ = $\{D_1, \ldots, D_n\}$ be adversarial apriori knowledge. That is, the adversary knows that the  real database, 
denoted by $D_{real}$ corresponds to one of these database instances $D_i \in \mathcal{D}_a$ prior to observing the output of the FDA algorithm. 
For an adversary with no knowledge of the actual database, $\mathcal{D}_a$ could be the set of all possible database instances based on the schema  that are \emph{consistent} with the data dependencies. Alternatively, it could be some subset of the above possible instances.

 Now consider the adversary with the knowledge of the output of FDA on $D_{real}$ and a set of sensitive cells $\vCellSet{}^S$, i.e., 
 FDA$(D_{real}, \vCellSet{}^S)=(D^*, \vCellSet{}^H)$,   where $D^*$ is the output view with full deniability, and 
 $\vCellSet{}^H$ is the set of hidden cells in $D^*$. 
Let $\mathcal{D}_{p}\subseteq \mathcal{D}_a$ be the posterior knowledge of the adversary based on the output of FDA, $(D^*, \vCellSet{}^H)$ and the knowledge of the algorithm FDA$(\cdot,\cdot)$. We want to show that any database instance $D_i\in \mathcal{D}_a$ from the adversarial apriori knowledge is either (i) \emph{incompatible} with $D^*$, i.e.,  the values of the cells that are not hidden in $D^*$ do not match with the values of corresponding cells in $D_i$; or (ii) \emph{undeniable}, i.e., it cannot be eliminated in the posterior knowledge, i.e., $D_i\in \mathcal{D}_p$. 

For the first case, if a $D_i\in \mathcal{D}_a$ is incompatible with $D^*$, such an elimination of a database instance from the posterior knowledge does not constitute violation of full deniability since the adversary 
can determine that $D_{real} \neq D_i$ based on simply the value of the
visible cells.
For the second case,  for a $D_i\in \mathcal{D}_a$ that is compatible with $D^*$, if the adversary finds that it is impossible to generate the output $(D^*,\vCellSet{}^H)$ by running FDA on $D_i$, then there is a leakage as $D_i$ cannot be $D_{real}$ and should be eliminated from $\mathcal{D}_p$. 
In order to run FDA, the adversary needs to input a guessed database instance and a set of sensitive cells which can only be a subset of the hidden cells $\vCellSet{}^H$. To eliminate $D_i$, the adversary needs to test all possible subsets of $\vCellSet{}^H$. 
We would like to prove that such an elimination is not possible and formalize the security guarantee as follows. 
\begin{theorem}
\label{theorem:reverse_engineering_attack}
Let apriori knowledge of adversary be that
$D_{real}$ is one of the database instances in $\mathcal{D}_{a}$ = $\{D_1, \ldots, D_n\}$. Let FDA run on $D_{real}$ with a set of sensitive cells $\vCellSet{}^S$ and output ($D^*, \vCellSet{}^H$). 
For all $D_i \in \mathcal{D}_a$ that is compatible with $D^*$, there exists a subset of the hidden cells $\vCellSet{j}^H \subseteq \vCellSet{}^H$  such that running FDA($D_i, \vCellSet{j}^H$) returns $(D^*, \vCellSet{}^H)$.
\end{theorem}
\begin{proof} (Sketch)
For any $D_i \in \mathcal{D}_a$ that is compatible with $D^*$, we can show that 
FDA($D_i,\vCellSet{}^H$)=$(D^*, \vCellSet{}^H)$. In particular, $D_i$ that starts with $\vCellSet{}^H$ as the sensitive cells and hence all hidden already achieves full deniability. FDA does not need to hide more cells, and thus the output is $(D^*, \vCellSet{}^H)$.
\end{proof}

\begin{example}
Consider running FDA on a real database $D_{real} = \{ \vCell{}{1}, \dots, \vCell{}{10}\}$ and a set of sensitive cells $\vCellSet{}^S$,
and it results in a view $D^*$ with
hiding cells $\vCellSet{}^{H} = \{ \vCell{}{1}, \vCell{}{3}, \vCell{}{4}, \vCell{}{7} \}$. 
The adversary may try to test
 a database $D_i = \{ \vCell{\prime}{1}, \dots, \vCell{\prime}{3}, \dots, \vCell{\prime}{4}, \dots, \vCell{\prime}{7} \} \in \mathcal{D}_a$ that  is compatible  with $D^*$,  
where $\vCell{\prime}{1}, \vCell{\prime}{3}, \dots $ are values guessed by the adversary for the hidden cells in $D^*$. 
The adversary only knows that $\vCellSet{}^S\subseteq \vCellSet{}^H$. 
Hence, the adversary can guess a subset $\vCellSet{j}^H\subseteq \vCellSet{}^H$, e.g., 
$\vCellSet{1}^{H} = \{ \vCell{}{1}, \vCell{}{3} \}$,
as a possible sensitive cell set input and run FDA over $D_i$ and $\vCellSet{j}^{H}$. When $\vCellSet{j}^H=\vCellSet{}^H$, FDA($D_i,\vCellSet{j}^H$) outputs $(D^*,\vCellSet{}^H)$. Hence $D_i$ cannot be eliminated from $\mathcal{D}_p$.
\end{example}

\subsection{Extending FC Constraints}
\label{subsec:fc_constraints}
We now describe a general model of Function-based Constraints  by extending the model to invertibility as follows for a a function $fn(r_1, r_2, \ldots, r_n) = s_i$ where $r_1, r_2, \ldots, r_n$ are the general representation for input values (e.g., WorkHrs) and $s_i$ is the general representation of the derived value or the output of the function (e.g., Salary). 

\begin{definition}[$(m, n)$-Invertibility]
For a function of the following form $fn(r_1, r_2, \ldots,$ $ r_p) = s_i$, given its output $s_i$ and any $m - 1$ out of $p$ inputs, if we could find another function $fn^\prime (r_t, r_{t+1}, \ldots, r_{t+m-2}; s_i)=$ $\{r_k, r_{k+1}, \ldots$, $r_{k+n-1} \}$ that disclose $n$ of the rest input values, we say this function $fn$ is $(m, n)$-invertible; otherwise, we say this function is $(m, n)$-non-invertible.
\end{definition}

The previously mentioned Salary function is (2, 1)-invertible as given any two of the three variables, the rest one could be disclosed. 

\begin{definition}[Fully invertible]
If a function \textit{f} is \textit{(1, n)}-invertible, we say this function is fully invertible.
\end{definition}

Cross product (Cartesian product) is an example of full invertibility, since all the input values can be inferred if given the result of cross product. That is to say, cross product is $(1, n)$-invertible. 
Other examples of commonly used functions are user-defined functions (UDFs) (e.g. oblivious functions, secret sharing), and aggregation functions.

\eat{
\begin{theorem}
Any $(m, n)$-invertible function is $(m-1, n)$-non-invertible.
\end{theorem}
}

\stitle{Computing Leakage for FCs.}
As for a $(m, n)$-invertible function, denoted by $fn^\prime (r_1, r_2, \ldots r_n)$ $= \{s_1, s_2,$ $\ldots, s_m \}$, it can be apparently observed that, given the $m$ inputs, the function will lead to the leakage towards $n$ values.
Take, for example, the \textit{Salary}, \textit{WorkHrs} and the \textit{SalPerHr}.
Since as analyzed, the function to calculate the salary is $(2, 1)$-invertible, it means that if taking any two values of the three attributes, the adversary can fully convert and leak the exact value of the remaining attribute.
We call this case \textit{full leakage} from FC.
However, a subset of the $m$ input values of an $(m, n)$-invertible function could also leak some information about some disclosable values based on some domain knowledge.
As an example, suppose the adversary knows that the \textit{salary} of an employee is 8,000 but they do not know the exact \textit{WorkHrs} and \textit{SalPerHr}.
Even though, with some background knowledge, for e.g., the information that no one could work more than 40 hours per week by law, the adversary could reduce the domain of possibilities the \textit{SalPerHr} value could take. We call this \textit{partial leakage} from FC. We leave further exploration of this extended model as future work.

\subsection{Adversary is A Data Owner}
\label{app:adv_data_owner}

We stated in Section \ref{sect:assumptions} that queriers and data owners are non-overlapping parties.
If we want to relax this assumption to consider an adversary that is a querier and data owner, we can modify our Inference Detection algorithm to only include cells that do not belong to the querier in the cueset.

The proof sketch for correctness of this modified algorithm is as follows.
\reviseTKDE{Theorem~\ref{theorem:fdcell} in the paper (Full Deniability for a Querier View) states that $\view$ achieves full deniability if for any cell $\vCell{}{i}$ in the set of hidden cells, there exists another cell $\vCell{}{j}$ from $\vCell{}{i}$'s cueset that is also in the set of hidden cells.}
In the modified algorithm, we have only updated Inference Detection (which generates cuesets) and not Inference Protection (which hides cells from the cuesets).
Each of these modified cuesets contain at least one cell as the dependencies are binary which leads to Inference Protection successfully hiding a cell from these cuesets in the next step.
Therefore, based on Theorem~\ref{theorem:fdcell} it satisfies the necessary condition for achieving full deniability which is to have at least one cell from each cueset be present in the set of hidden cells.
If in a dependency instantiation, both tuples belong to the querier then it is possible its corresponding cueset is empty using this modified algorithm.
However in such a case, the sensitive cell in the dependency instantiation already belongs to the querier so it is not possible to prevent them from learning about it.

\section{Additional Related Work}

We demonstrate the difference between our full deniability model and query-view security model proposed by \cite{miklau2007formal} more concretely using 
simple examples. 
Consider a relation schema
$R$ with three attributes $A, B, C$ and a functional dependency 
$A \rightarrow B$, and a secret query $S$ that
projects the $B$ values of tuples which have $C = 5$. Also, consider
a view  $V$ that 
projects attribute  $A,B$. 
This view and secret query can be expressed in the data log notation used by Miklau \& Suciu as follows: 
\begin{eqnarray}
V(A, B):& R(A, B, -) & \nonumber\\
S(B):& R(-, B, c), c = 5\nonumber
\end{eqnarray}

Based on the secret query  and
the above view definitions, 
 Miklau \& Suciu \cite{miklau2007formal} will determine that $V$ violates  perfect secrecy since there   exists
a database for which the view may reveal sensitive data. 
As an example, 
 consider the instance of the table 
 $R(A,B,C) = \{<1,1,5>, <1,1,6>\}$.
 Note that the value of $B$ for the
 first tuple is sensitive since 
 the corresponding value of $C$ is 5.  Given the functional dependency,  the view $V(A,B)$ will
 indeed leak the sensitive data.
 
In contrast to \cite{miklau2007formal} which is
motivated by determining if a specific view definition could
lead to  leakage of sensitive data in a data exchange scenario
(i.e., if there can exist a database such that the published view
may leak information about the secret in that database)
our paper is motivated by access control.  Our 
goal is, for a given instance of the database, answer the query, hiding as few cells  as possible while ensuring
full deniability for sensitive cells (i.e., the adversary cannot eliminate
 any possible value from  the domain of the sensitive
 cell). In the example above, 
 we will allow a view $V$ to be  computed with
 some cells hidden to ensure full deniability. In particular,
 a possible answer could be  $V(A,B): t_1: \{<1,NULL>, t_2: <NULL,1>\}$ since it allows full deniability of the 
 sensitive cell $t_1[B]$.  As another example, 
 consider a different instance of $D$  with 
$R(A,B,C) = \{<1,1,5>, <2,1,6>\}$. 
For the above database, given
a query that projects attributes
$A,B$, we will return the answer
$V(A,B): \{t_1: <1,NULL>, t_2: <2,1>\}$  since
for the above instance, the 
result of the query does not reveal
any information about the sensitive cell $t_1[B]$.
Thus, the work by Miklau \& Suciu \cite{miklau2007formal}, as mentioned above, does not address
access control but determining if view definitions violate
perfect secrecy. As a result, irrespective of the database instance, \cite{miklau2007formal} will consider the above view definition to be unsafe.

We note that we could implement access control using  the framework developed in their work. In particular, given a query $V$ 
and
sensitive cells (expressed as a sensitive query), and constraints
encoded as prior knowledge of the adversary, we could check
if  $V$ violates the secrecy of  $S$. If $V$ does violate perfect secrecy, 
access control can be implemented by preventing $V$ to execute.
However, to the best of our knowledge, this is not the
intended use case for \cite{miklau2007formal} since the
resulting mechanism would be too pessimistic for
it to be useful. It would disallow a view (query) for which there exists a database instance that
could result in leakage.  We thus believe that
\cite{miklau2007formal} as described in their work is unsuitable
for access control and hence, consider their work to be 
addressing a different, though loosely related issue. The
follow-up work to \cite{miklau2007formal} \cite{dalvi2005asymptotic} relaxed the notion of perfect secrecy to make it more practical. The original
definition of perfect secrecy disallowed any leakage, while
the new definition allows for bounded leakages. Furthermore, checking perfect secrecy is $\pi^P_2$-complete  even for simple databases which makes it computationally intractable. However, none of the extensions addressed the challenge of access control i.e., suppress as few cells as possible while
answering a query $Q$ given a database $D$..  As such, extensions, along with work by \cite{miklau2007formal} 
is best suited for determining the safety of data publishing and not access control. 
}

\section{Additional Experimental Details}
\label{app:app_details}

\subsection{Information of the Datasets for Evaluation}

Some statistics of the datasets are summarized in Table \ref{tab:datasets}.
The data dependencies used for experiments can be found in Table \ref{tab:DependencyListTax} (for the Tax dataset) and Table \ref{tab:DependencyListHos} (for the Hospital10K dataset).
In addition, we outline the schema information of the datasets below.

\commentrequired{
\noindent
\textbf{Schema Information.}
Every tuple (\textit{T\_ID}) from the Tax table specifies tax information of an individual with their first name (\textit{FName}), last name (\textit{LName}), gender (\textit{Gender}), area code for phone number (\textit{AreaCode}), phone number (\textit{Phone}), city (\textit{City}), state of residence (\textit{State}), zip (\textit{Zip}), marital status (\textit{MaritalStatus}),  Has Children (\textit{HasChild}), salary earned (\textit{Salary}), tax rate (\textit{Rate}), Single Exemption rate (\textit{SingleExemp}), Married Exemption rate (\textit{MarriedExemp}), and Child Exemption rate (\textit{ChildExemp}).
As for the Hospital table, each tuple has the information of hospitals, which includes the provider number (\textit{ProviderNumber}), the hospital name (\textit{HospitalName}), city (\textit{City}), state of the hospital (\textit{State}), zip (\textit{ZIPCode}), county name (\textit{CountyName}), phone number for contact (\textit{PhoneNumber}), the type of the hospital (\textit{HospitalType}), the owner of the hospital (\textit{HospitalOwner}), emergency service (\textit{EmergencyService}), condition (\textit{Condition}), measure code (\textit{MeasureCode}), measure name (\textit{MeasureName}), the number of patient samples (\textit{Sample}), and the state average (\textit{StateAvg}).
}

\begin{table*}[t]
    \centering
    \caption{\extend{Statistics of the Datasets for Evaluation}}
    \label{tab:datasets}
    \vspace{-0.4cm}
    \begin{tabular}{|c|c|c|c|c|c|}
    \hline
        Dataset & \# Tuples & \# Attributes & \# Discrete attributes & Domain size & \# Dependencies \\ \hline
        Tax \cite{bohannon2007conditional} & 9,998 & 15 + 1 & 10 & $\approx$ $2^{107}$ ($2^{82}$ active) & 10 DCs + 1 FC \\ \hline
        Hospital10K \cite{xuchuDC} & 10,000 & 15 & 15 & $\approx$ $2^{115}$ ($2^{104}$ active) &  14 DCs \\ \hline
        Hospital \cite{xuchuDC} & 100,000 & 15 & 15 & $\approx$ $2^{115}$ ($2^{104}$ active) &  14 DCs \\ \hline
    \end{tabular}
\end{table*}

\begin{table*}[t]
  \centering
  \caption{\commentrequired{Dependency List for Tax Dataset}}
  \label{tab:DependencyListTax}
  \scriptsize
  \begin{tabular}{|l|}
  \hline
    \eat{ {\footnotesize $\delta_{1}^{t}$} &  DC &} {\scriptsize $\neg$($t_1$[zip]=$t_2$[zip] $\land$ $t_1$[city]$\neq t_2$[city]) } \\
    \eat{{\footnotesize $\delta_{2}^{t}$} &  DC &} {\scriptsize $\neg$($t_1$[areaCode]=$t_2$[areaCode] $\land$ $t_1$[state] $\neq t_2$[state] ) }  \\
    \eat{{\footnotesize $\delta_{3}^{t}$} &  DC &} {\scriptsize $\neg$($t_1$[zip]=$t_2$[zip] $\land$ $t_1$[state] $\neq t_2$[state]) } \\ 
    \eat{{\footnotesize $\delta_{4}^{t}$} &  DC &} {\scriptsize $\neg$($t_1$[state]$ \neq t_2$[state] $\land$ $t_1$[hasChild]=$t_2$[hasChild] $\land$ $t_1$[childExemp]$ \neq t_2$[childExemp])) }\\ 
    \eat{{\footnotesize $\delta_{5}^{t}$} &  DC &} {\scriptsize $\neg$($t_1$[state]$ \neq t_2$[state] $\land$ $t_1$[marital]=$t_2$[marital] $\land$ $t_1$[singleExemp]$ \neq t_2$[singleExemp]) } \\ 
    \eat{{\footnotesize $\delta_{6}^{t}$} &  DC &} {\scriptsize $\neg$($t_1$[state]$ \neq t_2$[state] $\land$ $t_1$[salary]$>$$t_2$[salary] $\land$ $t_1$[rate]$<$$t_2$[rate]) } \\ 
    \eat{{\footnotesize $\delta_{7}^{t}$} &  DC &} {\scriptsize $\neg$($t_1$[areaCode]$ \neq t_2$[areaCode] $\land$ $t_1$[zip]=$t_2$[zip] $\land$ $t_1$[hasChild]=$t_2$[hasChild] $\land$ $t_1$[salary]$>t_2$[salary] $\land$ $t_1$[rate]$<t_2$[rate] $\land$ $t_1$[singleExemp]$ \neq t_2$[singleExemp]) }  \\ 
    \eat{{\footnotesize $\delta_{8}^{t}$} &  DC &} {\scriptsize $\neg$($t_1$[marital]$  \neq t_2$[marital] $\land$ $t_1$[salary]$ \neq t_2$[salary] $\land$ $t_1$[rate]=$t_2$[rate] $\land$ $t_1$[singleExemp]=$t_2$[singleExemp] $\land$ $t_1$[childExemp] $ \neq t_2$[childExemp] ) } \\ 
    \eat{{\footnotesize $\delta_{9}^{t}$} &  DC &} {\scriptsize $\neg$($t_1$[state]$ \neq t_2$[state] $\land$ $t_1$[marital]$ \neq t_2$[marital] $\land$ $t_1$[rate]=$t_2$[rate] $\land$ $t_1$[singleExemp]=$t_2$[singleExemp] $\land$ $t_1$[childExemp]$ \neq t_2$[childExemp]) } \\ 
    \eat{{\footnotesize $\delta_{10}^{t}$} &  DC &} {\scriptsize $\neg$($t_1$[state]=$t_2$[state] $\land$ $t_1$[salary]=$t_2$[salary] $\land$ $t_1$[rate]$ \neq t_2$[rate]) } \\ 
    \eat{{\footnotesize $\delta_{11}^{t}$} &  DC &} {\scriptsize ``tax'' = fn(``salary'', ``rate'') } \\ 
    \hline
  \end{tabular}

\end{table*}

\begin{table*}[t]
  \centering
  \caption{\commentrequired{Dependency List for Hospital Dataset}}
  \label{tab:DependencyListHos}
  \scriptsize
  \begin{tabular}{|l|}
  \hline
    \eat{{\footnotesize $\delta_{1}^{t}$} & DC &} {\scriptsize $\neg$($t_1$[Condition]=$t_2$[Condition] $\land$ $t_1$[MeasureName]=$t_2$[MeasureName] $\land$ $t_1$[HospitalType]$\neq t_2$[HospitalType]) } \\
    \eat{{\footnotesize $\delta_{2}^{t}$} & DC &} {\scriptsize $\neg$($t_1$[HospitalName]=$t_2$[HospitalName] $\land$ $t_1$[ZIPCode] $\neq t_2$[ZIPCode] ) }  \\
    \eat{{\footnotesize $\delta_{3}^{t}$} & DC &} {\scriptsize $\neg$($t_1$[HospitalName]=$t_2$[HospitalName] $\land$ $t_1$[PhoneNumber] $\neq t_2$[PhoneNumber]) } \\ 
    
    \eat{{\footnotesize $\delta_{4}^{t}$} & DC &} {\scriptsize $\neg$($t_1$[MeasureCode]=$t_2$[MeasureCode] $\land$ $t_1$[MeasureName] $\neq t_2$[MeasureName] ) }  \\
    \eat{{\footnotesize $\delta_{5}^{t}$} & DC &} {\scriptsize $\neg$($t_1$[MeasureCode]=$t_2$[MeasureCode] $\land$ $t_1$[StateAvg] $\neq t_2$[StateAvg] ) }  \\
    \eat{{\footnotesize $\delta_{6}^{t}$} & DC &} {\scriptsize $\neg$($t_1$[MeasureCode]=$t_2$[MeasureCode] $\land$ $t_1$[Condition] $\neq t_2$[Condition] ) }  \\
    \eat{{\footnotesize $\delta_{7}^{t}$} & DC &} {\scriptsize $\neg$($t_1$[HospitalName]=$t_2$[HospitalName] $\land$ $t_1$[HospitalOwner] $\neq t_2$[HospitalOwner] ) }  \\
    \eat{{\footnotesize $\delta_{8}^{t}$} & DC &} {\scriptsize $\neg$($t_1$[HospitalName]=$t_2$[HospitalName] $\land$ $t_1$[ProviderNumber] $\neq t_2$[ProviderNumber] ) }  \\
    \eat{{\footnotesize $\delta_{9}^{t}$} & DC &} {\scriptsize $\neg$($t_1$[ProviderNumber]=$t_2$[ProviderNumber] $\land$ $t_1$[HospitalName] $\neq t_2$[HospitalName] ) }  \\
    
    \eat{{\footnotesize $\delta_{10}^{t}$} & DC &} {\scriptsize $\neg$($t_1$[City]=$t_2$[City] $\land$ $t_1$[CountyName] $\neq t_2$[CountyName] ) }  \\
    \eat{{\footnotesize $\delta_{11}^{t}$} & DC &} {\scriptsize $\neg$($t_1$[ZIPCode]=$t_2$[ZIPCode] $\land$ $t_1$[EmergencyService] $\neq t_2$[EmergencyService] ) }  \\
    \eat{{\footnotesize $\delta_{12}^{t}$} & DC &} {\scriptsize $\neg$($t_1$[HospitalName]=$t_2$[HospitalName] $\land$ $t_1$[City] $\neq t_2$[City] ) }  \\
    \eat{{\footnotesize $\delta_{13}^{t}$} & DC &} {\scriptsize $\neg$($t_1$[MeasureName]=$t_2$[MeasureName] $\land$ $t_1$[MeasureCode] $\neq t_2$[MeasureCode] ) }  \\
    \eat{{\footnotesize $\delta_{14}^{t}$} & DC &} {\scriptsize $\neg$($t_1$[HospitalName]=$t_2$[HospitalName] $\land$ $t_1$[PhoneNumber]=$t_2$[PhoneNumber] $\land$ $t_1$[HospitalOwner]=$t_2$[HospitalOwner] $\land$ $t_1$[State] $\neq t_2$[State] ) }  \\
    \hline
  \end{tabular}

\end{table*}

\reviseTKDE{
\subsection{Experiments with Highly Sensitive Databases}
\label{app:more_sensitive_exp}

\textit{First}, we extend the experiments to the settings when more cells are specified as sensitive via access control policies. The goal of this experiment is to simulate the cases where the shared database view is highly sensitive and contains a large number of sensitive cells.
In earlier presented experiments, we select 10 sensitive cells and gradually increase the total number of sensitive cells to 100 (step=10) and test how many additional cells are hidden in the inference control views.
In this experiment, we start with 100 sensitive cells and increase it up to 1000 with 10 different experiments (step=100).
We compare the number of hidden cells in the inference control view and the access control view and plot the results in Figure \ref{fig:more_sen_utility}(a).
Our inference control approach hides 1.5-2x cells compared to the access control view. The growth in the number of hidden cells slows down with the increasing number of sensitive cells.

In the second experiment (see Figure \ref{fig:more_sen_utility}(b)), we start with 100 sensitive cells which are all selected from the same attribute (\emph{per column}). Similar to previous experiment, we perform 10 experiments (i.e., step=100) increasing number of sensitive cells up to 1000.
We notice that only a few additional cells need to be hidden to achieve full deniability of the shared view.
This is because, the cuesets of sensitive cells are more likely to contain other sensitive cells, as they are all selected from the same attribute. Therefore, it is not required to hide additional cells to prevent inferences.
}
\begin{figure*}
    \centering
    \includegraphics[width=0.3\textwidth]{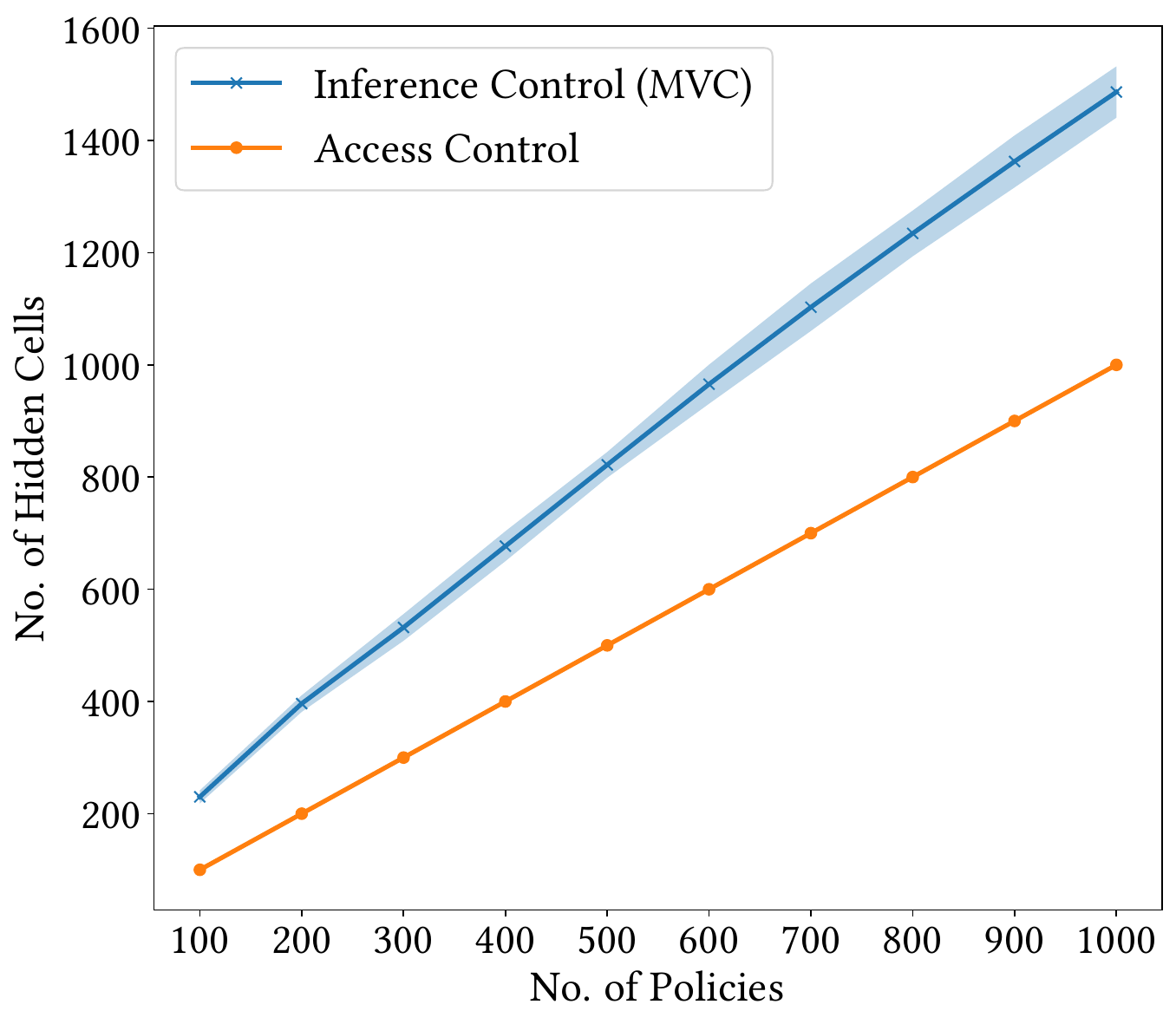}
    \includegraphics[width=0.3\textwidth]{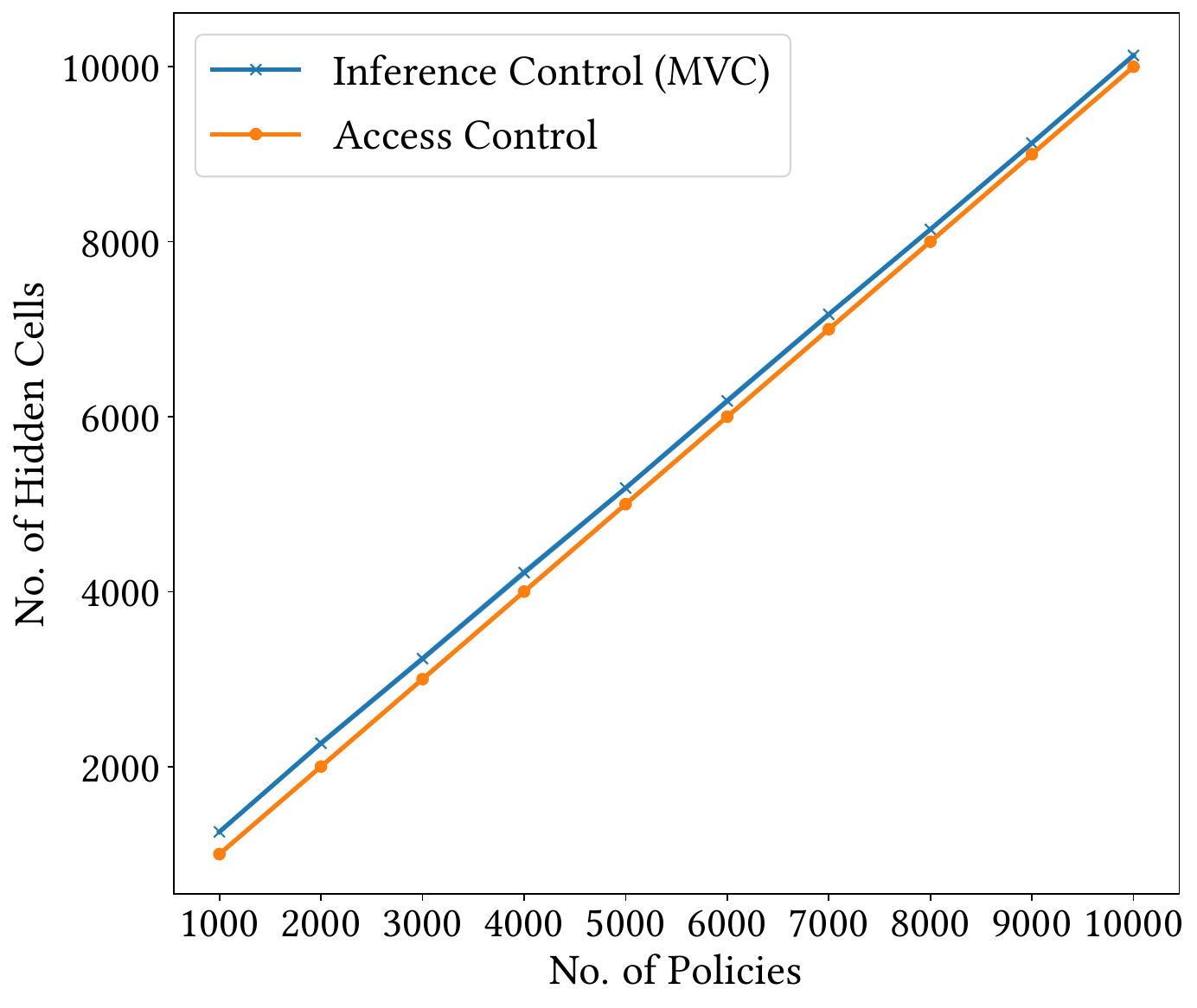}
    \caption{The no. of hidden cells vs. the no. of sensitive cells (policies): (a) w. 100-1000 cells specified as sensitive; (b) w. 100-1000 cells specified as sensitive \emph{from the same attribute}.}
    \label{fig:more_sen_utility}
\end{figure*}

\reviseTKDE{
\subsection{Experiments with Sensitive Cells Selected from Attributes with Diverse Distributions}
\label{app:exp_std}

In Figure \ref{fig:sensitivity_group_tax}, we present the results of executing the inference control algorithm over the Tax dataset and the Hospital10K dataset.
The standard deviation of the number of hidden cells is different for these two datasets. This can be explained based on the properties of the attribute from which the sensitive cells are chosen such as the domain size of the attribute, and the distribution of values from the domain in a given database instance. 
Moreover, it is more likely for the values in the domain of an attribute with a smaller domain size to be non-uniformly repeated in a database instance than an attribute with larger domain size\footnote{
Note that the relation between domain size of the attributes and distribution of values may not always be 1-to-1 as presented in this experiment.
A comprehensive study of domain size and value distribution on the impact of hidden cells is interesting but goes beyond the scope of this paper. This would require a complicated ablation study with a synthetic dataset with diverse domain sizes, distributions, and fabricated dependencies. Our goal in this paper is to study the impact of realistic dependencies as inference channels. 
}.
In the experiment measuring impact of dependency connectivity on hidden cells (Figure \ref{fig:sensitivity_group_tax}), we selected sensitive cells by randomly sampling over a group of attributes in both datasets. The Tax dataset contained predominantly large domain attributes with uniform distribution of values from the domain (such as Salary) whereas the Hospital dataset contained mostly small domain attributes with non-uniform distribution of values (such as CountryName, HospitalType).
Thus, when selecting a sensitive cell from an attribute with a smaller domain size (e.g., \textit{HospitalType} from the Hospital dataset), the number of relevant dependency instantiations and therefore number of cuesets are going to be non-uniform, resulting in a larger standard deviation across different experiments.

To validate this hypothesis, we design an experiment that only selects sensitive cells from larger domain attributes and small domain attributes on the Tax and Hospital datasets.
For the Tax dataset, the large and small domain attributes were (\textit{SingleExemp}), and (\textit{HasChild}, \textit{MaritalStatus}) respectively.
For the Hospital dataset, the large and small domain attributes were (\textit{PhoneNumber}, \textit{HospitalOwner}), and (\textit{CountyName}) respectively.
As shown in Figure \ref{fig:std_explain} (a), and (b), the standard deviation with respect to number of hidden cells is much higher when selecting sensitive cells from smaller domain size attributes with non-uniform distributions, in both datasets.
}

\begin{figure*}[t]
    \centering
    \includegraphics[width=0.3\textwidth]{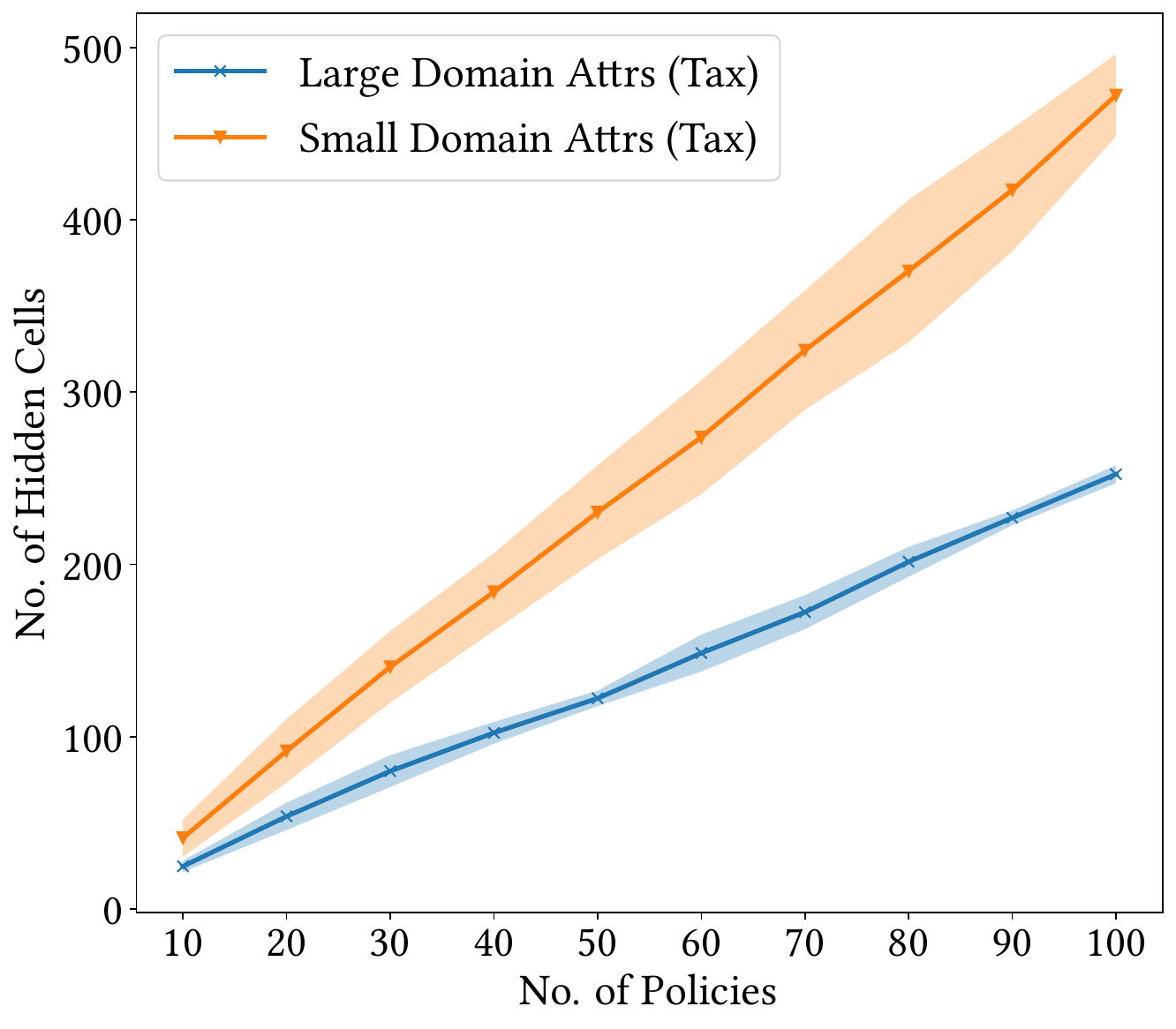}
    \includegraphics[width=0.3\textwidth]{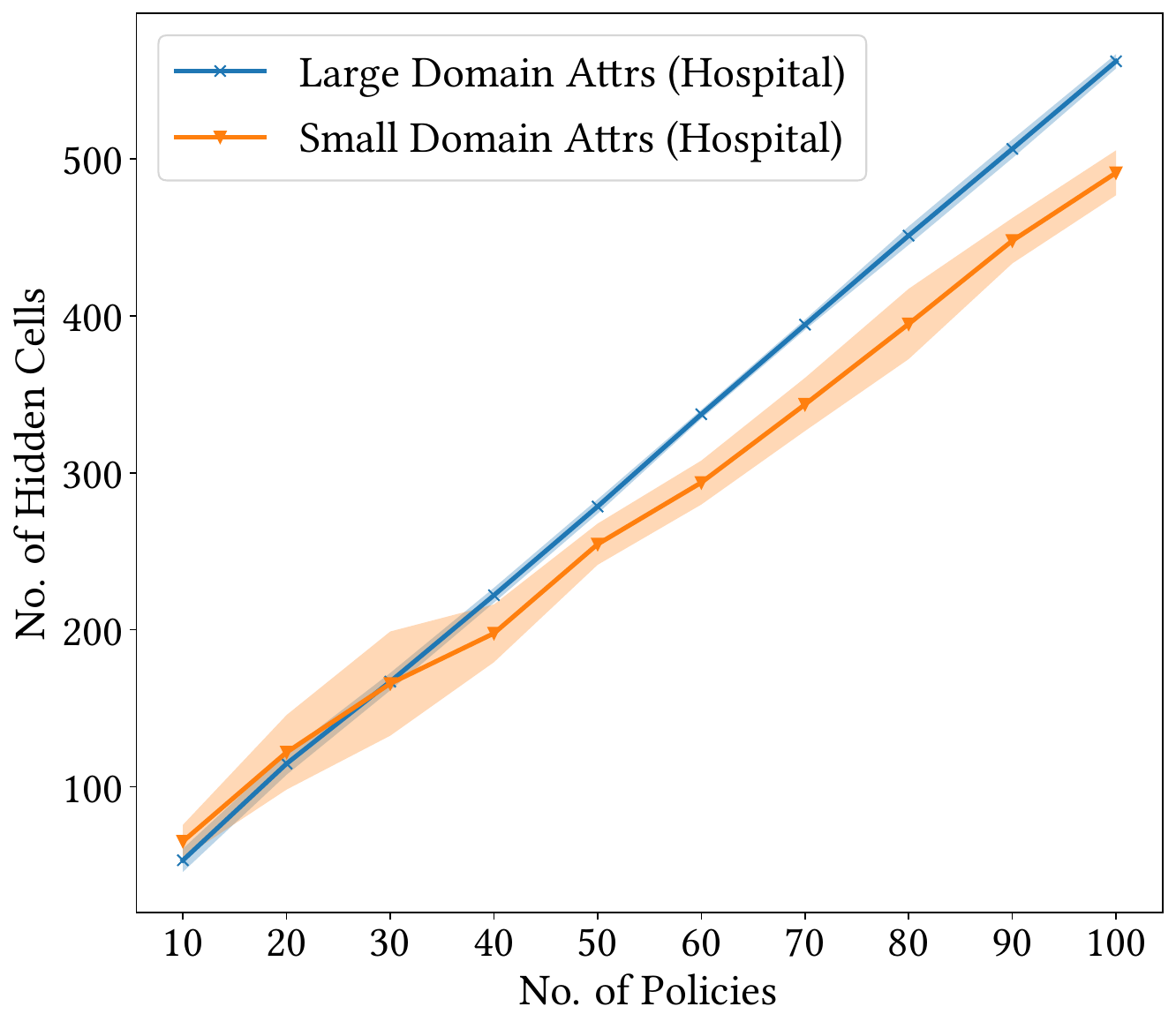}
    \caption{Setting Access Control Policies on Specific Attributes: (a) on Tax dataset; (b) on Hospital dataset.}
    \label{fig:std_explain}
\end{figure*}

\end{document}